\documentclass[11pt,numsec,webpdf,modern,medium,namedate]{oup-authoring-template}

\onecolumn

\graphicspath{{Fig/}}

\theoremstyle{thmstyleone}%
\newtheorem{theorem}{Theorem}
\newtheorem{lemma}{Lemma}
\newtheorem{corollary}{Corollary}

\theoremstyle{thmstyletwo}%
\newtheorem{remark}{Remark}%

\newtheorem{condition}{Condition}
\newtheorem{assumption}{Assumption}

\RequirePackage{amsthm,amsmath,amsfonts,amssymb}

\usepackage[doublespacing]{setspace}
\usepackage[fontsize=11pt]{fontsize}

\providecommand{\appendix}{\par\setcounter{section}{0}\renewcommand{\thesection}{A\arabic{section}}}

\usepackage{siunitx}
\RequirePackage{float}
\RequirePackage[caption = false]{subfig}
 
\usepackage{booktabs} 
\usepackage{siunitx}

\usepackage{tabularx}
\usepackage{makecell}
\usepackage{xcolor}
\usepackage{comment}

\usepackage{graphicx}

\newcommand{\intd}{\textrm{d}}

\newcommand{\simiid}{\overset{\text{iid}}{\sim}}
\newcommand{\cal}{\mathcal}

\DeclareMathOperator*{\argmin}{arg\,min}

\begin{document}

\journaltitle{}
\DOI{DOI HERE}
\copyrightyear{XXXX}
\pubyear{XXXX}
\access{Advance Access Publication Date: Day Month Year}
\appnotes{Original article}

\firstpage{1}


\title[Sphere-on-Sphere Regression]{Bayesian Sphere-on-Sphere Regression with Optimal Transport Maps}

\author[1,$\ast$]{Tin Lok James Ng}
\author[2]{Kwok-Kun Kwong}
\author[3]{Jiakun Liu}
\author[4]{Andrew Zammit-Mangion}

\authormark{Ng et al.}

\address[1]{\orgdiv{School of Computer Science and Statistics}, \orgname{Trinity College Dublin} 
}
\address[2]{\orgdiv{School of Mathematics and Applied Statistics}, \orgname{University of Wollongong, Australia} 
}
\address[3]{\orgdiv{School of Mathematics and Statistics}, \orgname{University of Sydney, Australia} 
}
\address[4]{\orgdiv{School of Mathematics and Statistics}, \orgname{University of New South Wales, Australia} 
}

\corresp[$\ast$]{ \href{Email:ngja@tcd.ie}{ngja@tcd.ie}}


\abstract{
 Spherical regression, in which both covariates and responses lie on the sphere, arises in many scientific applications and has attracted considerable methodological attention in recent years. Despite this progress, constructing flexible and expressive regression models between spherical domains remains challenging, particularly because a single global mapping is often insufficient to capture complex relationships across the entire sphere. A natural strategy is therefore to partition the spherical domain and allow distinct mappings within each region, though this introduces the additional challenge of modeling the partition structure itself. To address these issues, we propose an approach based on optimal transport to model spherical partitions, combined with parametric mappings defined locally within each region. We adopt a Bayesian framework to jointly model both the partitioning and the associated regression maps. This framework enables the identification of heterogeneous regions on the sphere while providing principled uncertainty quantification. Through real-data applications, we demonstrate that the proposed method achieves strong predictive performance, yields meaningful uncertainty estimates, and reveals interpretable clustering structure in spherical data.
 }

\maketitle

\section{Introduction}
Sphere-on-sphere regression aims to model the relationship between a covariate and a response variable on the \( p \)-sphere, defined as \( \mathbb{S}^p := \{x \in \mathbb{R}^{p+1} : \|x\|^2 = 1\} \). Early research in this area represented the mean direction of the response variable as a rotation of the covariate \citep{Mackenzie1957, Stephens1979}. With this approach to regression, the problem reduces to estimating an unknown rotation matrix $R \in \mathbb{SO}(p+1)$ where $\mathbb{SO}(p+1)$ is the space of all orthogonal matrices with determinant equal to $+1$.  \cite{Chang1986} developed asymptotic tests and confidence regions for the rotation matrix $R$ under the assumption of rotationally symmetric errors. \cite{Rivest1989} focused on von Mises-Fisher errors, and developed asymptotic inferential procedures for $R$ when the sample size is held fixed and the concentration parameter of the von Mises-Fisher distribution diverges to infinity. \cite{Chang1989} extends the results of \cite{Rivest1989} to the case where the covariates are also modeled as random. Moving beyond rotations, a more general family of transformation, the M\"{o}bius group, was introduced by \cite{Downs2003} for the $p = 2$ case. A further generalization using the projective linear group was introduced by \cite{Rosenthal2014}.
\\\\
More recently, nonparametric methods for regression of spherical data have been proposed, which tend to be more flexible than parametric ones. \cite{DiMarzio2014} propose using Taylor polynomials to obtain a component-wise local approximation of the regression map $f: \mathbb{S}^p \rightarrow \mathbb{S}^p$. On the other hand, motivated by the fact that any two points on the sphere are related by a rotation along the shortest arc between them, \cite{DiMarzio2019} propose a nonparametric rotation model through a covariate-dependent rotation. That is, they model the conditional mean direction of the response $Y \in \mathbb{S}^p$ given covariate $x \in \mathbb{S}^p$ as $R_x x$, where the rotation matrix $R_x$ varies with $x$. Since any rotation matrix $R$ can be written as $R = \exp(A)$ where $A$ is a skew-symmetric matrix (i.e., $A^{T} = -A)$ and $\exp(\cdot)$ is the matrix exponential, one can write $R_x = \exp(A_x)$. \cite{DiMarzio2019} propose using Taylor approximations to locally approximate $A_x$ around the point $x$ so that $R_x$ can be estimated using a locally weighted least squares approach. \citet{Rosenthal2017} propose a more general approach where diffeomorphisms are used to model sphere-to-sphere maps.
\\\\
Despite their flexibility, nonparametric approaches to spherical regression present several practical and methodological challenges. In particular, locally weighted methods can be sensitive to bandwidth selection and may become unstable in regions with sparse data coverage on the sphere. Methods based on local approximations may also have difficulty capturing complex global structure. While diffeomorphism-based approaches provide greater modeling flexibility, they typically introduce increased model complexity and reduced interpretability. Furthermore, many existing nonparametric frameworks do not naturally support uncertainty quantification or the identification of heterogeneous structural regions across the sphere.
\\\\
There remains a relative shortage of flexible regression models for sphere-on-sphere regression compared with the rich methodology available in Euclidean settings. A key difficulty lies in the challenge of globally parameterizing regression maps from the sphere to itself, due to the nonlinear geometry and complexities of spherical domains. To address this issue, we adopt a strategy that partitions the sphere into multiple regions and specifies a separate regression map within each region. This localized modeling approach reduces the complexity associated with constructing a single global transformation while allowing greater flexibility in capturing spatially varying relationships across the sphere. 
\\\\
A central challenge in this framework is how to construct such partitions in a principled and flexible manner. To this end, we leverage tools from optimal transport theory to model the partition structure. Within each region, a wide class of parametric sphere-to-sphere maps can then be employed. We further adopt a Bayesian framework for inference, which enables coherent estimation of the regression maps while naturally providing uncertainty quantification and facilitating the identification of heterogeneous structure across the spherical domain. Through real data applications and simulation studies, we demonstrate that our proposed method performs competitively against state-of-the-art methods in general, particularly in less smooth settings. In cases where the data exhibit a clustering structure, the proposed method outperforms existing approaches. Relatedly, recent work has used measure transport to define new notions of distribution and quantile functions on the hypersphere \citep{Hallin2024, Bercu2026}.
\\\\
The main contributions of this work are threefold: (i) we introduce a novel use of semi-discrete optimal transport to induce data-adaptive Laguerre tessellations of the sphere, providing a principled mechanism for partitioning the domain in regression problems; (ii) we develop a general posterior contraction theorem for Bayesian sphere-on-sphere regression models and establish near-parametric contraction rates for our proposed framework; and (iii) we demonstrate through simulations and real-data applications that the method recovers heterogeneous structure while achieving competitive predictive performance.
\\\\
The remainder of the manuscript is organized as follows. Section \ref{sec_background} reviews relevant background on optimal transport theory, with particular emphasis on semi-discrete optimal transport and its approximation through entropic-regularized fully discrete optimal transport. Section \ref{sec_rotation_model} introduces our proposed modeling framework for the rotation case, develops a Markov chain Monte Carlo algorithm for posterior inference of the regression map, and establishes a posterior contraction result. Section \ref{sec_non_rotation_models} then extends the framework to non-rotation models. Section \ref{sec_simulation} presents simulation studies designed to assess the behavior of the proposed model and the performance of the MCMC algorithm, with a focus on recovery of the true regression functions. Section \ref{sec_data_app} contains two real-data applications, illustrating the empirical performance of our method relative to existing parametric and nonparametric alternatives. Finally, Section \ref{sec_discussion} concludes with a discussion of the results and possible directions for future research.

\section{Background on Optimal Transport Theory}
\label{sec_background}

\subsection{Optimal Transport on the Sphere}
\label{sec_optimal_transport}
Given two probability measures $\mu, \nu$ defined on a space ${\cal X}$ equipped with a metric $d$, a transport map $T: {\cal X} \rightarrow {\cal X}$ is said to push forward the measure $\mu$ to $\nu$ (written compactly as $T_{\#}\mu = \nu$) if
$$ \nu(B) = \mu(T^{-1}(B)), \quad \mbox{for any Borel subset } B \subset {\cal X} .$$
Given a cost function $c: {\cal X} \times {\cal X} \rightarrow \mathbb{R}_{+}$, the optimal transport map $S_\nu$ is the solution to Monge's optimization problem \citep{Monge1781}:
\begin{eqnarray}
\label{eqn_opt_monge_formulate}
 S_\nu = \argmin_{T: T_{\#} \mu = \nu } \int_{{\cal X}} c(x, T(x)) \intd\mu(x). 
\end{eqnarray}
In this work, we focus on the spherical domain ${\cal X} = \mathbb{S}^p$ with geodesic distance $d(x, z) = \cos^{-1}(x^{T} z)$ for $x,z \in \mathbb{S}^p,$ and the half squared distance cost function $c(x,z) = d^2(x,z)/2$.
\cite{McCann2001} shows that when the source measure $\mu$ is absolutely continuous, a unique solution to Monge's problem \eqref{eqn_opt_monge_formulate} is guaranteed to exist.

\subsection{Semi-discrete Optimal Transport}\label{sec_semi_discrete}
In the semi-discrete optimal transport setting, the source measure $\mu$ is absolutely continuous whereas the target measure $\nu$ is discrete and supported on a finite set ${\cal Z} := \{z_1, \ldots, z_k\} \subset {\cal X}$. The solution to Monge's problem \eqref{eqn_opt_monge_formulate} induces a partition of the domain into Laguerre cells, defined by
\begin{eqnarray}
\label{eqn_laguerre_cell}
 \mbox{Lag}_{z_j}(\psi) := \{ x \in {\cal X} : c(x,z_j) - \psi(z_j) \le c(x,z) - \psi(z), \forall z \in {\cal Z} \} , \quad z_j \in {\cal Z},
\end{eqnarray}
where $\psi: {\cal Z} \rightarrow \mathbb{R}$ is the dual potential function. The optimal transport map assigns each point $x$ to the atom whose Laguerre cell contains it:
\begin{eqnarray}
\label{eqn_opt_map_semi_discrete}
    S_\nu(x) = \arg \min_{z \in {\cal Z}} \{ c(x,z) - \psi(z) \} , \quad x \in {\cal X}.
\end{eqnarray}
An illustration of of a Laguerre tessellation is shown in Figure~\ref{fig_Laguerre}. Note how, unlike in the classical Voronoi case, a Laguerre cell $ \mbox{Lag}_{z_j}(\psi)$ may not contain its nucleus $z_j$. When ${\cal X} = \mathbb{R}^p$ and the cost is the squared Euclidean distance, efficient algorithms exist for computing the Laguerre tessellation \citep{Fabri2009}; for non-Euclidean spaces, modifications are required \citep{Kitagawa2016,Cui2019}.

\begin{figure}[t!]
\includegraphics[width=\textwidth]{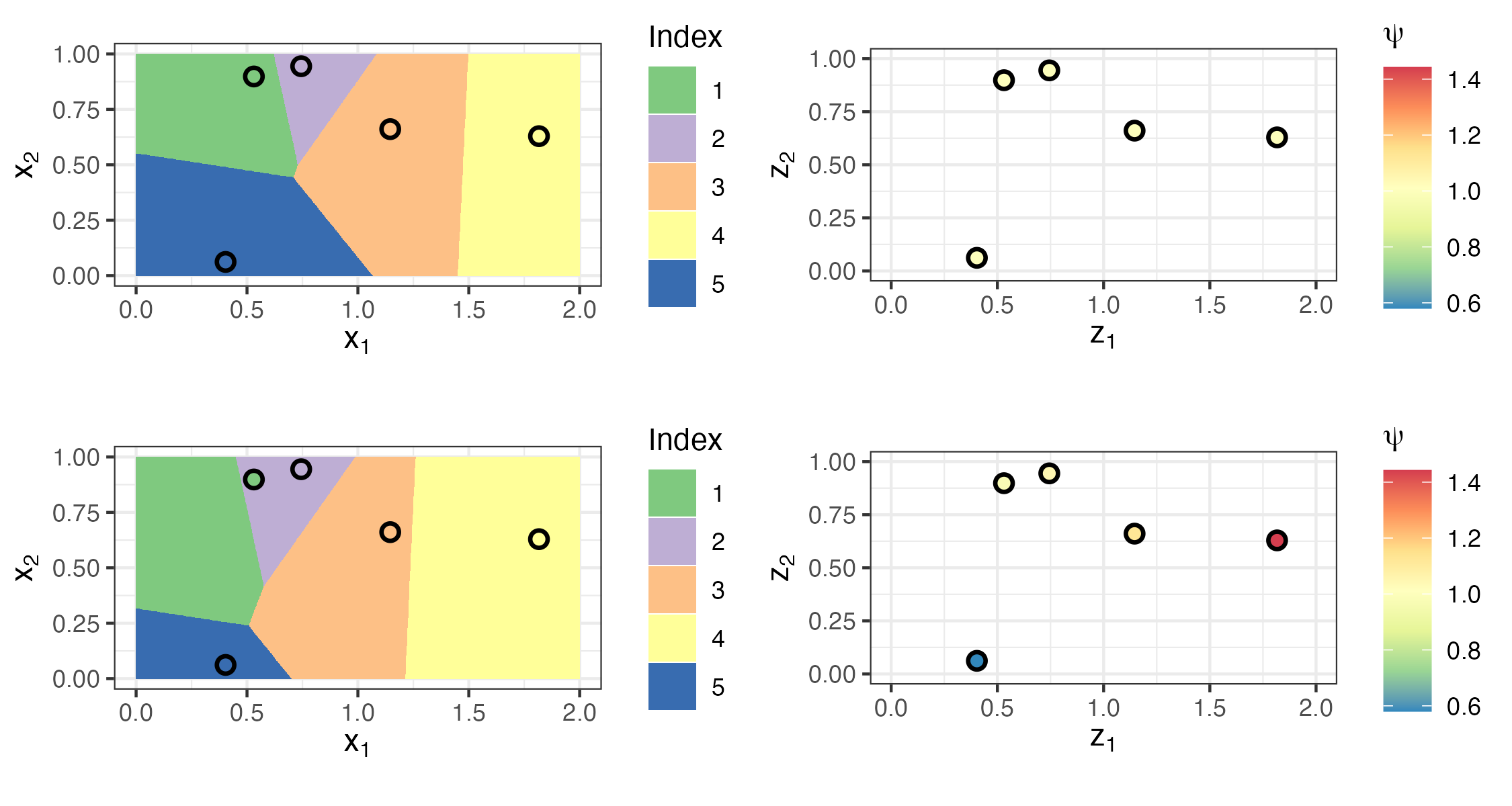}
\caption{Illustration depicting Voronoi (top-left) and Laguerre (bottom-left) tessellations for five atoms on the Euclidean domain $[0,2] \times [0,1]$, with $\psi(z) = 1$ for the Voronoi tessellation (top-right) and $\psi(z) = (z_1^3 + z_2^3)^\frac{1}{5}$ for the Laguerre tessellation (bottom-right). \label{fig_Laguerre}}
\end{figure}

\subsection{Fully-discrete Optimal Transport}
\label{sec_fully_disc_OT}
Although algorithms exist for solving the Monge problem in the semi-discrete setting, they are typically computationally intensive. Fully discrete optimal transport formulations, particularly with entropic regularization, are substantially more efficient.
\\\\
In the fully discrete setting, both the source and target measures are supported on finite sets, and can be written as
\[
\mu_n = \sum_{i=1}^{n} a_i \, \delta_{s_i}, \qquad \nu = \sum_{j=1}^{k} p_j \, \delta_{z_j}.
\]
A transport plan is a coupling matrix $\boldsymbol{\tilde{\gamma}} \in \mathbb{R}_+^{n \times k}$ whose row sums match $\mathbf{a}$ and column sums match $\mathbf{p}$. The set of admissible couplings is
\[
\Pi(\mu_n, \nu) := \Big\{ \boldsymbol{\tilde{\gamma}} \in \mathbb{R}_+^{n \times k} : \boldsymbol{\tilde{\gamma}} \mathbf{1}_k = \mathbf{a}, \, \boldsymbol{\tilde{\gamma}}^\top \mathbf{1}_n = \mathbf{p} \Big\}.
\]
The discrete Kantorovich optimal transport problem seeks the optimal coupling $\boldsymbol{\gamma}$ minimizing the total transport cost:
\begin{equation}
\label{eqn_Kantorovich_discrete}
\boldsymbol{\gamma} := \arg \min_{\boldsymbol{\tilde{\gamma}} \in \Pi(\mu_n, \nu)} \sum_{i=1}^{n} \sum_{j=1}^{k} c(s_i, z_j) \, \tilde{\gamma}_{ij}.
\end{equation}

\subsection{Entropic Regularization}

The solution to \eqref{eqn_Kantorovich_discrete} is generally not unique. Entropic regularization \citep{Cuturi2013} enforces uniqueness by augmenting the objective with a negative entropy term, yielding
\begin{equation}
\label{eqn_entropy_reg}
\boldsymbol{\gamma} := \arg\min_{\boldsymbol{\tilde{\gamma}} \in \Pi(\mu_n, \nu)} \sum_{i,j} c(s_i,z_j) \tilde{\gamma}_{ij} - \omega H(\boldsymbol{\tilde{\gamma}}),
\end{equation}
where $\omega > 0$ is a regularization parameter and $H(\boldsymbol{\tilde{\gamma}}) := - \sum_{i,j} \tilde{\gamma}_{ij} \big( \log(\tilde{\gamma}_{ij}) - 1 \big)$ is the discrete entropy. Introducing the kernel matrix $G \in \mathbb{R}_+^{n \times k}$ with $G_{ij} = e^{-c(s_i, z_j)/\omega}$, the unique solution can be expressed in the scaling form
\begin{equation}
\label{eqn_sinkhorn_solution}
\gamma_{ij} = u_i \, G_{ij} \, v_j,
\end{equation}
for vectors $\mathbf{u} = (u_1, \ldots, u_n)^\top$ and $\mathbf{v} = (v_1, \ldots, v_k)^\top$ satisfying the marginal constraints
\begin{equation}
\label{eqn_sinkhorn_marginal_constraints}
\mathbf{u} \odot (G \mathbf{v}) = \mathbf{a}, \qquad
\mathbf{v} \odot (G^\top \mathbf{u}) = \mathbf{p},
\end{equation}
where $\odot$ denotes element-wise multiplication. These constraints can be efficiently enforced using the Sinkhorn scaling algorithm \citep{Cuturi2013}.

\subsection{Connections between semi-discrete and fully-discrete optimal transport}

Suppose we approximate the continuous source measure $\mu$ by an empirical measure $\mu_n = \frac{1}{n} \sum_{i=1}^{n} \delta_{s_i}$ with $s_1, \ldots, s_n \sim \mu$. Using the entropic-regularized coupling $\boldsymbol{\gamma}$, we define an approximate transport map
\[
\hat{S}_{\nu}(s_i) = z_j, \qquad 
j := \arg\max_{j'=1,\ldots,k} \gamma_{ij'},
\]
which assigns each source point to the target atom receiving the largest proportion of its transported mass. This provides a discrete approximation of the semi-discrete transport map $S_\nu$ at the sampled points. As $\omega \to 0$, the coupling becomes increasingly concentrated around a Monge-type solution, further supporting the use of this assignment rule.

\section{Model}

\label{sec_rotation_model}
We now present our model for sphere-on-sphere regression that builds on a semi-discrete optimal transport map. Let $(X,Y) \in \mathbb{S}^p \times \mathbb{S}^p$ be a random covariate-response pair. Conditional on the covariate $X=x$, we model the mean direction of the response $Y$ using a regression map $f: \mathbb{S}^{p} \rightarrow \mathbb{S}^{p}$. As a frequently used modeling assumption in the literature, we also posit that the error distribution of $Y$ given $X=x$ follows a von Mises-Fisher distribution,
\begin{eqnarray}
\label{eq_response_dist}
    Y | (X = x) \sim g_p(y; f(x), \kappa),
\end{eqnarray}
where $g_p(\cdot\,;f(x),\kappa)$ is the probability density function of the von Mises-Fisher distribution on $\mathbb{S}^p$ with mean direction $f(x) \in \mathbb{S}^{p}$ and concentration parameter $\kappa > 0$. Many existing approaches attempt to model the regression function globally, that is, by directly estimating a mapping $f : \mathbb{S}^p \rightarrow \mathbb{S}^p$ over the entire sphere. However, constructing such a global map is inherently challenging. The difficulty arises from factors such as spatially varying relationships between predictors and responses, and the nonlinear geometry of $\mathbb{S}^p$. As a result, a single global representation may either lack sufficient flexibility to capture local structure or become overly complex and difficult to interpret. The optimal transport framework developed here addresses this by providing a principled, data-driven mechanism for partitioning the sphere, enabling complex global mappings to be constructed from simple, interpretable local transformations.
\\\\
Recall from Section \ref{sec_background} that in the semi-discrete optimal transport setting, the source measure $\mu$ is absolutely continuous on $\mathbb{S}^p$, while the target measure is discrete and given by
\[
\nu = \sum_{j=1}^k p_j \delta_{z_j},
\]
supported in a finite set ${\cal Z} := \{z_1, \ldots, z_k\} \subset \mathbb{S}^p$. Let $\psi : {\cal Z} \to \mathbb{R}$ denote the associated dual potential. This potential induces a Laguerre tessellation of the sphere, yielding the partition
\[
\mathbb{S}^p = \bigcup_{j=1}^k {\cal A}_j,
\]
where each Laguerre cell is defined by
\begin{eqnarray}
\label{eqn_laguerre_sphere}
{\cal A}_j 
:= \left\{ x \in \mathbb{S}^p : 
c(x,z_j) - \psi(z_j) \le c(x,z) - \psi(z), 
\; \forall z \in {\cal Z} 
\right\}.
\end{eqnarray}
The semi-discrete optimal transport (OT) framework provides a natural mechanism for introducing locality into the modeling of the regression map. In particular, by inducing a Laguerre partition of the sphere, the domain $\mathbb{S}^p$ is decomposed into cells ${\cal A}_1, \ldots, {\cal A}_k$, each associated with a representative support point of the target measure. This partitioning enables us to replace the difficult task of estimating a single global mapping with a collection of simpler, locally defined transformations.
\\\\
To this end, let $R_1, \ldots, R_k \in \mathbb{SO}(p+1)$ denote rigid rotations of the ambient space $\mathbb{R}^{p+1}$, where $\mathbb{SO}(p+1)$ is the special orthogonal group consisting of orthogonal matrices with determinant one. We then define the regression map in a piecewise manner by assigning a separate rotation to each Laguerre cell:
\begin{eqnarray}
\label{eq_model_comp}
f(x) = R_j  x, \qquad x \in {\cal A}_j.
\end{eqnarray}
In this formulation, the global regression function is constructed by stitching together locally rigid transformations. We refer to this model as the \emph{Optimal Transport Sphere-on-Sphere Rotation Model} (OT-SOS-Rotation). 
\\\\
This approach provides several advantages. First, within each cell the model remains geometrically simple and interpretable, as it corresponds to a rotation on the sphere. Second, the overall mapping can capture complex, spatially varying relationships through the combination of multiple local rotations. Finally, the number of cells $k$ controls the trade-off between flexibility and parsimony: increasing $k$ enhances expressive power, while smaller values promote smoother and more interpretable structure.
\\\\
Our proposed model naturally accommodates Bayesian inference. For the prior specification, we first place a prior on the number of atoms $k$ in the target measure. Conditional on $k$, we assign independent priors to the discrete target measure $\nu$ - including the atom locations $\{z_1, \ldots, z_k\}$, the associated probability vector $(p_1, \ldots, p_k)$, and the rotation matrices $R_1, \ldots, R_k$.

\subsection{Posterior Contraction Rate}

\label{sec_post_cont}

In this section, we derive the posterior contraction rate of the OT-SOS-Rotation model. Consider an i.i.d.~sample 
\(
{\cal D}_n := \{(x_i, y_i)\}_{i=1}^{n},
\)
where we assume the covariates $\{x_i\}_{i=1}^{n}$ are distributed according to a continuous measure $\mu$ on $\mathbb{S}^p$.
\\\\
Our \emph{working} model for these data is based on the  OT-SOS regression model \eqref{eq_response_dist}:
\begin{eqnarray}
\label{eq_reg_model_general}
Y_i \mid (X_i = x_i) \simiid g_p(y_i; f(x_i), \kappa), \quad i=1,\ldots,n,
\end{eqnarray}
while we assume the \emph{true} data-generating model is
\begin{eqnarray}
\label{eq_reg_model_general_true}
Y_i \mid (X_i = x_i) \simiid g_p(y_i; f_0(x_i), \kappa), \quad i=1,\ldots,n,
\end{eqnarray}
where $f_0 : \mathbb{S}^p \to \mathbb{S}^p$ is the unknown true regression map. 
\\\\
Let $\mathbb{P}_f$ denote the joint distribution of $(X,Y)$ corresponding to regression map $f$ with joint density $p_f(x,y)$, and let $\mathbb{P}_0$ denote the joint distribution corresponding to $f_0$ with density $p_0(x,y)$. Given the sample ${\cal D}_n$ and a prior $\Pi$ on $f$, the posterior distribution of $f$ is
\begin{eqnarray}
\Pi(f \in A \mid {\cal D}_n) = \frac{ \int_A \prod_{i=1}^{n} \left( \frac{p_f(x_i,y_i)}{p_0(x_i, y_i)} \right) \Pi(\mathrm{d} f) }{ \int_{\mathcal{F}} \prod_{i=1}^{n} \left( \frac{p_f(x_i, y_i)}{p_0(x_i, y_i)} \right) \Pi(\mathrm{d} f) },
\end{eqnarray}
where $\mathcal{F} := \{f: \mathbb{S}^{p} \rightarrow \mathbb{S}^p \}$ is the space of regression maps and $A \subseteq \mathcal{F}$. 
\\\\
We study posterior contraction with respect to the integrated $L_2$ semi-metric
\[
\tilde{d}^{2}(f_1,f_2) := \int_{\mathbb{S}^{p}} d^2(f_1(x), f_2(x)) \, \mathrm{d}\rho(x), \quad f_1,f_2 \in \mathcal{F},
\]
where $d(\cdot,\cdot)$ denotes the geodesic distance on $\mathbb{S}^{p}$ and $\rho$ is the uniform measure on $\mathbb{S}^{p}$. The posterior $\Pi(\cdot \mid {\cal D}_n)$ is said to contract around the true regression map $f_0$ at rate $\epsilon_n$ if
\[
\mathbb{E}_0 \big( \Pi( \tilde{d}(f,f_0) > M_n \epsilon_n \mid {\cal D}_n) \big) \to 0
\]
as $n \to \infty$, where $\mathbb{E}_0$ is expectation under $\mathbb{P}_0$, and $M_n$ either grows arbitrarily slowly to infinity or is a sufficiently large constant $M>0$. 
\\\\
The assumptions on the source probability measure, the true regression map, and the prior distributions are given in Assumptions \ref{source_measure_assumption}-\ref{prior_rotation_assumption} in Appendix~\ref{sec_assumptions}, and are discussed in detail in Section \ref{sec_prior}. Under these assumptions, we obtain the main posterior contraction result, showing that the posterior distribution concentrates around the true regression map at the near-parametric rate $n^{-1/2}\log n$.

\begin{theorem}
\label{thm_contraction_rates_opt}
Suppose the source measure $\mu$ satisfies Assumption \ref{source_measure_assumption}, the true regression map $f_0$ satisfies Assumption \ref{true_reg_map_assumption} and the prior satisfies Assumptions \ref{prior_k_assumption}, \ref{prior_measure_assumption}, and \ref{prior_rotation_assumption}. Let $(\epsilon_n)_{n=1}^\infty$ be a sequence such that $\epsilon_n \asymp n^{-1/2} \log n$. Then, for sufficiently large $M>0$,
\[
\mathbb{E}_0 \big( \Pi( \tilde{d}(f,f_0) > M \epsilon_n \mid {\cal D}_n) \big) \to 0
\quad \text{as } n \to \infty.
\]
\end{theorem}

While we use the integrated $L_2$ semi-metric with respect to the uniform measure $\rho$ to study posterior contraction, Assumption \ref{source_measure_assumption} on the source measure $\mu$ implies that the same contraction rate continues to hold when $\rho$ is replaced by $\mu$.

The proof of Theorem \ref{thm_contraction_rates_opt} is provided in Appendix \ref{sec_proof_main}. As a first step, we derive a general posterior contraction result for Bayesian sphere-on-sphere models (Theorem \ref{general_thm}), extending Theorem 2.1 of \cite{Xie2019}, originally established for the Euclidean setting, to our sphere-on-sphere regression framework. We then prove Theorem \ref{thm_contraction_rates_opt} by verifying the conditions of Theorem \ref{general_thm}.

\begin{remark}
Our posterior contraction result assumes the concentration parameter $\kappa$ to be known. Extending the result to the case where $\kappa$ is unknown and assigned a prior distribution with density bounded away from $0$ and $\infty$ is desirable. However, this extension introduces additional complications in the construction of the test statistic used in the proof (similar to those encountered in \cite{Xie2019}); we therefore leave this extension for future research. In practice, $\kappa$ is estimated jointly with other model parameters, and the simulation studies and data applications in Sections~\ref{sec_simulation}--\ref{sec_data_app} demonstrate that the method performs well under this setting.
\end{remark}

\subsection{Posterior Computation}

\subsubsection{Entropic-regularized OT approximation}
Under this modeling framework, evaluating the likelihood function and hence the posterior distribution requires solving a semi-discrete optimal transport problem between a continuous source measure $\mu$ and a discrete target measure $\nu$, which is computationally expensive.  To address this issue, we instead approximate the semi-discrete optimal transport problem using entropic-regularized fully discrete optimal transport, as described in Section~\ref{sec_fully_disc_OT}. Specifically, we approximate the continuous source measure $\mu$ by its empirical counterpart supported on the observed covariates:
\[
\mu_n = \frac{1}{n} \sum_{i=1}^{n} \delta_{x_i}.
\]
This allows us to replace the semi-discrete optimal transport problem with a fully discrete one defined between the empirical measure $\mu_n$ and the discrete target measure $\nu$. 
\\\\
With a regularization parameter $\omega > 0$, the resulting entropic-regularized optimal transport problem becomes strictly convex and admits a unique solution in the form of a coupling matrix $\boldsymbol{\gamma} \in \mathbb{R}_+^{n \times k}$. Using this coupling matrix, we construct an approximate transport map defined on the observed covariate locations by
\[
\hat{S}_{\nu} : \{x_i\}_{i=1}^n \to \{z_j\}_{j=1}^k, 
\qquad 
\hat{S}_{\nu}(x_i) = z_j, \quad 
j := \arg\max_{j'=1,\ldots,k} \gamma_{ij'}.
\]
That is, each covariate $x_i$ is assigned to the target atom $z_j$ receiving the largest proportion of transported mass from $x_i$ under the entropic coupling. Note that this construction is well defined up to a measure 0 set. This provides a computationally efficient approximation to the semi-discrete transport map evaluated at the sample locations.
\\\\
Given an i.i.d. sample ${\cal D}_n := \{(x_i,y_i)\}_{i=1}^{n}$, the (approximate) likelihood is 
\begin{eqnarray}
\label{eqn_likelihood}
L({\cal D}_n; \{R_j\}_{j=1}^k, \{z_j\}_{j=1}^k, \{p_j\}_{j=1}^{k}, \kappa)
=
\prod_{i=1}^{n} g_p\big(y_i; f(x_i), \kappa \big),
\end{eqnarray}
where
\[
f(x_i) = R_j x_i, \qquad 
j := \arg\max_{j'=1,\ldots,k} \gamma_{ij'} .
\]

\subsubsection{Prior Specification}
\label{sec_prior}
In this section, we discuss in more detail the prior assumptions stated in Appendix \ref{sec_assumptions} and describe several modifications introduced for computational efficiency. Our prior specifications for the atom locations, atom weights, and rotation matrices are given in Assumptions~\ref{prior_measure_assumption}--\ref{prior_rotation_assumption}. In particular, as stated in Assumption~\ref{prior_measure_assumption}, we place a uniform prior on the atom locations $z_1,\ldots,z_k$ subject to a minimal separation constraint on the pairwise distances between atoms, together with a Dirichlet prior on the corresponding weights,
\[
(p_1,\ldots,p_k) \sim \mathrm{Dir}(\alpha_p),
\]
for $\alpha_p < 1$. Although this assumption imposes a minimum separation constraint on the atom locations, the constraint vanishes asymptotically. In practice, for posterior sampling, the separation parameter can be chosen arbitrarily small, and thus has negligible effect on the resulting posterior samples.

In addition, as specified in Assumption~\ref{prior_rotation_assumption}, we place a uniform prior on the rotation matrices $R_1,\ldots,R_k$ over the special orthogonal group $\mathbb{SO}(p+1)$. While our posterior contraction rate result assumes that the concentration parameter $\kappa$ is known, in practice it must be estimated. For simplicity, we assign an improper flat prior $\kappa \sim \mathrm{Unif}(0,\infty)$.

Our prior assumption on the number of atoms $k$, given in Assumption~\ref{prior_k_assumption}, requires posterior inference over models of differing dimensions. In principle, this would necessitate the use of trans-dimensional sampling methods such as reversible jump Markov chain Monte Carlo (RJMCMC) \citep{Green1995}. However, such approaches can be computationally demanding in our setting, since each update of $k$ would require repeatedly solving optimal transport problems within the sampling procedure.

To improve computational tractability, we instead perform posterior inference conditional on a fixed number of atoms $k$ in the target measure. Model selection over $k$ is then carried out as a post-processing step by fitting the model over a range of candidate values of $k$ and selecting the optimal value using the Watanabe-Akaike information criterion (WAIC) \citep{Watanabe2010}. This strategy reduces computational cost while retaining flexibility in determining an appropriate model complexity.

\subsubsection{MCMC Sampling}
The MCMC procedure involves sampling from the posterior distribution of the rotation matrices $\{R_j\}_{j=1}^k$, the concentration parameter $\kappa$, and the atom locations and weights $\{z_j\}_{j=1}^k$ and $\{p_j\}_{j=1}^k$. We now describe how proposals for each of these parameters are generated within our MCMC scheme. We note that proposing new atom locations or new atom weights results in a change to the target measure $\nu$, which in turn requires recomputing the entropic-regularized OT problem to update the associated coupling matrix. 
\\\\
\textbf{Proposal for the rotation matrices $R_j$}
\\
Let $I_{p+1}$ denote the identity matrix of order $p+1$, and let $\exp(\cdot)$ denote the matrix exponential:
\[
\exp(A) = I_{p+1} + A + \frac{A^2}{2} + \cdots .
\]
We use the property that any rotation matrix $R \in \mathbb{SO}(p+1)$ can be represented as $R = \exp(A)$, where $A$ is a skew-symmetric matrix satisfying $A^T = -A$. That is, $\exp: so(p+1) \rightarrow \mathbb{SO}(p+1)$.
\\\\
Suppose the current rotation matrix is $R_j  \in \mathbb{SO}(p+1)$. We propose a new rotation matrix $R_j' \in \mathbb{SO}(p+1)$ as follows. Let $\boldsymbol{\epsilon} = (\epsilon_1, \ldots, \epsilon_{\frac{(p+1)p}{2}})^T$ be a multivariate normal random vector with mean zero and covariance $\sigma_\epsilon^2 I_{\frac{p(p+1)}{2}}$, where $\sigma_\epsilon$ is a small positive scalar. Define a mapping
\[
\Omega: \mathbb{R}^{\frac{p(p+1)}{2}} \to \mathbb{R}^{(p+1) \times (p+1)}
\]
that transforms a vector $\boldsymbol{\epsilon}$ into a skew-symmetric matrix. For example, when $p=2$,
\[
\Omega(\boldsymbol{\epsilon}) = 
\begin{bmatrix}
0 & -\epsilon_3 & \epsilon_2 \\
\epsilon_3 & 0 & -\epsilon_1 \\
-\epsilon_2 & \epsilon_1 & 0
\end{bmatrix}.
\]
The proposed rotation matrix is then
\[
R_j' = \exp\big(\Omega(\boldsymbol{\epsilon})\big) R_j .
\]
 The reverse move is given by
\[
R_j = \exp\big(\Omega(-\boldsymbol{\epsilon})\big) R_j'.
\]


Under this construction, the proposal distribution is symmetric due to the normality of $\boldsymbol{\epsilon}$.
\\\\
\textbf{Proposal for the concentration parameter $\kappa$}\\
Given the current value of the concentration parameter $\kappa$, we generate a proposed value $\kappa'$ from a proposal distribution that is symmetric around $\kappa$, ensuring that the proposed value satisfies the positivity constraint $\kappa' > 0$. 
\\\\
\textbf{Proposal for the atom locations $z_j$}\\
For each atom location $z_{j}$, a symmetric proposal generates a new location, subject to the constraint that the proposed location remains within the spatial domain ${\cal S}$. 
\\\\
\textbf{Proposal for the atom weights $(p_j)_{j=1}^k$}\\
For the atom weight vector $\mathbf{p} = (p_1, \ldots, p_k)$, we propose a new value $\mathbf{p}'$ from a Dirichlet distribution centered at the current value:
\[
\mathbf{p}' \sim \mathrm{Dir}(\delta p_1, \ldots, \delta p_k),
\]
where $\delta > 0$ controls the concentration of the proposal around the current weights. Unlike the proposals for $\kappa$ or the rotation matrices, this Dirichlet proposal is not symmetric. Consequently, the corresponding Dirichlet densities must be included in the Metropolis-Hastings acceptance ratio.

\subsection{Predictions at new Covariates}
\label{sec_pred_new_cov}
Because we employ fully discrete optimal transport (OT) to approximate the semi-discrete OT within our Bayesian inference framework, the resulting coupling matrix yields assignments only at the observed covariate locations. However, in many applications it is of interest to make predictions at new, unobserved covariate locations. To address this issue, we propose two complementary approaches.
\\\\
The first approach leverages the fact that the solution to the entropic-regularized OT problem can be expressed in the scaling form given in \eqref{eqn_sinkhorn_solution}--\eqref{eqn_sinkhorn_marginal_constraints}, which allows for efficient interpolation to new covariates. The second approach directly combines the new covariates with the existing covariates to form an augmented source measure, and then solves the entropic-regularized OT problem for a given target measure. This procedure naturally extends the coupling to include the new covariate locations.
\\\\
\textbf{Solution based on scaling form.}
\\
Recall that the solution to the entropic-regularized fully discrete OT problem can be expressed in the scaling form \eqref{eqn_sinkhorn_solution}--\eqref{eqn_sinkhorn_marginal_constraints}:
\begin{eqnarray}
\label{eqn_sinkhorn_solution_repeat}
\gamma_{ij} = u_i \, G_{ij} \, v_j,
\end{eqnarray}
where the kernel matrix is defined as \(G_{ij} = \exp\{-c(x_i, z_j)/\omega\}\), with \(\omega > 0\) denoting the entropy regularization parameter. The scaling vectors \(\mathbf{u} = (u_1, \ldots, u_n)^T\) and \(\mathbf{v} = (v_1, \ldots, v_k)^T\) are chosen to satisfy the marginal constraints
\begin{eqnarray}
\label{eqn_sinkhorn_marginal_constraints_repeat}
\mathbf{u} \odot (G \mathbf{v}) = \mathbf{a}, \qquad
\mathbf{v} \odot (G^\top \mathbf{u}) = \mathbf{p},
\end{eqnarray}
where \(\mathbf{a} = (1/n, \ldots, 1/n)^T\).
\\\\
The multiplicative structure in \eqref{eqn_sinkhorn_solution_repeat} shows that the transport plan depends on the spatial locations only through the kernel matrix \(G\), which is a smooth function of the pairwise distances between points. This property allows us to extend the probabilistic cluster assignments to new, unseen locations without having to solve an additional optimal transport problem.
\\\\
Specifically, for a new location \(x \in \mathbb{S}^p\), let \(\mathbf{q}(x) = (q_j(x))_{j=1}^{k}\) denote its probabilistic allocation vector, where \(q_j(x)\) represents the proportion of mass assigned to cluster \(j\). To maintain consistency with the learned transport structure, we adopt the same kernel-based multiplicative form:
\[
q_j(x) \propto \exp\{-c(x, z_j)/\omega\} \, v_j,
\]
where \(\omega\) is the entropy regularization parameter and \(v_j\) is the scaling vector from the Sinkhorn solution. Normalizing so that \(\sum_{j=1}^{k} q_j(x) = 1\) gives
\[
q_j(x) = \frac{\exp\{-c(x, z_j)/\omega\} \, v_j}{\sum_{l=1}^{k} \exp\{-c(x, z_l)/\omega\} \, v_l}.
\]
This defines a smooth, distance-based interpolation across clusters, with \(\omega\) controlling the degree of spatial smoothing and uncertainty in the assignments.
\\\\
Finally, we assign the new location \(x\) to the cluster with the highest allocation probability by setting \(j := \arg\max_{j'=1,\ldots,k} q_{j'}(x)\), and define the predicted response as
\[
f(x) = R_j x.
\]
This approach provides a principled way to extend predictions to new covariate locations while preserving the learned transport structure and cluster assignments.
\\\\
\textbf{Solution based on an augmented source measure}
\\
An alternative approach is to augment the source measure by combining the existing covariates with new covariates to form an expanded source measure. Suppose we wish to predict responses at new covariates \(x'_1, \ldots, x'_m\). We construct the augmented source measure
\[
\mu_{n+m} = \frac{1}{n+m} \left( \sum_{i=1}^n \delta_{x_i} + \sum_{i=1}^m \delta_{x'_i} \right),
\]
and solve the entropic-regularized discrete OT problem with this source measure and the target measure \(\nu\). This produces a \((n+m) \times k\) coupling matrix \(\boldsymbol{\gamma}\), which naturally extends the learned transport plan to the new covariates. 
\\\\
For each new point \(x'_i\), we assign it to the cluster receiving the largest portion of its mass by setting
\[
j := \arg \max_{j'=1,\ldots,k} \gamma_{n+i,j'},
\]
and define the predicted response as
\[
f(x'_i) = R_j x'_i.
\]
This approach is reasonable if new points are drawn from the same distribution as the existing covariates, i.e., \(x'_1, \ldots, x'_m \sim \mu\). 

\section{Extension to non-Rotation models}
\label{sec_non_rotation_models}
The OT-SOS-Rotation model represents each region in the partition using a separate rotation model. However, this framework allows for greater flexibility. In particular, since the OT-based approach models the partition of the spherical domain, one can choose an arbitrary model within each region. As an example, we consider an extension of the rotation model that uses projective linear transformations, as proposed by \citet{Rosenthal2014}.
\\\\
This approach replaces rotations with general linear transformations drawn from the special linear group
\[
\mathbb{SL}(p+1) = \{ P \in \mathbb{R}^{(p+1)\times(p+1)} : \det(P) = 1 \}.
\]
Since \(\mathbb{SO}(p+1) \subset \mathbb{SL}(p+1)\), this formulation naturally generalizes the rotation-based model. By allowing the full group \(\mathbb{SL}(p+1)\), we can model more general linear transformations and thereby capture richer geometric deformations of the sphere.
\\\\
Given \(P_1, \ldots, P_k \in \mathbb{SL}(p+1)\), the regression map for a covariate \(x_i\) is defined as
\[
f(x_i) = \frac{ P_j x_i}{\| P_j x_i \|}, 
\qquad 
j := \arg\max_{j'=1,\ldots,k} \gamma_{ij'},
\]
where \(\gamma_{ij'}\) is the assignment weight from the OT coupling $\boldsymbol{\gamma}$. We refer to this extension as the \emph{Optimal Transport Sphere-on-Sphere Projective Linear Transformation} (OT-SOS-PLT) model.
\\\\
The group \(\mathbb{SL}(p+1)\) is a Lie group with associated Lie algebra
\[
\mathfrak{sl}(p+1) := \{ A \in \mathbb{R}^{(p+1)\times(p+1)} : \operatorname{tr}(A) = 0 \}.
\]
The matrix exponential provides a smooth map
\[
\exp : \mathfrak{sl}(p+1) \rightarrow \mathbb{SL}(p+1),
\]
which offers a convenient parameterization for constructing proposals in our MCMC algorithm, analogous to the parameterization used for rotation matrices in the OT-SOS-Rotation model. We assume a uniform prior distribution on \(\mathbb{SL}(p+1)\), which is improper due to the non-compactness of \(\mathbb{SL}(p+1)\).
\\\\
The MCMC scheme for the projective linear transformation model proceeds similarly to the rotation-based model, with the primary difference being that we now update the matrices \(P_1, \ldots, P_k \in \mathbb{SL}(p+1)\). To propose a new value for \(P_j\), we first generate a random matrix \(A' \in \mathbb{R}^{(p+1)\times(p+1)}\) with independent standard normal entries:
\[
A'_{ij} \sim N(0,1), \qquad i,j = 1,\ldots,p+1.
\]
We then construct a trace-zero matrix
\[
A = A' - \frac{\operatorname{tr}(A')}{p+1} I_{p+1},
\]
so that \(\operatorname{tr}(A)=0\) and hence \(A \in \mathfrak{sl}(p+1)\).
\\\\
Given a small step size \(\epsilon > 0\), the proposed update for \(P_j\) is defined as
\[
P'_j = P_j \exp(\epsilon A).
\]
Because \(\exp(\epsilon A) \in \mathbb{SL}(p+1)\), this ensures that \(P'_j \in \mathbb{SL}(p+1)\). The reverse move is given by
\[
P_j = P'_j \exp(-\epsilon A).
\]
By Gaussian assumption, we have \( A \overset{d}{=} -A \), hence the proposal distribution in the Lie algebra is symmetric. For small $\epsilon>0$, the exponential map acts as a local chart near the identity, and its Jacobian varies smoothly, so that the associated volume distortion is approximately constant up to $O(\epsilon^2)$.
\\\\
These imply that for step size small $\epsilon > 0$, the proposal density satisfies
\[ q(P|P') \approx q(P'|P) \] where the density is taken with respect to the Haar measure on $\mathbb{SL}(p+1)$.

\section{Simulation Studies}
\label{sec_simulation}

We conduct simulation studies to empirically assess the performance of the proposed methods under controlled settings where the true data-generating mechanism is known. In particular, we simulate data from the OT-SOS-Rotation and OT-SOS-PLT models and apply the corresponding MCMC sampling procedures for posterior inference. To quantify estimation accuracy, we evaluate the discrepancy between the estimated regression map $\hat{f}$ and the true regression map $f_0$ using the metric $\tilde{d}(\hat{f}, f_0)$. This provides an empirical validation of the theoretical posterior contraction rates established in our analysis. The metric \(\tilde{d}(\hat{f}, f_0)\) is estimated using a sample of test (or new) covariates. In our simulation studies and real data applications, we use the augmented-source-measure approach for prediction described in Section~\ref{sec_pred_new_cov}; however, we observe comparable performance with the scaling-form approach.

\subsection{OT-SOS-Rotation}

For the OT-SOS-Rotation model, we generate covariates $x_i \in \mathbb{S}^2$ independently from the uniform distribution on $\mathbb{S}^2$, so that the source measure $\mu$ is uniform over the sphere. To construct the target measure, we fix an integer $k$ and sample atom weights $(p_1,\ldots,p_k) \sim \mathrm{Dir}(1/2,\ldots,1/2)$, together with atom locations $z_1,\ldots,z_k \in \mathbb{S}^2$ drawn independently from the uniform distribution on $\mathbb{S}^2$. This yields the discrete target measure
\[
\nu = \sum_{j=1}^k p_j \delta_{z_j}.
\]
We additionally generate independent random rotation matrices $R_1,\ldots,R_k \in \mathbb{SO}(3)$, which define local transformations associated with each atom. Let
\[
\mu_{\tilde{n}} = \frac{1}{\tilde{n}}\sum_{i=1}^{\tilde{n}} \delta_{x_i}
\]
denote the empirical approximation of $\mu$ where $\tilde{n}$ is a large integer. We then solve the entropic-regularized optimal transport problem between $\mu_{\tilde{n}}$ and $\nu$ (with regularization parameter $\omega$), which partitions the sample $\{x_i\}_{i=1}^{\tilde{n}}$ into $k$ clusters. Let $c_i \in \{1,\ldots,k\}$ denote the cluster membership of $x_i$.
\\\\
We define the regression map used in the simulation as
\[
f_0(x_i) = R_{c_i} x_i.
\]
We note that the regression function $f_0$ is not a fixed map from $\mathbb{S}^2$ to itself, but rather a random function depending on both the realised covariate sample (through the entropic OT coupling between $\mu_n$ and $\nu$) and the regularisation parameter $\omega$. However, for large $\tilde{n}$ and small $\omega$, this construction provides a good approximation of the case of a fixed map induced by semi-discrete optimal transport. Finally, conditional on $x_i$, responses $y_i$ are generated from a vMF distribution with mean direction $f_0(x_i)$ and concentration parameter $\kappa$. This setting is challenging for existing sphere-on-sphere regression approaches because the true regression function contains sharp discontinuities at the cluster boundaries; in particular, the rotations associated with neighboring clusters can differ arbitrarily.
\\\\
In our simulations, we consider $k=5$ and $k=10$, with the expectation that larger values of $k$ lead to a more difficult estimation problem. We also examine two concentration levels, $\kappa=100$ and $\kappa=1000$. We set regularization parameter for entropic-regularized optimal transport $\omega = 0.01$. 
\\\\
For both simulated settings, where the data are generated using either $k=5$ or $k=10$ clusters, we fit the OT-SOS-Rotation model under two specifications: a rotation model which corresponds to the OT-SOS-Rotation model with $k=1$ and the the OT-SOS-Rotation model with correctly specified number of clusters ($k=5$ or $k=10$, respectively). We additionally fit the nonparametric rotation model of \cite{DiMarzio2019} using the R package implementation provided in \cite{Taylor2025} (version 1.1.1) where bandwidth is selected using cross-validation.
\\\\
We apply our MCMC algorithm to fit the OT-SOS-Rotation model on simulated datasets with sample sizes $n=100$, $500$, and $1000$, in the regime where $n \ll \tilde{n}$. Although our posterior contraction results are derived under the assumption that $\kappa$ is fixed, in these experiments $\kappa$ is treated as unknown and estimated jointly with other model parameters.
\\\\
The integrated $L_2$ distance between the estimated regression function $\hat{f}$, obtained from the MCMC samples, and the true regression function $f_0$ is computed using the remaining $\tilde{n}-n$ data points that are not used for model fitting.
The results, reported in Table \ref{tab_sim1_1} for the setting with true $K=5$ and in Table \ref{tab_sim1_2} for true $K=10$, show that for the OT-SOS-Rotation model (with correctly specified numbers of clusters $k=5$ and $k=10$, respectively), the error $\tilde{d}(\hat{f},f_0)$ decreases as $n$ increases. As expected, estimation is more accurate in the simpler $K=5$ setting than in the $K=10$ setting, while performance remains broadly similar across the two concentration levels $\kappa=100$ and $\kappa=1000$.
\\\\
In contrast, for the rotation model without clustering (corresponding to the OT-SOS-Rotation model with $k=1$), the estimation error does not improve substantially as $n$ increases. A similar pattern is observed for the nonparametric rotation method of \cite{DiMarzio2019}. This is due to the fact that these approaches are not designed to capture the discontinuous cluster structure present in the true regression function. Additional experiments using mismatched values of the regularization parameter $\omega$ for data generation and model fitting led to qualitatively similar behavior, suggesting a degree of robustness to moderate misspecification of $\omega$.

\begin{table}[ht]
\small
    \centering
    \caption{Simulation results where data is generated from OT-SOS-Rotation Model with $k=5$. Reported values correspond to the estimated integrated $L_2(\mathbb{S}^2)$ distance $\tilde{d}(\hat{f},f_0)$ between the fitted regression function and the true regression function, together with posterior standard errors, across different sample sizes, and concentration parameters.}
    \label{tab_sim1_1}
    \begin{tabular}{lccc}
        \toprule
        \textbf{Model}  & $\kappa $ & $n$ & $\tilde{d}(\hat{f}, f_0)$ \\
        \midrule
        OT-SOS-Rotation ($k=5$) & 100  & 100 & 0.516 (0.004) \\
             &   100 & 500 & 0.296 (0.003) \\
              & 100  & 1000 & 0.199 (0.003) \\
         & 1000
        & 100 & 0.476 (0.008) \\
         & 1000
        & 500 & 0.305 (0.005) \\
& 1000
        & 1000 &  0.221 (0.005) \\
        \midrule
        Rotation & 100  & 100 & 1.476 (0.024)  \\
             &   100 & 500 & 1.444 (0.009) \\
              & 100  & 1000 & 1.435 (0.003) \\
         & 1000
        & 100 & 1.540 (0.027) \\
         & 1000
        & 500 & 1.430 (0.012) \\
& 1000
        & 1000 & 1.427 (0.006) \\
           \midrule
        Nonparametric Rotation & 100  & 100 &  0.487 \\
             &   100 & 500 &  0.459 \\
              & 100  & 1000 & 0.453 \\
         & 1000
        & 100 & 0.457 \\
         & 1000
        & 500 &  0.436 \\
& 1000
        & 1000 & 0.433 \\
        \bottomrule
    \end{tabular}
\end{table}

\begin{table}[ht]
\small
    \centering
    \caption{Simulation results where data is generated from OT-SOS-Rotation Model with $k=10$. Reported values correspond to the estimated integrated $L_2(\mathbb{S}^2)$ distance $\tilde{d}(\hat{f},f_0)$ between the fitted regression function and the true regression function, together with posterior standard errors, across different sample sizes, and concentration parameters.}
    \label{tab_sim1_2}
    \begin{tabular}{lccc}
        \toprule
        \textbf{Model}  & $\kappa $ & $n$ & $\tilde{d}(\hat{f}, f_0)$ \\
        \midrule
        OT-SOS-Rotation ($k=10$) & 100  & 100 & 0.599 (0.014) \\
             &   100 & 500 &  0.368 (0.005) \\
              & 100  & 1000 & 0.303 (0.005)  \\
         & 1000
        & 100 & 0.612 (0.008) \\
         & 1000
        & 500 & 0.382 (0.004) \\
& 1000
        & 1000 & 0.336 (0.004) \\
        \midrule
        Rotation & 100  & 100 & 1.329 (0.005) \\
             &   100 & 500 & 1.320 (0.005) \\
              & 100  & 1000 &  1.317 (0.003) \\
         & 1000
        & 100 &  1.375 (0.010) \\
         & 1000
        & 500 & 1.346 (0.009) \\
& 1000
        & 1000 & 1.333 (0.003) \\
           \midrule
        Nonparametric Rotation & 100  & 100 & 0.728   \\
             &   100 & 500 & 0.626  \\
              & 100  & 1000 & 0.584 \\
         & 1000
        & 100 & 0.751 \\
         & 1000
        & 500 & 0.642 \\
& 1000
        & 1000 & 0.602 \\
        \bottomrule
    \end{tabular}
\end{table}

\subsection{OT-SOS-PLT}
For the OT-SOS-PLT model, the generation of covariates and the target measure is the same as in the OT-SOS-Rotation setting. Likewise, we solve the same entropic-regularized optimal transport problem (with $\omega = 0.01$) to induce a partition of the covariate space. The key distinction lies in the class of transformations used to define the regression function. Specifically, we generate random matrices $P_1,\ldots,P_k \in \mathbb{SL}(3)$, where
\[
\mathbb{SL}(3)=\{P \in \mathbb{R}^{3\times 3} : \det(P)=1\}.
\]
To construct such matrices, we first simulate matrices $A_j \in \mathbb{R}^{3\times 3}$ for $j=1,\ldots,k$, with independent entries
\[
(A_j)_{l_1,l_2} \sim N(0,1), \qquad l_1,l_2=1,2,3.
\]
We then rescale each matrix to have unit determinant by setting
\[
P_j=\frac{A_j}{\det(A_j)^{1/3}},
\]
so that $P_j \in \mathbb{SL}(3)$.
\\\\
The true regression function $f_0$ is defined piecewise through the cluster-specific transformations,
\[
f_0(x_i)=\frac{P_{c_i}x_i}{\|P_{c_i}x_i\|},
\]
where $c_i \in \{1,\ldots,k\}$ denotes the cluster assignment induced by the optimal transport partition.
We fit the OT-SOS-PLT model, together with the PLT model (corresponding to the OT-SOS-PLT model with $k=1$) and the nonparametric rotation model. The results are reported in Table \ref{tab_sim2_1} for the setting with true $k=5$ and in Table \ref{tab_sim2_2} for true $k=10$.
\\\\
We observe a pattern similar to that seen for the OT-SOS-Rotation model: for the OT-SOS-PLT model with the correctly specified number of clusters, estimation accuracy improves as the sample size increases. In contrast, the PLT model without clustering and the nonparametric rotation model show little improvement in estimation accuracy as $n$ increases.

\begin{table}[t!]
\small
    \centering
    \caption{Simulation results where data is generated from OT-SOS-PLT Model with $k=5$. Reported values correspond to the estimated integrated $L_2(\mathbb{S}^2)$ distance $\tilde{d}(\hat{f},f_0)$ between the fitted regression function and the true regression function, together with posterior standard errors, across different sample sizes, and concentration parameters.}
    \label{tab_sim2_1}
    \begin{tabular}{lccc}
        \toprule
        \textbf{Model}  & $\kappa $ & $n$ & $\tilde{d}(\hat{f}, f_0)$ \\
        \midrule
        OT-SOS-PLT ($k=5$) & 100  & 100 & 0.585 (0.011)  \\
             &   100 & 500 & 0.337 (0.006)  \\
              & 100  & 1000 & 0.227 (0.004)   \\
         & 1000
        & 100 & 0.453 (0.008) \\
         & 1000
        & 500 & 0.246 (0.005) \\
& 1000
        & 1000 & 0.203 (0.005) \\
        \midrule
        PLT & 100  & 100 & 0.808 (0.015) \\
             &   100 & 500 & 0.774 (0.012)  \\
              & 100  & 1000 & 0.777 (0.013)  \\
         & 1000
        & 100 &  0.982 (0.023)  \\
         & 1000
        & 500 & 0.937 (0.010)  \\
& 1000
        & 1000 & 0.902 (0.007) \\
           \midrule
        Nonparametric Rotation & 100  & 100 &  0.558   \\
             &   100 & 500 &   0.447 \\
              & 100  & 1000 &  0.443 \\
         & 1000
        & 100 & 0.527 \\
         & 1000
        & 500 & 0.484 \\
& 1000
        & 1000 & 0.472 \\
        \bottomrule
    \end{tabular}
\end{table}

\begin{table}[t!]
\small
    \centering
    \caption{Simulation results where data is generated from OT-SOS-PLT Model with $k=10$. Reported values correspond to the estimated integrated $L_2(\mathbb{S}^2)$ distance $\tilde{d}(\hat{f},f_0)$ between the fitted regression function and the true regression function, together with posterior standard errors, across different sample sizes, and concentration parameters.}
    \label{tab_sim2_2}
    \begin{tabular}{lccc}
        \toprule
        \textbf{Model}  & $\kappa $ & $n$ & $\tilde{d}(\hat{f}, f_0)$ \\
        \midrule
        OT-SOS-PLT ($k=10$) & 100  & 100 & 0.638 (0.018)  \\
             &   100 & 500 &  0.476 (0.008) \\
              & 100  & 1000 &  0.381 (0.006) \\
         & 1000
        & 100 & 0.588 (0.009)\\
         & 1000
        & 500 & 0.445 (0.007) \\
& 1000
        & 1000 & 0.402 (0.007) \\
        \midrule
        PLT & 100  & 100 & 1.447 (0.007) \\
             &   100 & 500 & 1.423 (0.005)   \\
              & 100  & 1000 & 1.424 (0.006)   \\
         & 1000
        & 100 & 1.557 (0.007)  \\
         & 1000
        & 500 &  1.551 (0.010) \\
& 1000
        & 1000 & 1.534 (0.012)  \\
           \midrule
        Nonparametric Rotation & 100  & 100 &   0.715 \\
             &   100 & 500 &    0.531\\
              & 100  & 1000 &   0.536\\
         & 1000
        & 100 &  0.575\\
         & 1000
        & 500 &  0.496 \\
& 1000
        & 1000 & 0.502 \\
        \bottomrule
    \end{tabular}
\end{table}

\section{Data Applications}
\label{sec_data_app}

In this section, we illustrate the applications of the OT-SOS-Rotation and OT-SOS-PLT models across two real-world datasets. In particular, we compare the predictive performance of our proposed models with that of existing parametric and nonparametric approaches. The Bayesian framework naturally provides uncertainty quantification for predictions. Moreover, the partitions induced by optimal transport reveal heterogeneous regions on the sphere, highlighting areas with distinct local behaviors and enabling more interpretable modeling of complex spherical relationships.
\\\\
We assess the predictive performance of the OT-SOS-Rotation and OT-SOS-PLT models in comparison with the classical rotation model and projective linear transformation (PLT) model \citep{Rosenthal2014}, as well as the nonparametric rotation model \citep{DiMarzio2019}. It is worth noting that the rotation model corresponds to a special case of OT-SOS-Rotation in which the number of atoms in the target measure is set to one. Similarly, the PLT model is a special case of OT-SOS-PLT with a single atom in the target measure. These comparisons highlight the advantage of incorporating a partition of the sphere, as it allows the model to capture heterogeneous local structures and provides greater flexibility in representing complex sphere-to-sphere mappings. As with the simulation studies, the nonparametric rotation method is implemented using the R package of \citet{Taylor2025}, where cross-validation is used to determine the smoothing parameters.
\\\\
We employ the Root Mean Squared Error (RMSE) and Mean Absolute Error (MAE) as 
performance metrics, measured in terms of  the 
geodesic distance $d(\cdot, \cdot)$ on the sphere. Specifically, for $n_{test}$ 
test observations, these metrics are defined as:
\[
\text{RMSE} = \sqrt{\frac{1}{n_{test}} \sum_{i=1}^{n_{test}} d(y_i, \hat{y}_i)^2}, \quad 
\text{MAE} = \frac{1}{n_{test}} \sum_{i=1}^{n_{test}} d(y_i, \hat{y}_i)
\]
where $y_i$ is the observed response and $\hat{y}_i$ is the predicted value.
\\\\
We focus on uncertainty quantification for the predicted response direction at a given covariate value. For each covariate $x_i$, we assess uncertainty in the associated predicted response direction using posterior samples obtained from the fitted model. Specifically, let
\[
\hat{y}_i^{(1)}, \ldots, \hat{y}_i^{(M)} \in \mathbb{S}^p
\]
denote the predicted mean directions corresponding to $M$ posterior draws, and let 
\[ \kappa^{(1)}, \ldots, \kappa^{(M)} \]
denote the $M$ posterior draws of concentration parameter. For $j=1,\ldots,M$, we sample $y_i^{(j)}$ from the vMF distribution with mean direction $\hat{y}_i^{(j)}$ and concentration parameter $\kappa^{(j)}$. 
\\\\
We then compute the overall mean predicted direction
\[
\bar{y}_i
=
\frac{\sum_{j=1}^{M} y_i^{(j)}}
{\left\|\sum_{j=1}^{M} y_i^{(j)}\right\|}
\]
which serves as a point estimate of the posterior mean direction at the covariate value $x_i$. To assess posterior uncertainty around this direction, we compute the angular deviations between each posterior sample and the mean predicted direction:
\[
\theta_i^{(j)}
=
\cos^{-1}\!\left(
(y_i^{(j)})^{\top}\bar{y}_i
\right),
\quad j=1,\ldots,M.
\]
These angles represent the geodesic distances on the sphere between each sampled predicted direction and the posterior predicted direction.
\\\\
Next, we compute the empirical $(1-\alpha)$-quantile of the set $\{\theta_i^{(j)}\}_{j=1}^{M}$. This quantile defines the half-angle of a spherical cap centered at $\bar{y}_i$, which provides a posterior credible region for the predicted response direction at the covariate value $x_i$. Specifically, the resulting region
\[
\mathcal{C}_{i,\alpha}
=
\left\{
y \in \mathbb{S}^2 :
y^{\top}\bar{y}_i \ge \cos(\theta_{i,\alpha})
\right\}
\]
represents a $100(1-\alpha)\%$ posterior credible region on the sphere describing uncertainty in the predicted direction direction.

\subsection{Cyclone Data}
The first application we consider arises in the statistical analysis of tropical cyclone trajectories, an important problem in climate science and natural hazard modelling due to its relevance for understanding storm genesis, movement patterns, and associated risks to coastal populations and infrastructure. For this purpose, we analyse data obtained from the United States National Hurricane Center\footnote{https://www.nhc.noaa.gov/data/hurdat/}. The data contains details of cyclones in the North Pacific Ocean and comprises six-hourly data points with information on the location, maximum wind speeds, and central pressure of all identified tropical and subtropical cyclones, totaling 1207 instances between the year 1949 and 2022. Previously, this dataset was examined by \cite{Ng2022} within the framework of Poisson point process intensity estimation. In contrast, our approach focuses on modelling the relationship between cyclone initiation and termination locations: specifically, we treat the starting location of each cyclone as the covariate and the corresponding end location as the response, allowing us to investigate how spatial origin influences subsequent displacement patterns.
\\\\
For predictive performance evaluation, we partition the dataset into 1,107 training observations and 100 test observations, preserving the original temporal ordering to avoid information leakage from future to past events. The models are trained on the training partition and evaluated on the corresponding test set to determine their predictive accuracy. For the OT-SOS-Rotation and OT-SOS-PLT models, it is necessary to determine an appropriate number of atoms in the target measure. To do so, we fit a sequence of candidate models with $k = 2, 3, \ldots$ using the training data and increment $k$ iteratively until the Watanabe-Akaike Information Criterion (WAIC) ceases to improve, thereby implementing a greedy model selection procedure. The model associated with the selected optimal value of $k$ is then used to generate predictions on the held-out test set.
\\\\
The results of the predictive performance comparison are presented in Table~\ref{tab_cyclone}. In addition to reporting the value of $k$ corresponding to the best-performing OT-SOS-Rotation and OT-SOS-PLT models, we also include results for the case $k = 2$. This facilitates comparison with the rotation and PLT models, which correspond to the special case $k = 1$. We observe that, for both the OT-SOS-Rotation and OT-SOS-PLT models, increasing $k$ from 1 to 2 leads to noticeably improved performance. Compared with the nonparametric rotation model, both the best-performing OT-SOS-Rotation model ($k = 14$) and the best-performing OT-SOS-PLT model ($k = 10$) achieve smaller MAE values, although they exhibit slightly larger RMSE values.
\\\\
In addition, we illustrate the model fit on the training set and the predictive performance on the test set using the best-performing OT-SOS-PLT model, corresponding to $k = 10$. Specifically, we display randomly selected covariate–response pairs, together with their associated predicted mean directions, from both the training set used for model fitting and the test set used for prediction. These results are shown in Figure~\ref{fig:hur_fit_pred}.
\\\\
We further illustrate the posterior credible regions for the predicted mean responses at randomly selected covariate locations, along with the corresponding observed responses, in Figure~\ref{fig:hur_conf_pred}. These credible regions take the form of spherical caps centred at the posterior mean predicted directions, thereby quantifying uncertainty in the predicted mean responses on the sphere.
\\\\
The OT-SOS-Rotation and OT-SOS-PLT models naturally induce a clustering of covariate locations through the OT-based partitioning of the spherical domain. This clustering is useful for identifying potentially heterogeneous regions of the sphere in which covariate–response relationships vary spatially. By assigning each covariate location to a region associated with a local transformation, the models provide an interpretable segmentation of the domain that highlights spatial variation in the underlying regression structure. We illustrate the resulting clustering in Figure~\ref{fig_cluster_hur}.

\begin{table}[ht]
\small
    \centering
    \caption{Predictive Performance Comparison on the Test Set for Tropical Cyclone dataset}
    \label{tab_cyclone}
    \begin{tabular}{lccc}
        \toprule
        \textbf{Model} & \textbf{Number of Atoms} & \textbf{MAE} & \textbf{RMSE} \\
        \midrule
        OT-SOS-Rotation & 2  & 0.122 & 0.150 \\
           OT-SOS-Rotation  &   14 & 0.116 & 0.140 \\
        OT-SOS-PLT      & 2  & 0.118 & 0.143 \\
         OT-SOS-PLT      & 10
        & \textbf{0.112} & 0.140 \\
        \midrule
        Rotation      &  1 & 0.125 & 0.152 \\
        PLT      & 1  & 0.120 & 0.145 \\
        Nonparametric Rotation      &  NA & 0.120 & \textbf{0.138} \\
        \bottomrule
    \end{tabular}
\end{table}

\begin{figure}[htbp]
\centering

\begin{minipage}{0.04\textwidth}
    \centering
    \rotatebox{90}{Training}
\end{minipage}
\begin{minipage}{0.95\textwidth}
    \subfloat[]{
        \includegraphics[width=0.3\textwidth]{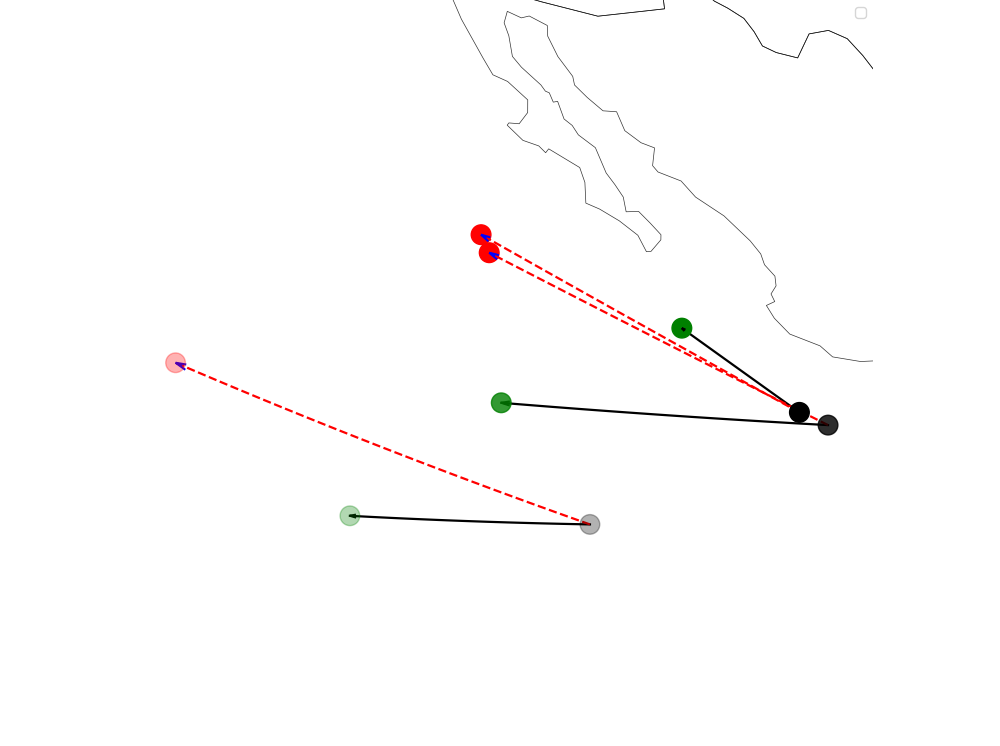}
    }
    \hfill
    \subfloat[]{
        \includegraphics[width=0.3\textwidth]{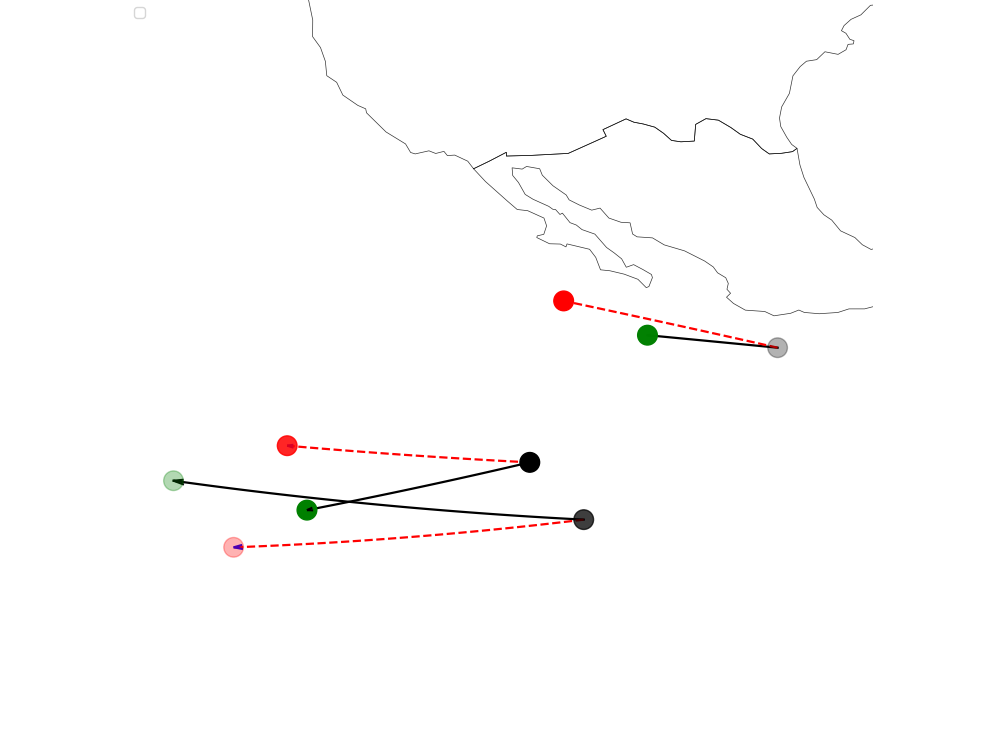}
    }
    \hfill
    \subfloat[]{
        \includegraphics[width=0.3\textwidth]{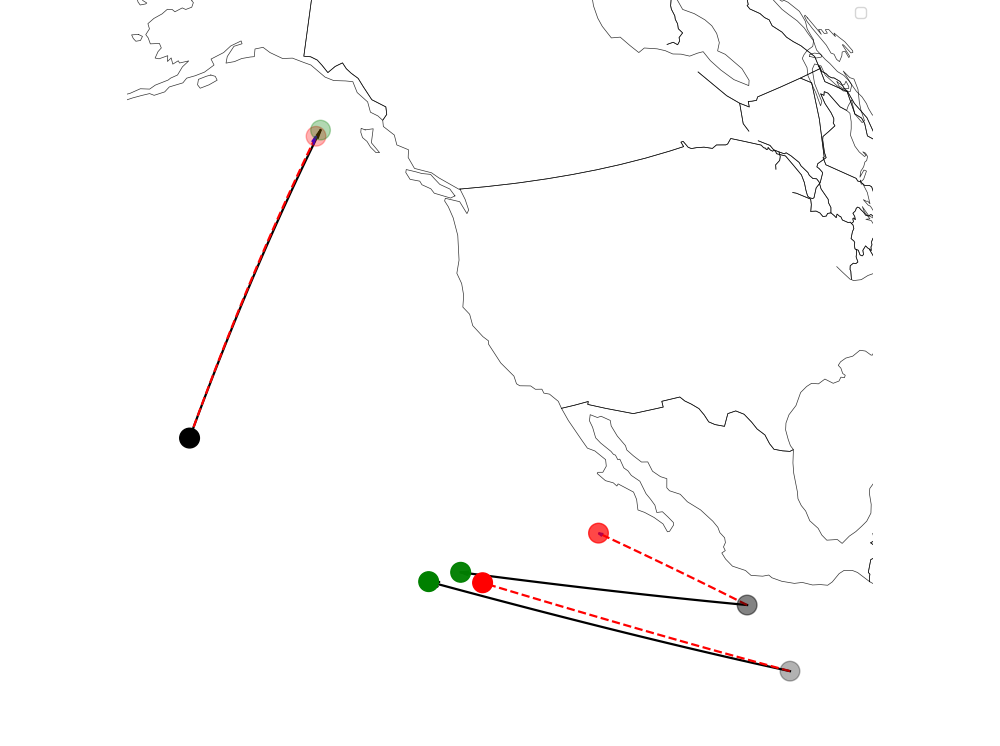}
    }
\end{minipage}

\vspace{0.5cm}

\begin{minipage}{0.04\textwidth}
    \centering
    \rotatebox{90}{Testing}
\end{minipage}
\begin{minipage}{0.95\textwidth}
    \subfloat[]{
        \includegraphics[width=0.3\textwidth]{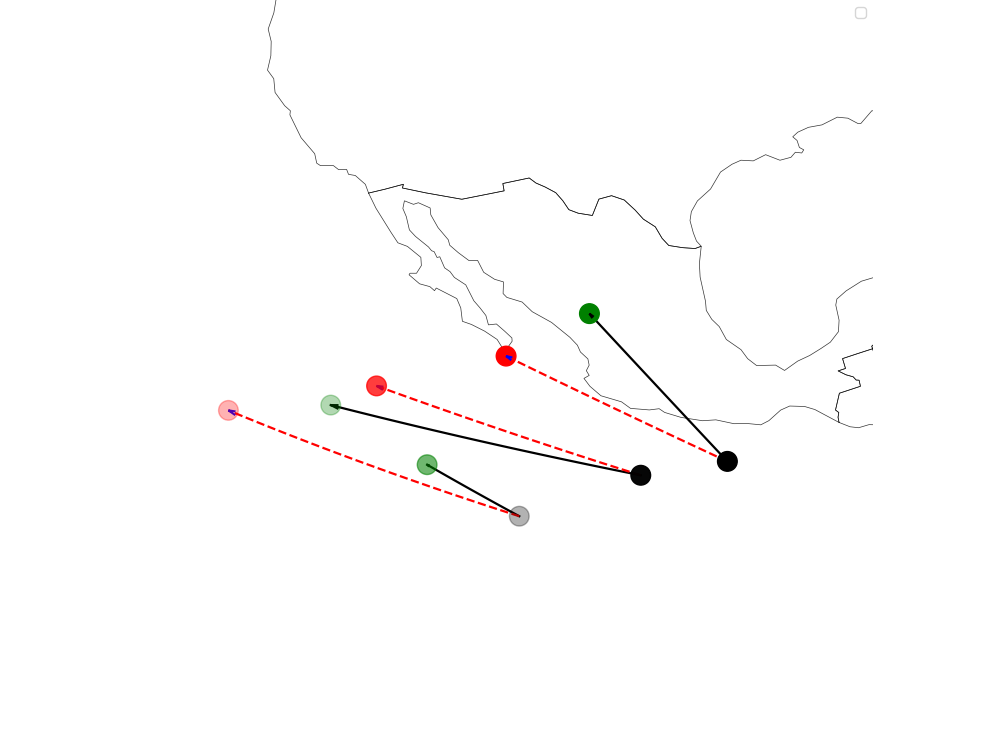}
    }
    \hfill
    \subfloat[]{
        \includegraphics[width=0.3\textwidth]{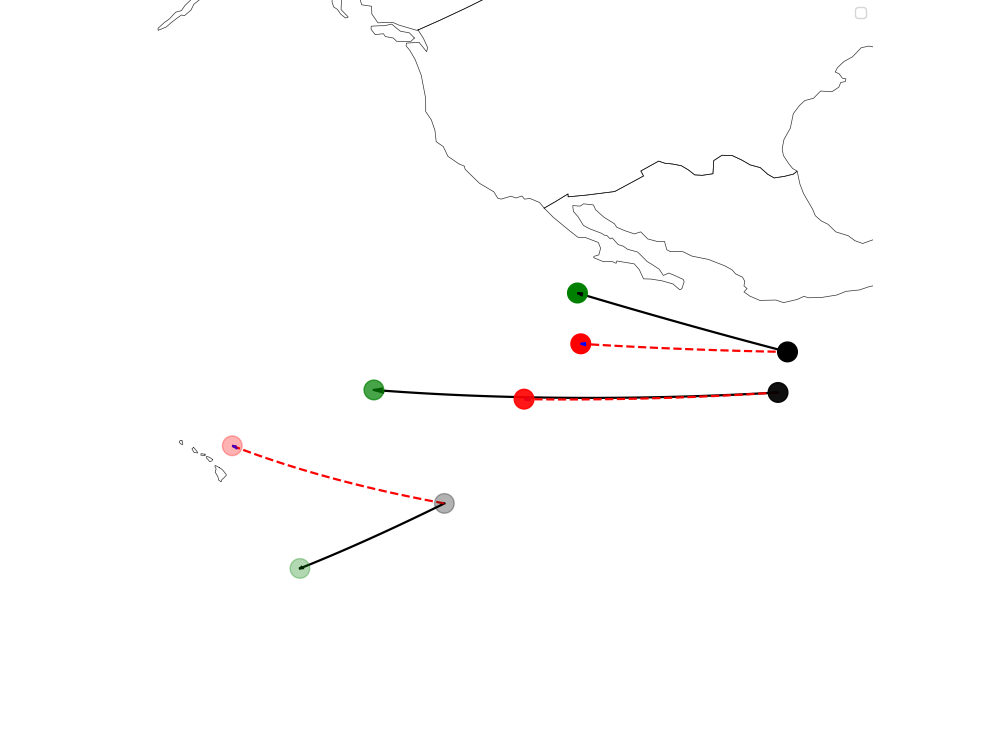}
    }
    \hfill
    \subfloat[]{
        \includegraphics[width=0.3\textwidth]{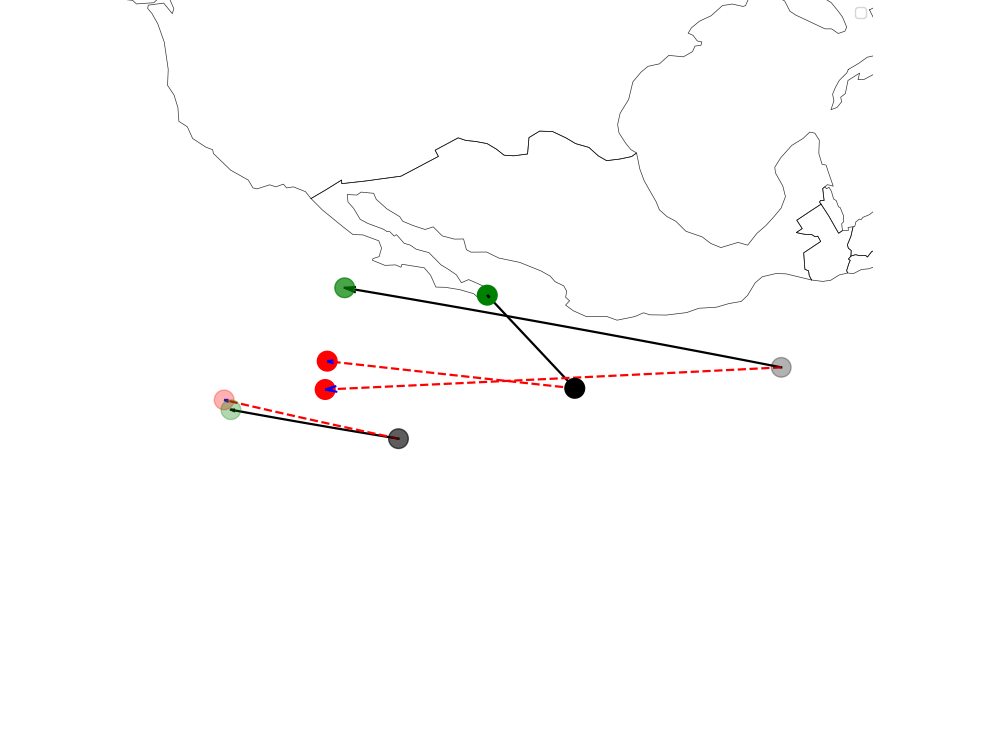}
    }
\end{minipage}

\caption{
Top row: Randomly selected covariate-response pairs with their predicted mean responses for the training set of the cyclone dataset. Bottom row: Randomly selected covariate-response pairs with predicted mean responses for the test set. Black dots denote covariates, green dots denote responses, and red dots denote predicted mean responses. Solid lines connect covariates to responses, while dashed lines connect covariates to predicted mean responses.
}
\label{fig:hur_fit_pred}

\end{figure}

\begin{figure}[htbp]
\centering

\subfloat[]{
    \includegraphics[width=0.3\textwidth]{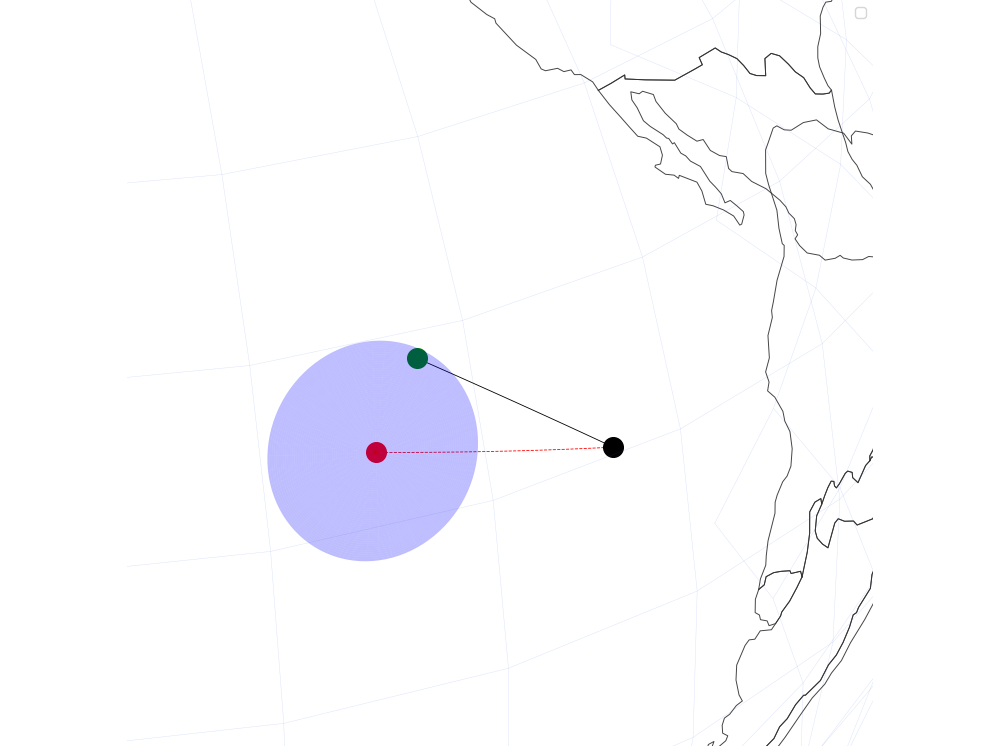}
}
\hfill
\subfloat[]{
    \includegraphics[width=0.3\textwidth]{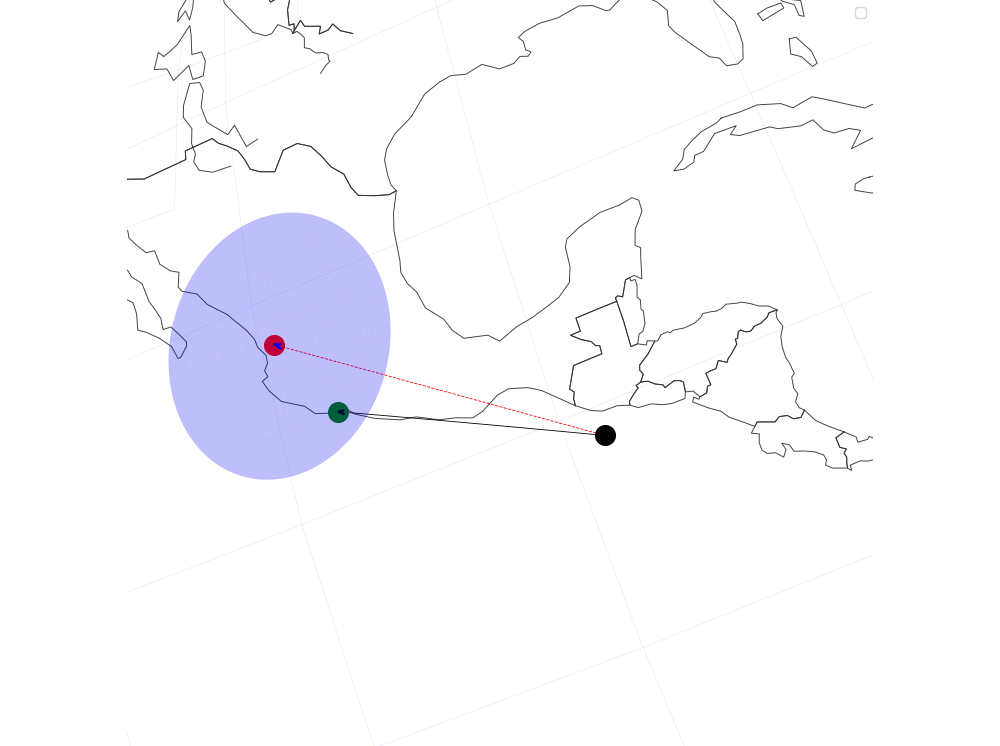}
}
\hfill
\subfloat[]{
    \includegraphics[width=0.3\textwidth]{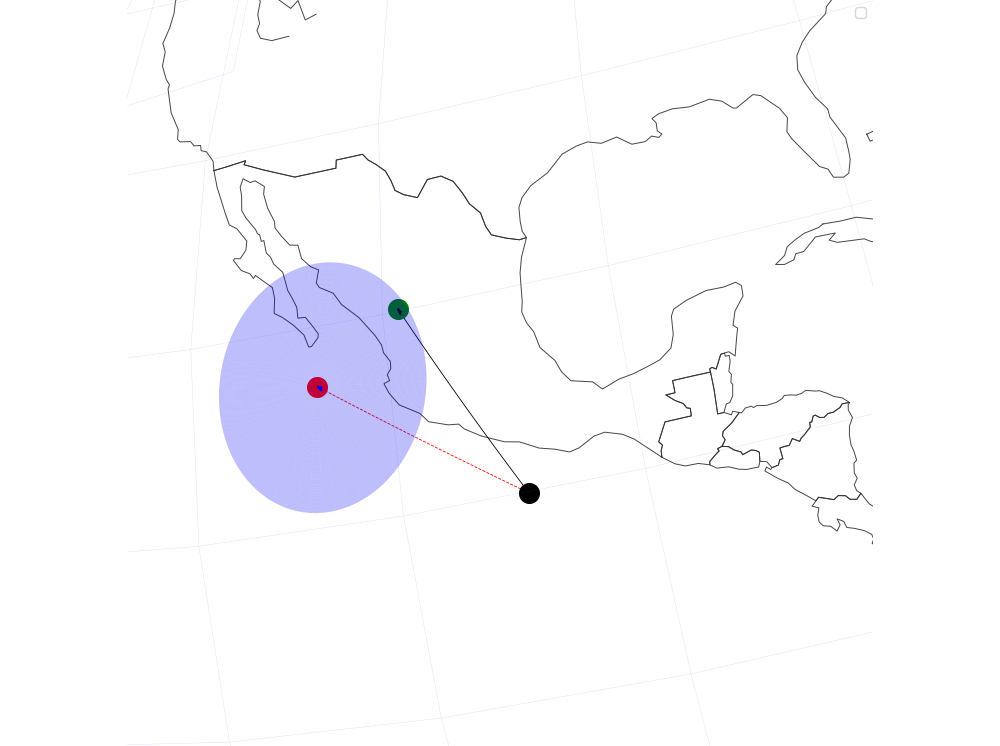}
}

\caption{
Posterior 99\% credible regions for the predicted mean response direction for the Tropical Cyclone dataset on the test set. The spherical caps are centered at the posterior mean directions. Black dots denote covariates, green dots denote observed responses, and red dots denote predicted mean responses.
}
\label{fig:hur_conf_pred}

\end{figure}

\begin{figure}[htbp]
    \centering

    \subfloat[Tropical Cyclone\label{fig_cluster_hur}]{
        \includegraphics[width=0.48\textwidth]{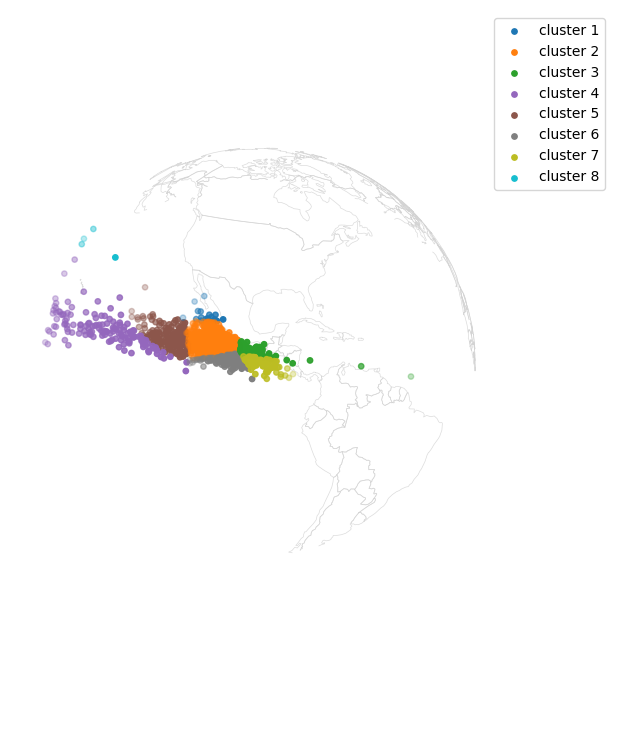}
    }
    \hfill
    \subfloat[Vector-Cardiogram\label{fig_cluster_cardio}]{
        \includegraphics[width=0.48\textwidth]{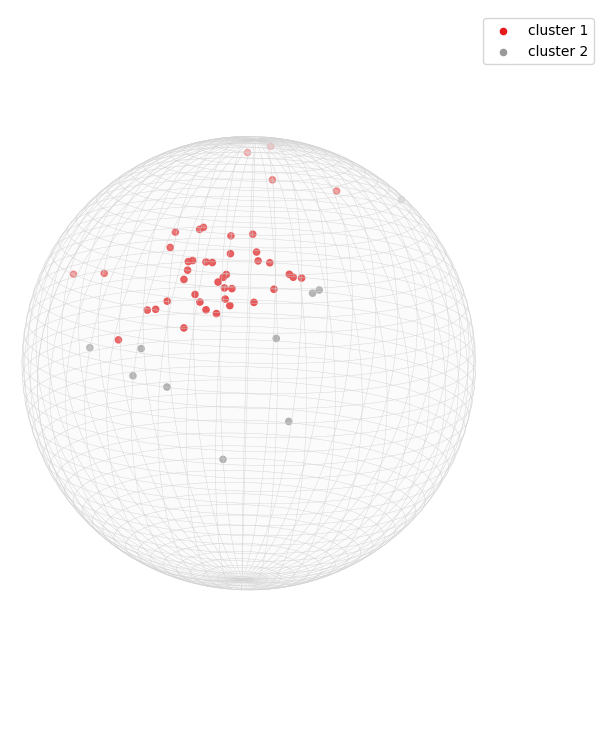}
    }

    \caption{Clustering of covariate locations induced by the OT-SOS-PLT model across two datasets.}
    \label{fig_cluster_all}
\end{figure}

\subsection{Vector-Cardiogram Data}
The second application we consider arises in vectorcardiography, a branch of electrocardiology that represents the heart’s electrical activity as time-varying vectors in three-dimensional space. Such representations are inherently directional and provide clinically informative summaries for diagnosing and monitoring cardiac abnormalities. We analyse data originally described by \cite{Downs2003} and further studied in \cite{Rosenthal2017}. The dataset consists of vectorcardiogram recordings from a cohort of 53 children aged between 11 and 19 years. For each subject, two recordings are available: one obtained using the Frank system and the other using the McFee system. These two systems differ in the spatial configuration of lead placement, leading to distinct but related representations of the same underlying cardiac electrical activity.
\\\\
The primary objective of this analysis is to characterise and quantify the relationship between these two measurement systems. From a statistical perspective, this corresponds to learning a mapping between two sets of spherical-valued responses arising from different observation mechanisms. We therefore treat the directional vectors obtained under the Frank system as covariates, and the corresponding directional vectors obtained under the McFee system as responses, enabling us to model the systematic transformation between these two clinically relevant representations of cardiac activity.
\\\\
The dataset has a relatively small sample size ($n = 53$), and therefore we assess predictive performance using a leave-one-out cross-validation (LOOCV) approach. Specifically, the model is fitted to 52 observations and used to predict the held-out response, with this procedure repeated for each observation in the dataset. In this low-sample regime, WAIC-based model selection is less stable, and we instead evaluate predictive performance for $k = 2, 3$, selecting the smallest value of $k$ beyond which no further improvement in predictive accuracy is observed.
\\\\
The results of the predictive performance comparisons are presented in Table~\ref{tab_cardio}. Overall, predictive accuracy is lower than that observed in the cyclone dataset, which is likely attributable to the limited sample size as well as increased heterogeneity in the covariate–response relationship across the sphere. In this setting, both the OT-SOS-Rotation and OT-SOS-PLT models outperform the nonparametric rotation approach, indicating improved flexibility in capturing systematic differences between the two vectorcardiographic systems.
\\\\
Finally, we fit the best-performing OT-SOS-PLT model with $k = 2$ to the full dataset and display randomly selected covariate–response pairs together with their corresponding predicted mean directions in Figure~\ref{fig:cardio_fit}. The induced clustering of covariate locations is shown in Figure~\ref{fig_cluster_cardio}, illustrating the spatial segmentation of the sphere learned by the model and highlighting regions with distinct transformation behaviour between the two recording systems.

\begin{table}[ht]
\small
    \centering
    \caption{Predictive Performance Comparison using LOO approach for Vector-Cardiogram dataset}
    \label{tab_cardio}
    \begin{tabular}{lccc}
        \toprule
        \textbf{Model} & \textbf{Number of Atoms} & \textbf{LOO MAE} & \textbf{LOO RMSE} \\
        \midrule
        OT-SOS-Rotation & 2  & 0.312 & 0.407  \\
           OT-SOS-Rotation  & 3   & 0.284 & 0.386  \\
        OT-SOS-Rotation  & 4   & 0.307 & 0.392 \\
        OT-SOS-PLT      & 2  & \textbf{0.266} & \textbf{0.321} \\
         OT-SOS-PLT      & 3
        & 0.273 & 0.327 \\
        \midrule
        Rotation      &  1 & 0.322 & 0.420 \\
        PLT      &  1 & 0.298 & 0.374 \\
        Nonparametric Rotation      &  NA & 0.317 &  0.418 \\
        \bottomrule
    \end{tabular}
\end{table}

\begin{figure}[htbp]
\centering

\subfloat[]{
    \includegraphics[width=0.3\textwidth]{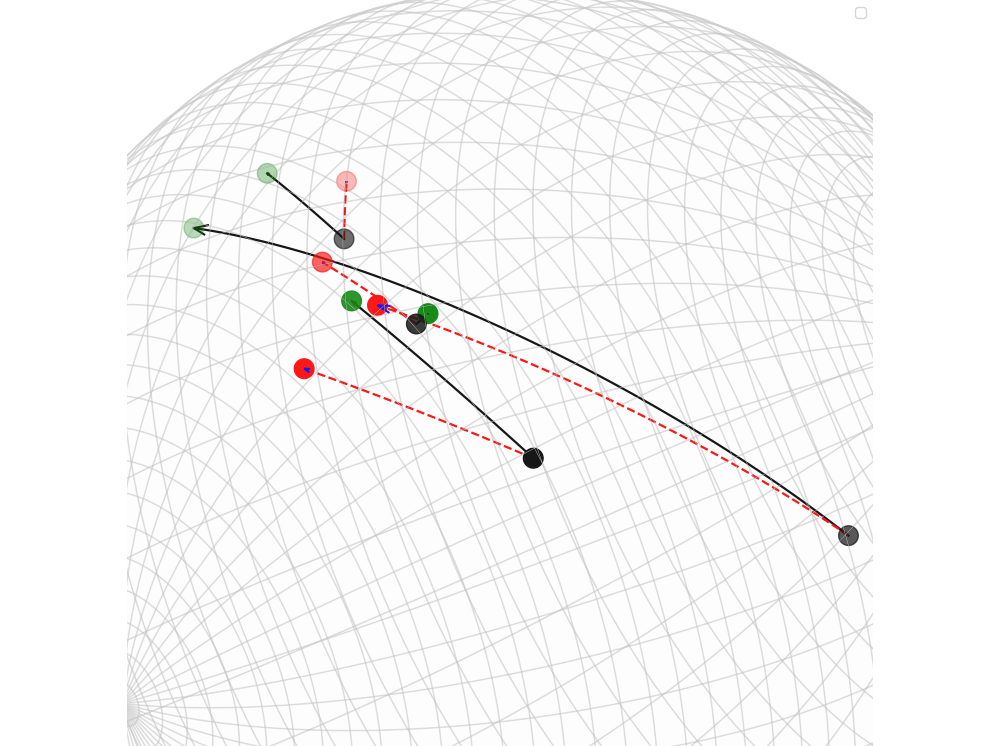}
}
\hfill
\subfloat[]{
    \includegraphics[width=0.3\textwidth]{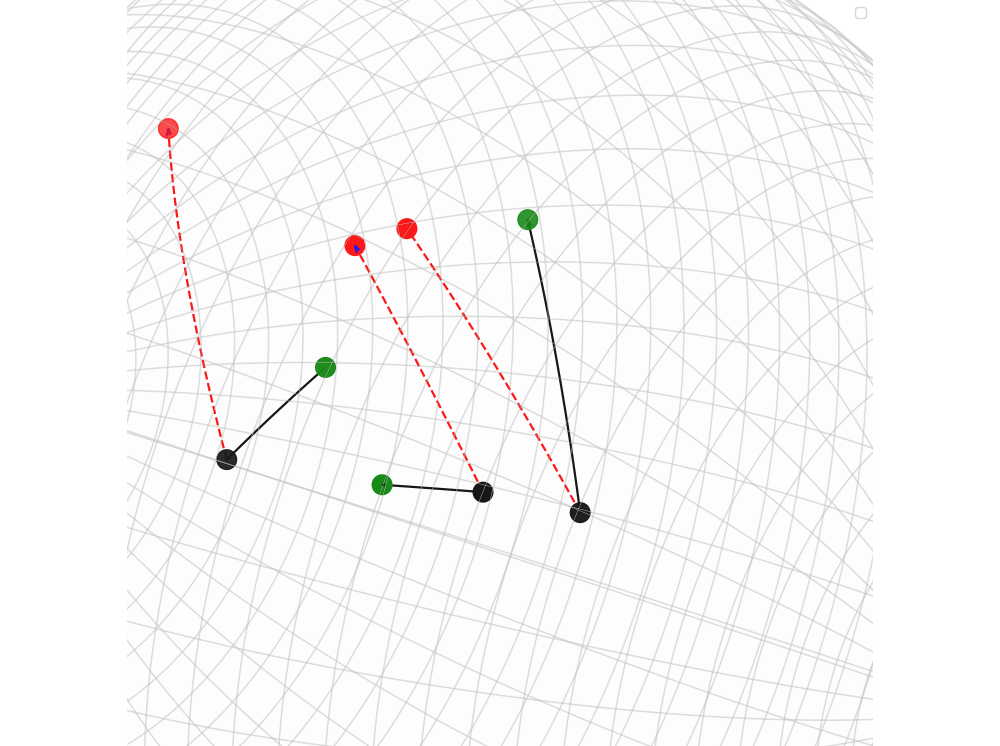}
}
\hfill
\subfloat[]{
    \includegraphics[width=0.3\textwidth]{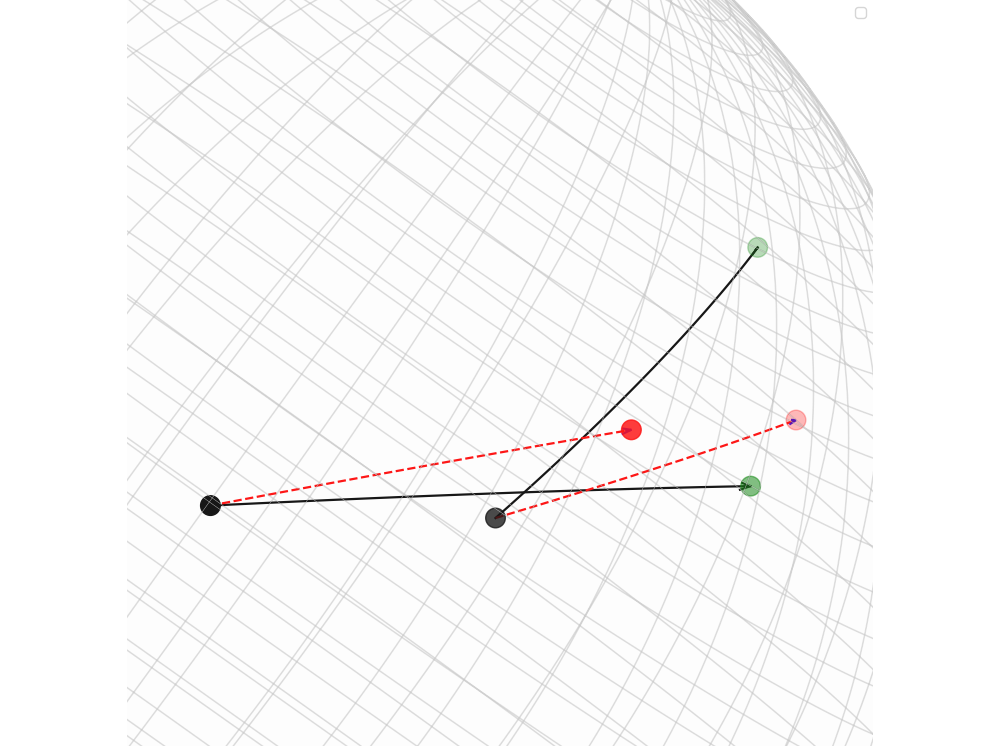}
}

\caption{
Randomly selected covariate-response pairs with their predicted mean responses for the Vector-Cardiogram dataset. Black dots denote covariates, green dots denote responses, and red dots denote predicted mean responses. Solid lines connect covariates to responses, while dashed lines connect covariates to predicted mean responses.
}
\label{fig:cardio_fit}

\end{figure}

\subsection{Computational Time}
We report the computational times of the MCMC algorithms for the OT-SOS-Rotation and OT-SOS-PLT models across the two datasets, using a Snapdragon\textsuperscript{\textregistered} X Elite X1E80100 processor with a Qualcomm\textsuperscript{\textregistered} Oryon\texttrademark\ CPU. For each dataset, the reported times correspond to the model fitted with the selected optimal number of atoms. The results are shown in Table \ref{tab_comp_time}. As expected, computational cost increases with both the number of atoms $k$ and the sample size $n$.

\begin{table}[ht]
\small
    \centering
    \caption{Computational Time of MCMC run with 2000 iterations}
    \label{tab_comp_time}
    \begin{tabular}{lccc}
        \toprule
    \textbf{Dataset} &    \textbf{Model} & \textbf{Number of Atoms} & \textbf{Computational Time (mins)}  \\
        \midrule
       Cyclone ($n=1107$) & OT-SOS-Rotation & 14 & 193.83  \\
       & OT-SOS-PLT & 10 & 181.36  \\
         Vector-Cardiogram ($n=53$) & OT-SOS-Rotation & 3 & 2.57 \\
         & OT-SOS-PLT & 2 & 1.72 \\
        \bottomrule
    \end{tabular}
\end{table}

\section{Discussion}
\label{sec_discussion}
In this work, we proposed an optimal transport–based approach for partitioning the spherical domain, enabling flexible modeling for the problem of sphere-on-sphere regression. Our framework leverages partitioning to define locally adaptive regression models, specifically using rotation-based transformations and projective linear transformations within each region. Together with full Bayesian inference, our approach yields coherent uncertainty quantification for both the induced segmentation and predictions; moreover, the near-parametric posterior contraction rates established here and the empirical performance in simulations and applications illustrate that the method can match or improve upon global parametric specifications while remaining competitive with flexible nonparametric alternatives.
\\\\
The proposed method has some limitations. First, computational cost increases with sample size $n$ and the number of atoms $k$, as each MCMC iteration requires solving an entropic-regularized optimal transport problem. Second, the entropic regularization parameter $\omega$ introduces a trade-off between computational efficiency and approximation accuracy to the semi-discrete solution. Third, model selection for $k$ currently relies on information criteria rather than formal trans-dimensional inference, though the WAIC-based approach performed well in our studies.
\\\\
Several extensions of this work are possible. Alternative local models could be employed within each partition region, including nonlinear mappings or other Lie group--based transformations. The partitioning framework naturally extends to related settings, such as spherical covariates with Euclidean responses or regression between spheres of different dimensions. Finally, connections to Bayesian clustering and mixture models suggest potential for developing joint inference procedures that simultaneously learn partitions and local regression structures.

\newpage
\appendix

\begin{center}
 \Large Supplementary Material for ``Bayesian Inference for Sphere-on-Sphere Regression with Optimal Transport Map''
\end{center}

\section{Notations}
For $1 \le r \le \infty$, we let $||\cdot||_{L^r({\cal X})}$ denote the $L^r$-norm of a measurable function (with respect to the Lebesgue measure) on a domain ${\cal X}$. We use the notation $a \lesssim b$ and $a \gtrsim b$ to denote inequalities up to a positive
multiplicative constant, and $a \asymp b$ if $a \lesssim b$ and $b \lesssim a$. We denote ${\cal F}$ as the space of regression maps from the sphere $\mathbb{S}^p$ to itself, that is, ${\cal F} := \{f: \mathbb{S}^p \rightarrow \mathbb{S}^p \}$. For each $f \in {\cal F}$, we let $\mathbb{P}_f$ denote the joint distribution of a covariate-response pair $(x,y)$ where the covariate $x$ is uniformly distributed on $\mathbb{S}^p$. Note that although the assumptions state that the covariate \(x\) follows a distribution \(\mu\) whose density is bounded away from \(0\) and \(\infty\), there is no loss of generality in assuming throughout the proof that \(\mu\) is the uniform measure. We let $\Pi(\cdot)$ denote the prior distribution on ${\cal F}$ and, given independently and identically distributed (i.i.d.) observations ${\cal D}_n := \{(x_i, y_i)\}_{i=1}^{n}$ from $\mathbb{P}_f$ for some $f \in {\cal F}$, we let $\Pi(\cdot|{\cal D}_n)$ denote the posterior distribution on ${\cal F}$. Given a function $h: \mathbb{S}^p \rightarrow \mathbb{R}$, we define $\mathbb{P}_n h := \frac{1}{n} \sum_{i=1}^{n} h(x_i)$ for i.i.d. uniformly distributed points $\{x_i\}_{i=1}^{n}$. Let $\tilde{d}$ denote a semi-metric on ${\cal F}$, and let $f \in {\cal F}$ be an arbitrary regression map; we define the $\epsilon$-ball around  map $f$ with respect to the semi-metric $\tilde{d}$ as 
$$ B(f; \epsilon, \tilde{d}) = \{f' \in {\cal F}: \tilde{d}(f', f) \le \epsilon\}.$$
When the underlying semi-metric is obvious from context, we use the simplified notation $B(f; \epsilon)$ instead. For any $\epsilon > 0$ and any set $A \subset {\cal F}$, the $\epsilon$-covering number of $(A, \tilde{d})$, denoted by ${\cal N}(\epsilon, A, \tilde{d})$, is defined to be the minimum number of $\epsilon$-balls of the form $B(f; \epsilon, \tilde{d})$ that are needed to cover $A$.

\section{Assumptions for Posterior Contraction Rate Result}
\label{sec_assumptions}
\begin{assumption}
    \label{source_measure_assumption}
The absolutely continuous source probability measure $\mu$ on $\mathbb{S}^p$ is known, with a probability density function that is bounded above and below by positive constants.
\end{assumption}

\begin{assumption}
\label{true_reg_map_assumption}
The true regression map $f_0$ is determined by source measure $\mu$, a discrete target measure
\[
\nu_0 = \sum_{j=1}^{k_0} p_{0,j} \delta_{z_{0,j}},
\]
and rotations $R_{0,j}$ for $j = 1, \ldots, k_0$. The pair $(\mu, \nu_0)$ induces a Laguerre partition of the sphere,
\[
\mathbb{S}^p = \bigcup_{j=1}^{k_0} {\cal A}_{0,j},
\]
with $f_0$ acting piecewise as
\[
f_0(x) = R_{0,j} x, \qquad x \in {\cal A}_{0,j}.
\]
Thus, on each cell ${\cal A}_{0,j}$, $f_0$ is a rigid rotation.
\end{assumption}

\begin{assumption}    
    \label{prior_k_assumption}
The prior distribution $\Pi_k$ is supported on the set of integers $\{1, \ldots, K_{\text{max}}\}$ where $K_\text{max} \ge k_0$ and has probability mass function $\pi_k$ satisfying 
\begin{eqnarray}
    \pi_k(\tilde{k}) > 0, \quad \tilde{k} = 1,\ldots, K_\text{max}. 
\end{eqnarray} 
\end{assumption}

\begin{assumption}
\label{prior_measure_assumption}
Conditional on the number of atoms $k$, the atoms $\{z_1, z_2, \ldots, z_k\}$ of the measure $\nu$ are uniformly distributed on $\mathbb{S}^p$, subject to the constraint 
\begin{eqnarray}
\label{eqn_prior_atoms_constraint}
 \min_{j_1 \ne j_2} \big\{ d(z_{j_1}, z_{j_2}) \big\} \ge \eta_n ,
\end{eqnarray}
for a positive decreasing sequence $\eta_n > 0$,satisfying
\[  \eta_n \asymp \epsilon_n^{\alpha} \]
for $\alpha \in (0, \frac{1}{9}) $
where $\epsilon_n \asymp n^{-1/2} \log n$.
The probability vector $(p_1, \ldots, p_k)$ corresponding to the atoms have a Dirichlet distribution $\mbox{Dir}(\alpha_p, \ldots, \alpha_p)$ for some $\alpha_p < 1$.
\end{assumption}

\begin{remark}
   Assumption \ref{prior_measure_assumption} implies that the prior on the discrete measure, and consequently on the regression map, depends on the sample size $n$. Strictly speaking, we should therefore write $\Pi_n$ instead of $\Pi$ for the prior distribution. However, for notational simplicity, we continue to write $\Pi$ and suppress the dependence on $n$.
\end{remark}

\begin{remark}
   The minimum distance constraint between atoms \eqref{eqn_prior_atoms_constraint} is related to repulsive mixture models in the literature \citep{Petralia2012, Xie2020, Huang2025} and the non-local prior for choosing mixture components \citep{Fuquene2019}. In mixture modeling, a common issue arises when independent priors for the location parameters of mixture components lead to potentially redundant components. 
\end{remark}

\begin{assumption}
\label{prior_rotation_assumption}
Conditional on the number of atoms $k$, the rigid rotation matrices $R_1,\ldots,R_k$ are uniformly distributed with respect to the Haar measure on $\mathbb{SO}(p+1)$.
\end{assumption}

\section{Proof of Main Results}
\label{sec_proof_main}
In this section, we derive a general result on the contraction rate for Bayesian regression models on the sphere. We assume that the model takes the form
\begin{align}
Y_i \mid (X_i = x_i) &\stackrel{\text{i.i.d.}}{\sim} g_p(y_i; f(x_i), \kappa), \quad i = 1, \ldots, n,
\end{align}
and that the true data-generating model is
\begin{align}
Y_i \mid (X_i = x_i) &\stackrel{\text{i.i.d.}}{\sim} g_p(y_i; f_0(x_i), \kappa), \quad i = 1, \ldots, n,
\end{align}
where $f_0 : \mathbb{S}^p \to \mathbb{S}^p$ is the unknown true regression map.

\begin{theorem}
\label{general_thm}
Assume the design points $(x_i)_{i=1}^n$ are i.i.d. distributed on $\mathbb{S}^p$ from a density that is bounded away from 0 and $\infty$. Assume there exists sequences $(\epsilon_n)_{n=1}^{\infty}$, $(\underline{\epsilon}_n)_{n=1}^{\infty}$ such that $0 \le \underline{\epsilon}_n \le \epsilon_n \rightarrow 0$ and $\min(n \epsilon_n^2, n \underline{\epsilon}_n^2) \rightarrow \infty$ as $n \rightarrow \infty$. Let $({\cal F}_n)_{n=1}^{\infty}$ be ``sieves'' for ${\cal F}$, that is, a sequence of sets ${\cal F}_n \subset {\cal F}$. Assume for some constants $D_1 > 0$ and for sufficiently large $M$, the following three conditions hold.
\begin{condition}
\label{cond_covering}
    
\begin{eqnarray}
\sum_{j=M}^{\infty} N_{nj} \exp(-D_1 n j^2 \epsilon_n^2) \rightarrow 0 
\end{eqnarray}
as $n \rightarrow \infty$, where $N_{nj} = {\cal N}(\xi j \epsilon_n, {\cal S}_{nj}(\epsilon_n), \tilde{d})$, for  $\xi = \frac{1}{32}$, is the covering number of the set
$$ {\cal S}_{nj}(\epsilon_n) = \{f \in {\cal F}_n: j \epsilon_n < \tilde{d}(f, f_0) \le (j+1) \epsilon_n \} .$$
\end{condition}
\begin{condition}
  \label{cond_remain_mass}
For sufficiently large $n$, 
    \begin{eqnarray}
    \Pi({\cal F}_n^{c}) \lesssim \exp(-3D_1 n \underline{\epsilon}_n^2).
\end{eqnarray}
\end{condition}

\begin{condition}
  \label{cond_prior_mass}
  For sufficiently large $n$,
    \begin{eqnarray}
    \Pi(B_n(f_0; \underline{\epsilon}_n)) \ge \exp(-D_1 n \underline{\epsilon}_n^2) .
\end{eqnarray}
\end{condition}

\noindent Then,
$$ \mathbb{E}_0 \big( \Pi( \tilde{d}(f,f_0) > M \epsilon_n \mid {\cal D}_n) \big) \rightarrow 0$$
as $n \rightarrow \infty$. 
\end{theorem}
The three conditions are analogous to those needed for existing results on posterior contraction rates \citep{Ghosal2000, Ghosal2007}. Condition \ref{cond_covering}~corresponds to the entropy condition, stipulating that the sieve space ${\cal F}_n$ should not be excessively large. Condition \ref{cond_remain_mass}, combined with Condition \ref{cond_covering},~asserts that ${\cal F}_n^{c}$, the complement of the sieve space, should attract minimal prior probability mass. Condition~\ref{cond_prior_mass}. requires that the prior distribution puts a sufficient amount of probability mass around the true map $f_0$. 
\\\\
Theorem \ref{general_thm} could likely be used for studying contraction rates of posterior distributions for other Bayesian sphere-on-sphere regression models. Similar to \cite{Ghosal2000, Ghosal2007,Xie2019}, the proof of Theorem \ref{general_thm}, which we give in Appendix~\ref{sec:Proof_Thm1} in the Supplementary Material, involves constructing a suitable test statistic and obtaining exponentially small Type I and Type II error probability bounds. The proof, along with the construction of the test statistic, builds on the approach of \cite{Xie2019} and adapts it to the sphere context.

\subsection{Proof of Theorem \ref{general_thm}}\label{sec:Proof_Thm1}
We present several preliminary lemmas that will be utilized in the proof of Theorem \ref{general_thm}, and the proof of these lemmas are presented in Appendix \ref{sec_proof_auxiliary}.
\\\\
As in the case of \cite{Xie2019}, by our assumption that $(x_i)_{i=1}^n$ are i.i.d. distributed from some density bounded away from 0 and $\infty$, there is no loss of generality to work under the assumption that $(x_i)_{i=1}^n$ are uniformly distributed on $\mathbb{S}^p$.
We first define the semi-metric $\hat{d}$ on the space of regression maps ${\cal F} = \{f: \mathbb{S}^p \rightarrow \mathbb{S}^p \}$ as
$$ \hat{d}^2(f_0, f_1) = \int_{\mathbb{S}^p} (1 - f_0(x)^T f_1(x) )  \intd\rho(x) .$$ We first show that the $\hat{d}$ and $\tilde{d}$ are equivalent. 
\begin{lemma}
\label{lemma_equi_metric}
For all $f_0, f_1: \mathbb{S}^p \rightarrow \mathbb{S}^p$, we have
$$ \tilde{d}(f_0, f_1) \lesssim \hat{d}(f_0, f_1) \lesssim \tilde{d}(f_0, f_1) .$$
\end{lemma}

Lemma \ref{lemma_equi_metric} allows us to work with the semi-metric $\hat{d}$ in instead of $\tilde{d}$. We have the following lemma concerning the conditional expectation of response $Y$ given covariate $x$. Recall that the conditional distribution of $Y$ given covariate $x$ follows the vMF distribution with mean direction $f_0(x)$ and concentration parameter $\kappa$.
\begin{lemma}
\label{lemma_cond_direction}
Let $f_0: \mathbb{S}^{p} \rightarrow \mathbb{S}^{p}$ be the true regression function. Conditional on $X=x$, We have
$$ \mathbb{E}_0(Y \mid X=x) = A_p(\kappa) f_0(x) ,$$
where $$ A_p(\kappa) = \frac{I_{p/2}(\kappa)}{I_{p/2-1}(\kappa)},$$
and where $I_{v}(\cdot)$ is the modified Bessel function of the first kind of order $v$.
\end{lemma}
Let $({\cal F}_n)_{n=1}^{\infty}$ be the sieves for ${\cal F}$. We first establish the following local test lemma.

\begin{lemma}
\label{lemma_exist_test}
Let $f_1 \in {\cal F}_n$ be a regression map that satisfies $\hat{d}(f_0, f_1) > 0$. For each $n=1,2,\ldots$, there exists a test functions $\phi_n: ( \mathbb{S}^p \times \mathbb{S}^p )^n \rightarrow [0,1]$ such that
\begin{align*}
    \mathbb{E}_0 \phi_n &\le \exp(-C n \hat{d}^2(f_0, f_1)), \\
    \sup_{\{ f \in {\cal F}_n: \hat{d}(f,f_1) \le \xi \hat{d}(f_0,f_1) \} } \mathbb{E}_f(1 - \phi_n) &\le 2 \exp\bigg( - C \frac{n \hat{d}^2(f_0, f_1) }{ \eta + \hat{d}^2(f_0, f_1) } \bigg) + \exp(-C n \hat{d}^2(f_0,f_1)),
\end{align*}
for some constant $C > 0$ and $\xi \in (0,1)$.
\end{lemma}
We next extend the local test lemma to the following global test lemma.

\label{sec_gloabl_test}
\begin{lemma}
\label{lemma_global_test}
Let $(\epsilon_n)_{n=1}^{\infty}$ be a sequence with $\epsilon_n \rightarrow 0$ and $n \epsilon_n^2 \rightarrow \infty$. For any $M>0$, there exists a sequence of tests $(\phi_n)_{n=1}^{\infty}$ where $\phi_n: (\mathbb{S}^p \times \mathbb{S}^p)^n \rightarrow [0,1]$ such that for some constants and $C>0, \eta > 0$
\begin{align*}
    \mathbb{E}_0 \phi_n  &\le \sum_{j=M}^{\infty} N_{nj} \exp(-C n j^2 \epsilon_n^2 ),\\
 \sup_{ \{f \in {\cal F}_n: \hat{d}(f, f_0) > M \epsilon_n \} } \mathbb{E}_f(1 - \phi_n)
    & \le  \exp(-C M^2 n \epsilon_n^2) + 2 \exp\bigg( - C \frac{n M^2 \epsilon_n^2}{M^2 \epsilon_n^2 + \eta}   \bigg) ,
\end{align*}
where we recall that $N_{nj} = {\cal N}(\xi j \epsilon_n, {\cal S}_{nj}(\epsilon_n), \tilde{d})$ is the covering number of the set
$$ {\cal S}_{nj}(\epsilon_n) = \{f \in {\cal F}_n: j \epsilon_n < \tilde{d}(f, f_0) \le (j+1) \epsilon_n \} ,$$

\end{lemma}
We also need the following lemma concerning the exponential lower bound of $\int_{{\cal F}}\prod_{i=1}^{n} \frac{p_f(x_i, y_i)}{p_0(x_i,y_i)} \Pi(\intd f) $.

\begin{lemma}
\label{lemma_exp_lower_bnd}
 Let $(\epsilon_n)_{n=1}^{\infty}$ be a sequence with $\epsilon_n \rightarrow 0, n \epsilon_n^2 \rightarrow \infty$, we have
 \begin{eqnarray*}
    \mathbb{P}_0 \bigg( \int_{{\cal F}} \prod_{i=1}^{n} \frac{p_f(x_i, y_i)}{p_0(x_i,y_i)} \Pi(\mathrm{d} f) \le \Pi(B(f_0; \epsilon_n)) \exp(-C n \epsilon_n^2) \bigg) \rightarrow 0
\end{eqnarray*}
for all sufficiently large constant $C > 0$. 
\end{lemma}

We are now in a position to prove Theorem \ref{general_thm}.
\begin{proof}[Proof of Theorem \ref{general_thm}]
Define the log likelihood ratio function
$$ \Lambda_n(f|{\cal D}_n) := \sum_{i=1}^{n} ( \log p_f(x_i, y_i) - \log p_0(x_i, y_i) ),$$
and the event
$$ {\cal H}_n := \bigg\{ \int_{{\cal F}} \exp( \Lambda_n(f \mid {\cal D}_n) ) \Pi(\intd f) \ge \exp(- 2 D_1 n \underline{\epsilon}_n^2 ) \bigg\} .$$
We have
$$ {\cal H}_n^{c} \subset \bigg\{ \int_{{\cal F}} \exp( \Lambda_n(f \mid {\cal D}_n) ) \Pi(\intd f) < \Pi(B_n(f_0, \underline{\epsilon}_n)) \exp(- D_1 n \underline{\epsilon}_n^2) \bigg\} .$$
We have by Lemma \ref{lemma_exp_lower_bnd} that
$$ \mathbb{P}_0({\cal H}_n^{c}) \rightarrow 0.$$
We now apply Lemma \ref{lemma_global_test}
to obtain the global test function $\phi_n$. We have that
\begin{eqnarray*}
    \mathbb{E}_0 \big( \Pi( \hat{d}(f,f_0) > M \epsilon_n | {\cal D}_n) \big) &\le& \mathbb{E}_0 \bigg( (1- \phi_n) 1({\cal H}_n) \Pi( \hat{d}(f,f_0) > M \epsilon_n | {\cal D}_n) \bigg)\\
    && + \mathbb{E}_0 \phi_n + \mathbb{E}_0 \big( (1-\phi_n) 1({\cal H}_n^c) \big) \\
    & \le & \mathbb{E}_0 \bigg( (1- \phi_n) 1({\cal H}_n) \frac{ \int_{ \hat{d}(f,f_0) > M \epsilon_n } \exp(\Lambda_n(f|{\cal D}_n)) \Pi(\intd f) }{ \int_{{\cal F}} \exp(\Lambda_n(f|{\cal D}_n)) \Pi(\intd f) } \bigg) \\
    && + \mathbb{E}_0 \phi_n + \mathbb{P}_0({\cal H}_n^c) ,
\end{eqnarray*}
and
$$ \mathbb{E}_0 \phi_n \le \sum_{j=M}^{\infty} N_{nj} \exp(-D_1 n j^2 \epsilon_n^2) \rightarrow 0 ,$$
where the convergence above follows from Condition \eqref{cond_covering}. 
\\\\
On ${\cal H}_n$, by Lemma \ref{lemma_exp_lower_bnd} we can lower bound 
$$ \int_{{\cal F}} \exp(\Lambda_n(f|{\cal D}_n)) \Pi(\intd f)$$
to obtain
\begin{eqnarray*}
    && \mathbb{E}_0 \bigg( (1- \phi_n) 1({\cal H}_n) \frac{ \int_{ \hat{d}(f,f_0) > M \epsilon_n} \exp(\Lambda_n(f|{\cal D}_n)) \Pi(\intd f) }{ \int_{{\cal F}} \exp(\Lambda_n(f|{\cal D}_n)) \Pi(\intd f) } \bigg) \\
    &\le& \exp(2 D_1 n \underline{\epsilon}_n^2) \mathbb{E}_0 \bigg( (1- \phi_n) \int_{f \in {\cal F}_n: \hat{d}(f,f_0) > M \epsilon_n} \exp(\Lambda_n(f|{\cal D}_n)) \Pi(\intd f) \bigg) \\
    && + \exp(2 D_1 n \underline{\epsilon}_n^2) \mathbb{E}_0 \bigg( \int_{f \in {\cal F}_n^c} \exp(\Lambda_n(f|{\cal D}_n) \Pi(\intd f) \bigg).
\end{eqnarray*}
Using Fubini's theorem and Lemma \ref{lemma_global_test}, we have
\begin{eqnarray*}
   && \mathbb{E}_0 \bigg( (1- \phi_n) \int_{f \in {\cal F}_n: \hat{d}(f,f_0) > M \epsilon_n} \exp(\Lambda_n(f|{\cal D}_n)) \Pi(\intd f) \bigg) \\
   &= & \int_{f \in {\cal F}_n: \hat{d}(f,f_0) > M \epsilon_n} \mathbb{E}_0 \big( (1-\phi_n) \exp(\Lambda_n(f|{\cal D}_n) \big) \Pi(\intd f) \\
   & \le & \sup_{ f \in {\cal F}_n: \hat{d}(f, f_0) > M \epsilon_n} \mathbb{E}_f (1 - \phi_n) \\
   & \le & \exp(-C M^2 n \epsilon_n^2) + 2 \exp\bigg( - C \frac{n M^2 \epsilon_n^2}{M^2 \epsilon_n^2 + \eta}   \bigg) \\
   & \le & \exp(- \tilde{C} M^2 n \epsilon_n^2),
\end{eqnarray*}
for some constant $\tilde{C}>0$ and for all sufficiently large $n$.
\\\\
We apply Fubini's theorem once more to obtain that
\begin{eqnarray*}
    \mathbb{E}_0 \bigg( \int_{f \in {\cal F}_n^c} \exp(\Lambda_n(f|{\cal D}_n) \Pi(\intd f) \bigg) &=& \int_{f \in {\cal F}_n^c} \mathbb{E}_0 \exp(\Lambda_n(f|{\cal D}_n)) \Pi(\intd f) \\
    & = & \Pi({\cal F}_n^c) \lesssim \exp(-3 D_1 n \underline{\epsilon}_n^2)
\end{eqnarray*}
by Condition \eqref{cond_remain_mass}. Putting it together we have
\begin{eqnarray*}
   &&  \mathbb{E}_0 \bigg( (1- \phi_n) 1({\cal H}_n) \frac{ \int_{ \hat{d}(f,f_0) > M \epsilon_n} \exp(\Lambda_n(f|{\cal D}_n)) \Pi(\intd f) }{ \int_{{\cal F}} \exp(\Lambda_n(f|{\cal D}_n)) \Pi(\intd f) } \bigg) \\
    &\le& \exp(2D_1 n \underline{\epsilon}_n^2 - \tilde{C} M^2 n \epsilon_n^2) + \exp(2D_1 n \underline{\epsilon}_n^2 - 3D_1 n \underline{\epsilon}_n^2) \rightarrow 0
\end{eqnarray*}
as $n \rightarrow \infty$, for large enough $M > 0$. Therefore, we have
$$ \mathbb{E}_0 \big( \Pi( \hat{d}(f,f_0) > M \epsilon_n | {\cal D}_n) \big) \rightarrow 0,$$
as $n \rightarrow \infty$.

\end{proof}

\subsection{Proof of Contraction Rate Result for OT-SOS-Rotation Model}

A key ingredient in establishing the contraction rates is the following stability result for optimal transport maps with respect to perturbations of the target measure.

\begin{lemma}
\label{lemma_stability}
Let $\mu$ be an absolutely continuous measure. Given two probability measures $\nu_1,\nu_2$ on $\mathbb{S}^p$, and let $S_{\nu_1}$ and $S_{\nu_2}$ denote the optimal transport maps from $\mu$ to $\nu_1$ and $\nu_2$, respectively, under the squared geodesic distance cost. Then
\begin{eqnarray}
\label{eqn_quantitative_stability}
   \left( \int_{\mathbb{S}^{p}} d^2(S_{\nu_1}(x), S_{\nu_2}(x)) \, \mathrm{d}\mu(x) \right)^{1/2} \lesssim W_1(\nu_1,\nu_2)^{\tau} 
\end{eqnarray}
for $\tau=1/9$.
\end{lemma}

\begin{remark}
    After we uploaded our preprint to arXiv, \cite{Kitigawa2025} posted a preprint establishing a quantitative stability result in the more general setting of Riemannian manifolds, obtaining an exponent of $\tau = \frac{1}{6}$ for inequality \eqref{eqn_quantitative_stability}. This improved exponent in \cite{Kitigawa2025} only slightly relaxes the atom separation condition in our Assumption \ref{prior_measure_assumption}, which is inconsequential in practice, and importantly does not affect our contraction rate. Our proof of Lemma \ref{lemma_stability}, which is specific to the sphere, may also be of independent interest.
\end{remark}

Lemma \ref{lemma_stability} allows us to control the distance between two optimal transport maps through the 1-Wasserstein distance between the corresponding target measures. The integrated distance between two optimal transport maps is relative to the source measure $\mu$. By Assumption \ref{source_measure_assumption}, we can have the following Corollary where the integrated distance between two optimal transport maps is relative to the uniform measure $\rho$. 

\begin{corollary}
\label{cor_stability}
Let $\mu$ be an absolutely continuous measure satisfies Assumption \ref{source_measure_assumption}. Given two probability measures $\nu_1,\nu_2$ on $\mathbb{S}^p$, and let $S_{\nu_1}$ and $S_{\nu_2}$ denote the optimal transport maps from $\mu$ to $\nu_1$ and $\nu_2$, respectively, under the squared geodesic distance cost. Then
\[ \tilde{d}(S_{\nu_1}, S_{\nu_2}) = \left( \int_{\mathbb{S}^{p}} d^2(S_{\nu_1}(x), S_{\nu_2}(x)) \, \mathrm{d}\rho(x) \right)^{1/2} \lesssim W_1(\nu_1,\nu_2)^{\tau} \]
for $\tau=1/9$.
\end{corollary}
The following result provides a lower bound on the small ball probability for a Dirichlet distribution.

\begin{lemma}
\label{lemma_dirichlet}
Suppose $(p_1, \ldots, p_k) \sim \mathrm{Dir}(\alpha)$ with $\alpha < 1$. Let $(p_{0,1}, \ldots, p_{0,k})$ be any point on the $k$-simplex such that $p_{0,j} > 0$ for all $j = 1, \ldots, k$. Then, for any $0 < \epsilon < 1$, 
\[
\mathbb{P}\Big( |p_j - p_{0,j}| < \epsilon, \; \forall j = 1, \ldots, k \Big) 
\ge 
\frac{\Gamma(\alpha k)}{\Gamma(\alpha)^k} \left( \frac{\epsilon}{2} \right)^{2k}.
\]
\end{lemma}
\begin{proof}
    By Lemma 24 of \cite{Liu2023}, we have
$$ \mathbb{P} \bigg( \sum_{j=1}^{k} |p_j - p_{0,j}| < \epsilon  \bigg) \ge \frac{\Gamma(\alpha k)}{(\Gamma(\alpha))^k} \bigg( \frac{\epsilon}{2 } \bigg)^{2k} . $$
Since 
$$ \mathbb{P}\bigg( |p_j - p_{0,j}| < \epsilon, \forall j=1,\ldots,k \bigg) \ge  \mathbb{P} \bigg( \sum_{j=1}^{k} |p_j - p_{0,j}| < \epsilon  \bigg) ,$$
the result follows.
\end{proof}

\medskip

\begin{proof}[Proof of Theorem \ref{thm_contraction_rates_opt}]
The proof proceeds in several steps. We first construct suitable sieves and subsequently verify the three conditions of Theorem \ref{general_thm}.
\\\\
\textbf{Construction of the sieves}
\\
Let $(\eta_n)$ be a sequence satisfying $\eta_n \asymp \epsilon_n^{\alpha}$ for some $\alpha\in(0,1/9)$ and recall $K_{\max}$ is the maximal number of atoms. Define 
\[
{\cal V}_n :=
\Big\{
\nu=\sum_{j=1}^k p_j\delta_{z_j} :
p_j>0,\,
\sum_{j=1}^k p_j=1,\,
z_j\in\mathbb S^p,\,
\min_{j\neq l} d(z_j,z_l)\ge \eta_n,\,
k\le K_{\max}
\Big\}.
\]
Each $\nu\in{\cal V}_n$, together with rotations $R_1,\dots,R_k\in\mathbb{SO}(p+1)$, determines a regression map
$f:\mathbb S^p\to\mathbb S^p$. We define the corresponding class of mappings as
\[
{\cal F}_n :=
\Big\{
f:\mathbb S^p\to\mathbb S^p :
\nu\in{\cal V}_n,\,
R_1,\dots,R_k\in\mathbb{SO}(p+1)
\Big\}.
\]
For $k=1,\dots,K_{\max}$ define
\[
{\cal V}_{n,k} :=
\Big\{
\nu=\sum_{j=1}^k p_j\delta_{z_j}\in{\cal V}_n
\Big\},
\qquad
{\cal F}_{n,k} :=
\Big\{
f:\nu\in{\cal V}_{n,k},\,R_1,\dots,R_k\in\mathbb{SO}(p+1)
\Big\}.
\]
Then
\[
{\cal F}_n=\bigcup_{k=1}^{K_{\max}}{\cal F}_{n,k}.
\]
By construction, the remaining mass condition (Condition \ref{cond_remain_mass} of Theorem \ref{general_thm}) is trivially satisfied, since the prior assigns zero mass outside the sieve, that is, $\Pi(\mathcal{F}_n^c) = 0$. It therefore suffices to verify Conditions \ref{cond_covering} and \ref{cond_prior_mass}.
\\\\
\textbf{Verification of Condition \ref{cond_covering}}
\\
We first use the following straightforward upper bound on the covering number $N_{nj}$:

\begin{eqnarray*}
    N_{nj} = {\cal N}(\xi j \epsilon_n, {\cal S}_{nj}(\epsilon_n), \tilde{d} ) &\le& {\cal N}(\xi \epsilon_n, {\cal F}_n, \tilde{d}) .
\end{eqnarray*}
where $\xi \in (0,1)$ is some fixed constant.
\\\\
Consider $f_1,f_2\in{\cal F}_{n,k}$ with corresponding measures
\[
\nu_1=\sum_{j=1}^k p_{1,j}\delta_{z_{1,j}},
\qquad
\nu_2=\sum_{j=1}^k p_{2,j}\delta_{z_{2,j}} .
\]
It is straightforward to show that if two discrete measures have sufficiently small differences between their corresponding atom locations and associated probabilities,
\[
d(z_{1,j}, z_{2,j}) \le \delta_n/k,
\qquad
|p_{1,j} - p_{2,j}| \le \delta_n/k,
\]
for $j = 1, \dots, k$, then
\[
W_1(\nu_1, \nu_2) \lesssim \delta_n .
\]
Let $S_{\nu_1}$ and $S_{\nu_2}$ denote the optimal transport maps from $\mu$ to $\nu_1$ and $\nu_2$. By Corollary \ref{cor_stability}, one has
\[
\tilde d(S_{\nu_1},S_{\nu_2})
\lesssim
W_1(\nu_1,\nu_2)^{\tau}
\lesssim
\delta_n^{\tau},
\]
where $\tau = 1/9$.
\\\\
Define the mismatch region
\[
S_{\nu_1,\nu_2}
=
\{x\in\mathbb S^p:
S_{\nu_1}(x)=z_{1,j_1},
\,
S_{\nu_2}(x)=z_{2,j_2},
\,
j_1\neq j_2
\}.
\]
Because atoms of $\nu_1$ (and resp. $\nu_2$) are separated by at least $\eta_n$, we have the following inclusion of the mismatch region:
\[
S_{\nu_1,\nu_2}
\subset
\{x:d(S_{\nu_1}(x),S_{\nu_2}(x))>\eta_n/2\}.
\]

Using Markov's inequality and Assumption \ref{source_measure_assumption},
\[
\rho( S_{\nu_1,\nu_2}) \lesssim \mu(S_{\nu_1,\nu_2})
\lesssim
\frac{\tilde d(S_{\nu_1},S_{\nu_2})^2}{\eta_n^2}
\lesssim
\delta_n^{2(\tau-\alpha)}.
\]

Splitting the integral,
\[
\tilde d^2(f_1,f_2)
\lesssim
\rho(S_{\nu_1,\nu_2})
+
\int_{S_{\nu_1,\nu_2}^c}
(1-f_1(x)^T f_2(x))\,\intd\rho(x).
\]
On the set $S_{\nu_1,\nu_2}^c$, the two maps correspond to the same atom index $j$, hence
\[
f_1(x)=R_{1,j}x,
\qquad
f_2(x)=R_{2,j}x .
\]
Using the following inequality
\[
1-(R_{1,j}x)^T(R_{2,j}x)
=
\frac12\|R_{1,j}x-R_{2,j}x\|^2
\le
\frac12\|R_{1,j}-R_{2,j}\|_{op}^2,
\]
we obtain
\[
\tilde d^2(f_1,f_2)
\lesssim
\delta_n^{2(\tau-\alpha)}
+
\max_{j\le k}
\|R_{1,j}-R_{2,j}\|_{op}^2 .
\]
Thus if
\[
d(z_{1,j},z_{2,j})\lesssim \delta_n,
\quad
|p_{1,j}-p_{2,j}|\lesssim \delta_n,
\quad
\|R_{1,j}-R_{2,j}\|_{op}\lesssim \delta_n^{(\tau - \alpha)}
\]
for $j=1,\dots,k$, with
\[
\delta_n=\epsilon_n^{1/(\tau-\alpha)},
\]
then
\[
\tilde d(f_1,f_2)\lesssim \epsilon_n .
\]
\\\\
From the above analysis, we obtain the following upper bound on the covering number of ${\cal F}_{n,k}$:
\[
{\cal N}(\epsilon_n,{\cal F}_{n,k},\tilde d)
\lesssim
{\cal N}(\delta,\mathbb S^p,d)^k
{\cal N}(\delta,[0,1]^k,|\cdot|)
{\cal N}(\delta^{(\tau - \alpha)},\mathbb{SO}(p+1),\|\cdot\|_{op})^k .
\]
\\
Using standard covering bounds
\[
{\cal N}(\delta,\mathbb S^p,d)\lesssim \delta^{-p},
\qquad
{\cal N}(\delta,[0,1]^k,|\cdot|)\lesssim \delta^{-k},
\]
and Theorem 7 of \cite{Szarek1998}
\[
{\cal N}(\delta^{(\tau-\alpha)},\mathbb{SO}(p+1),\|\cdot\|_{op})
\le
\left(\frac{C}{\delta^{(\tau-\alpha)}}\right)^{\dim(\mathbb{SO}(p+1))},
\]
we obtain
\[
\log {\cal N}(\epsilon_n,{\cal F}_{n,k},\tilde d)
\lesssim
n\epsilon_n^2
\]
for $\epsilon_n\asymp n^{-1/2}\log n$.
Since
\[ {\cal N}(\epsilon_n,{\cal F}_{n},\tilde d) = \sum_{k=1}^{K_\text{max}} {\cal N}(\epsilon_n,{\cal F}_{n,k},\tilde d)  \]
it follows that
\[
\log {\cal N}(\epsilon_n,{\cal F}_n,\tilde d)
\lesssim
n\epsilon_n^2 .
\]
\\
Thus, we have
\begin{eqnarray*}
    \sum_{j=M}^{\infty} N_{nj} \exp(-D_1 n j^2 \epsilon_n^2) &\lesssim& \exp(D_1 n \epsilon_n^2) \sum_{j=M}^{\infty} \int_{j-1}^{j} \exp(-D_1 n \epsilon_n^2 x^2 ) \intd x  \\
    & \le & \exp(D_1 n \epsilon_n^2) \int_{M-1}^{\infty} \exp(-D_1 n \epsilon_n^2 x^2 ) \intd x \\
    & \le & \exp(D_1 n \epsilon_n^2) \exp\Big(-\frac{1}{2} D_1 (M-1)^2 n \epsilon_n^2 \Big) \rightarrow 0
\end{eqnarray*} 
as $n \rightarrow \infty$ for sufficiently large $M > 0$.
\\\\
\textbf{Verification of Condition \ref{cond_prior_mass}}
\\
We now verify the prior mass condition. Recall that the true regression map $f_0$ is determined by the corresponding discrete measure
\[
\nu_0=\sum_{j=1}^{k_0} p_{0,j}\delta_{z_{0,j}}\in {\cal V}_{n,k_0},
\]
with atoms $z_{0,1},\ldots,z_{0,k_0}\in\mathbb S^p$ and associated rotations
$R_{0,1},\ldots,R_{0,k_0}\in\mathbb{SO}(p+1)$. Our goal is to obtain a lower bound on the prior probability of the neighborhood
\[
B(f_0,\epsilon)=\{f:\tilde d(f,f_0)<\epsilon\}.
\]
\\
From the verification of entropy condition above, there exists a constant
$c>0$ such that if
\[
|p_j-p_{0,j}|<c\epsilon_n^{1/(\tau-\alpha)},\quad
d(z_j,z_{0,j})<c\epsilon_n^{1/(\tau-\alpha)},\quad
\|R_j-R_{0,j}\|_{op}<c\epsilon_n,
\]
for $j=1,\ldots,k_0$, then the corresponding regression map $f$ satisfies
$\tilde d(f,f_0)\le \epsilon_n$. Hence, let
\[
\delta_n = c\epsilon_n^{1/(\tau-\alpha)} ,
\]
 we have
\begin{align*}
\Pi(B(f_0,\epsilon_n))
&\ge
\Pi\!\left(
|p_j-p_{0,j}|<\delta_n,\, j=1,\ldots,k_0 \mid K=k_0
\right) \\
&\quad\times
\Pi\!\left(
d(z_j,z_{0,j})<\delta_n,\, j=1,\ldots,k_0 \mid K=k_0
\right) \\
&\quad\times
\Pi\!\left(
\max_j\|R_j-R_{0,j}\|_{op}<c\epsilon_n \mid K=k_0
\right)
\,\pi_K(k_0).
\end{align*}
We lower bound each of the four probabilities on the RHS of the above display.
\\\\
\textbf{Atom weights.}
By Assumption \ref{prior_measure_assumption} and Lemma \ref{lemma_dirichlet}, we have
\[
\Pi(|p_j-p_{0,j}|<\delta_n,\ j=1,\ldots,k_0\mid K=k_0)
\gtrsim \delta_n^{2k_0}  \ge \exp(-D_1 n\epsilon_n^2),
\]
for some constant $D_1>0$.
\\\\
\textbf{Atom locations.}
The atoms of $\nu_0$ are separated by at least $\eta_n$.
Since $\delta_n=o(\eta_n)$, the conditions
\[
d(z_j,z_{0,j})<\delta_n,\quad j=1,\ldots,k_0
\]
automatically imply
\[
d(z_j,z_{j'})>\eta_n \quad \text{for } j\neq j',
\]
so that the separation constraint defining ${\cal V}_{n,k_0}$ remains
satisfied. By Assumption \ref{prior_measure_assumption}, we have
\[
\Pi(d(z_j,z_{0,j})<\delta_n,\ j=1,\ldots,k_0\mid K=k_0)
\gtrsim \delta_n^{pk_0}
\ge \exp(-D_2 n\epsilon_n^2),
\]
for some $D_2 > 0$.
\\\\
\textbf{Rotations.}
Since $\mathbb{SO}(p+1)$ is a compact manifold of dimension
$d_R=\dim(\mathbb{SO}(p+1))$, the volume of a ball of radius $\epsilon_n$
under $\|\cdot\|_{op}$ satisfies
\[
\Pi(\|R_j-R_{0,j}\|_{op}<c\epsilon_n)
\gtrsim \epsilon_n^{d_R}.
\]
Hence
\[
\Pi(\max_j\|R_j-R_{0,j}\|_{op}<c\epsilon_n\mid K=k_0)
\gtrsim \epsilon_n^{d_R k_0}
\ge \exp(-D_3 n\epsilon_n^2).
\]
for some $D_3 > 0$.
\\\\
\textbf{Number of atoms.}
Finally, $\pi_K(k_0)>0$ by Assumption \ref{prior_k_assumption}.
\\\\
combining the above bounds yields
\[
\Pi(B(f_0,\epsilon_n))
\ge
\exp(-D n\epsilon_n^2)
\]
for some constant $D>0$, establishing the prior mass condition.
\end{proof}

\section{Proof of Auxiliary Results}
\label{sec_proof_auxiliary}
\subsection{Proof of Lemma \ref{lemma_equi_metric}}
\begin{proof}
Using the following inequalities
$$ 1 - \frac{x^2}{2} \le \cos x \le 1 - \frac{x^2}{2} + \frac{x^4}{24} ,$$
we obtain
$$ f_1(x)^{T} f_0(x) \le 1 - \frac{(\cos^{-1}(f_1(x)^{T} f_0(x)))^2}{2} + \frac{(\cos^{-1}(f_1(x)^{T} f_0(x)))^4}{24} ,$$
and
$$ f_1(x)^{T} f_0(x) \ge 1 - \frac{(\cos^{-1}(f_1(x)^{T} f_0(x)))^2}{2} . $$
Since
$$ (\cos^{-1}(f_1(x)^{T} f_0(x)))^4 \le \pi^2 (\cos^{-1}(f_1(x)^{T} f_0(x)))^2,$$
it follows that
$$  1 - f_1(x)^{T} f_0(x) \ge \frac{12-\pi^2}{24} \big( \cos^{-1}(f_1(x)^{T} f_0(x)) \big)^2. $$
Thus, we have
$$ \tilde{d}^2(f_0, f_1) = \int_{\mathbb{S}^p}\big( \cos^{-1}( f_0(x)^T f_1(x) ) \big)^2 \intd\rho(x) \ge \int_{\mathbb{S}^p}  2 (1 - f_0(x)^T f_1(x) )  \intd\rho(x) ,$$
and
$$ \tilde{d}^2(f_0, f_1) = \int_{\mathbb{S}^p}\big( \cos^{-1}( f_0(x)^T f_1(x) ) \big)^2 \intd\rho(x) \le \int_{\mathbb{S}^p} \frac{24}{12-\pi^2} (1 - f_0(x)^T f_1(x) )  \intd\rho(x) .$$
\end{proof}

\subsection{Proof of Lemma \ref{lemma_cond_direction}}
\begin{proof}
For a vMF random variable with mean direction $\mu \in \mathbb{S}^{p}$ and concentration $\kappa$, the expectation is $A_p(\kappa) \mu $. Therefore, the result follows.
\end{proof}

\subsection{Proof of Lemma \ref{lemma_exist_test}}

\begin{proof}
For an arbitrary regression function $f_1: \mathbb{S}^p \rightarrow \mathbb{S}^p$, we define
$$ Z_i := (f_1(X_i) - f_0(X_i))^{T} Y_i - A_p(\kappa)(f_1(X_i)^{T} f_0(X_i) - 1), \quad i=1,\ldots,n.$$
We first consider the case of fixed covariates $x_1, \ldots, x_n$. We note that the random variables $Z_i, i=1,\ldots,n$ are bounded and by Lemma \ref{lemma_cond_direction} have zero mean 
$$ \mathbb{E}_0(Z_i|X_i=x_i) = 0.$$
Define the test statistic $T_n$ as
$$ T_n := \sum_{i=1}^{n} Z_i - \frac{1}{16} A_p(\kappa) t ,$$
where 
$$ t = n \bigg( \sqrt{\mathbb{P}_n(1 - f_1^T f_0)} \hat{d}(f_0, f_1) + \mathbb{P}_n(1 - f_1^T f_0) \bigg),$$
and the test function as
$$ \phi_n = 1\{T_n > 0\}.$$
We note that $t$ satisfies 
$$ t^2 \ge n^2\mathbb{P}_n(1 - f_1^T f_0) \hat{d}^2(f_0, f_1) . $$
We first derive the Type I and Type II errors of the test.
\\\\
\textbf{Type I Error}
\\
By Cauchy-Schwarz inequality, we have
$$ |(f_1(x) - f_0(x))^{T}Y|^2 \le (f_1(x) - f_0(x))^T (f_1(x) - f_0(x)) = 2(1 - f_1(x)^T f_0(x) ) .$$
Using Hoeffding's inequality we have
\begin{eqnarray*} \mathbb{E}_0(\phi_n|x_1,\ldots,x_n) &=& \mathbb{P}_0(T_n > 0|x_1,\ldots,x_n)  \\
 &=& \mathbb{P}_0 \bigg( \sum_{i=1}^{n} Z_i > \frac{1}{16} A_p(\kappa) t \bigg| x_1,\ldots,x_n \bigg) \\
& \le & \exp \bigg(- C_1\frac{ n^2 \hat{d}^{2}(f_0,f_1) \mathbb{P}_n(1 - f_1^T f_0) }{2n \mathbb{P}_n(1 - f_1^T f_0)} \bigg) \\
& = & \exp\bigg(- C_1\frac{n \hat{d}^{2}(f_0,f_1)}{2}\bigg)
\end{eqnarray*}
It follows that the unconditional expectation can also be bounded as:
$$ \mathbb{E}_0(\phi_n) \le \exp\bigg(- C_1\frac{n \hat{d}^{2}(f_0,f_1)}{2}\bigg).$$
\\\\
\textbf{Type II Error}
\\
We consider an arbitrary $f: \mathbb{S}^{p} \rightarrow \mathbb{S}^{p}$ satisfying
$$ \hat{d}^2(f,f_1) \le \hat{d}^2(f_0,f_1) / 32 .$$ 
Let $\mathbb{P}_n$ be the empirical measure of the covariates $X_i$, we re-write the test statistic as
\begin{eqnarray*}
    T_n &:=& \sum_{i=1}^{n} (f_1(X_i) - f_0(X_i))^{T} Y_i - A_p(\kappa) \sum_{i=1}^{n} (f_1(X_i) -f_0(X_i))^{T} f(X_i) \\
    && + A_p(\kappa) \sum_{i=1}^{n} (f_1(X_i) -f_0(X_i))^{T} f(X_i) - A_p(\kappa) \sum_{i=1}^{n} (f_1(X_i) -f_0(X_i))^{T} f_0(X_i) \\
    && - \frac{1}{16} A_p(\kappa) t \\
    & = & \sum_{i=1}^{n} (f_1(X_i) - f_0(X_i))^{T} Y_i - A_p(\kappa) \sum_{i=1}^{n} (f_1(X_i) -f_0(X_i))^{T} f(X_i) \\
    && + A_p(\kappa) n(\mathbb{P}_n(f_1 - f_0)^T f) - A_p(\kappa) n(\mathbb{P}_n(f_1 - f_0)^T f_0) - \frac{1}{16} A_p(\kappa) t .
\end{eqnarray*}
We consider the following event
\begin{eqnarray*}
{\cal E}_1 &:=& \bigg\{ (\mathbb{P}_n(f_1 - f_0)^T f) - (\mathbb{P}_n(f_1 - f_0)^T f_0) - \frac{1}{16} \hat{d}^2(f_0, f_1) \le \frac{1}{16} \hat{d}^2(f_0, f_1)  \bigg\} \\
&=& \bigg\{ \mathbb{P}_n(1 - f_0^{T} f) - \mathbb{P}_n(1 - f_1^{T} f) + \mathbb{P}_n(1 - f_1^{T} f_0) \le \frac{1}{8} \hat{d}^2(f_0, f_1) \bigg\} .   
\end{eqnarray*}
It is easy to see that the event ${\cal E}_1$ is contained in the event ${\cal E}_2$ defined as
$$ {\cal E}_2 := \bigg\{ - \mathbb{P}_n(1 - f_1^{T} f) + \mathbb{P}_n(1 - f_1^{T} f_0) \le \frac{1}{8} \hat{d}^2(f_0, f_1) \bigg\} .$$
We also define the event
$$ {\cal G}_1 := \bigg\{ \mathbb{P}_n(1 - f_1^T f) - \frac{1}{16} \mathbb{P}_n(1 - f_0^T f_1) \le 0 \bigg\} .$$
and upper bound $\mathbb{P}({\cal G}_1^{c})$ where ${\cal G}_1^c$ is the complement of ${\cal G}_1$.
We have
\begin{eqnarray*}
    && \mathbb{P}\bigg( \mathbb{P}_n(1 - f_1^T f) - \frac{1}{16} \mathbb{P}_n(1 - f_0^T f_1) > 0 \bigg) \\
    &=& \mathbb{P}\bigg( \mathbb{P}_n(1 - f_1^T f) - \frac{1}{16} \mathbb{P}_n(1 - f_0^T f_1) - \hat{d}^2(f, f_1) + \frac{1}{16} \hat{d}^{2}(f_0, f_1) \\
    && \quad > - \hat{d}^2(f, f_1) + \frac{1}{16} \hat{d}^{2}(f_0, f_1) \bigg) \\
    & \le & \mathbb{P}\bigg( \mathbb{P}_n(1 - f_1^T f) - \frac{1}{16} \mathbb{P}_n(1 - f_0^T f_1) - \hat{d}^2(f, f_1) + \frac{1}{16} \hat{d}^{2}(f_0, f_1) > \frac{1}{32} \hat{d}^{2}(f_0, f_1) \bigg)
\end{eqnarray*}
where the last inequality follows from the assumption that
$ \hat{d}^2(f,f_1) \le \hat{d}^2(f_0,f_1) / 32 .$
Now, for $i=1,\ldots,n$, define 
$$ U_i := ( 1 - f_1(X_i)^{T} f(X_i) ) - \frac{1}{16}( 1 - f_1(X_i)^{T} f_0(X_i) ) - \hat{d}^2(f, f_1) + \frac{1}{16} \hat{d}^2(f_0, f_1),$$
and using the inequality $ \hat{d}^2(f,f_1) \le \hat{d}^2(f_0,f_1) / 32$, we have
\begin{eqnarray*}
    \mathbb{E}(U_i^2) &\lesssim& \int_{ \mathbb{S}^p} (1 - f_1(x)^{T} f(x) )^2 \intd\rho(x) + \int_{ \mathbb{S}^p} (1 - f_1(x)^{T} f_0(x) )^2 \intd\rho(x) + \hat{d}^4(f_0, f_1) \\
    & \lesssim & \int_{ \mathbb{S}^p} 1 - f_1(x)^{T} f(x)  \intd\rho(x) + \int_{ \mathbb{S}^p} 1 - f_1(x)^{T} f_0(x)  \intd\rho(x) + \hat{d}^4(f_0, f_1) \\
    & \le &  \hat{d}^2(f, f_1) + \hat{d}^2(f_0, f_1) + \hat{d}^4(f_0, f_1) \\
    & \lesssim & \hat{d}^2(f_0, f_1) + \hat{d}^4(f_0, f_1)
\end{eqnarray*} 
Applying Bernstein's inequality, we obtain
\begin{eqnarray*}
\mathbb{P}( {\cal G}_1^{c}) &=&  \mathbb{P}\bigg( \mathbb{P}_n(1 - f_1^T f) - \frac{1}{16} \mathbb{P}_n(1 - f_0^T f_1) > 0 \bigg) \\
    &\le& \exp\bigg( - K_1 \frac{n \hat{d}^2(f_0, f_1) }{ \eta_1 + \hat{d}^2(f_0, f_1) } \bigg)
\end{eqnarray*} 
for some $K_1, \eta_1 > 0$.
\\\\
On the other hand, we have by Bernstein's inequality
\begin{eqnarray*}
    \mathbb{P}({\cal E}_2|{\cal G}_1) &\le &\mathbb{P}\bigg(\frac{15}{16} \mathbb{P}_n(1 - f_1^T f_0) \le \frac{1}{8} \hat{d}^2(f_0, f_1) \bigg) \\
    &=& \mathbb{P} \bigg( \frac{15}{16} \mathbb{P}_n(1 - f_1^T f_0) - \frac{15}{16} \hat{d}^2(f_0, f_1) \le -\frac{13}{16} \hat{d}^2(f_0, f_1) \bigg) \\
    & \le & \exp\bigg( - K_2 \frac{n \hat{d}^2(f_0, f_1) }{ \eta_2 + \hat{d}^2(f_0, f_1) } \bigg),
\end{eqnarray*}
for $K_2, \eta_2 > 0$.
\\\\
It follows that
\begin{eqnarray*}
    \mathbb{P}({\cal E}_1) \le \mathbb{P}({\cal E}_2) &\le& \mathbb{P}({\cal E}_2 | {\cal G}_1) + \mathbb{P}({\cal G}_1^{c}) \\
    & \le & \exp\bigg( - K_1 \frac{n \hat{d}^2(f_0, f_1) }{ \eta_1 + \hat{d}^2(f_0, f_1) } \bigg).+ \exp\bigg( - K_2 \frac{n \hat{d}^2(f_0, f_1) }{ \eta_2 + \hat{d}^2(f_0, f_1) } \bigg) \\
    & \le & 2 \exp\bigg( - C_2 \frac{n \hat{d}^2(f_0, f_1) }{ \eta + \hat{d}^2(f_0, f_1) } \bigg).
\end{eqnarray*}
for some constant $C_2, \eta > 0$.
\\\\
On the other hand, we have
\begin{eqnarray*}
 &&   \mathbb{P}_f( T_n < 0|{\cal E}_1^c) \\ &\le& \mathbb{P}_f \bigg( \sum_{i=1}^{n} (f_1(x_i) - f_0(x_i))^{T} Y_i - A_p(\kappa) \sum_{i=1}^{n} (f_1(x_i) -f_0(x_i))^{T} f(x_i) + \frac{1}{16} A_p(\kappa) t < 0 \bigg) \\
 & \le & \exp(-C_3 n \hat{d}^2(f_0,f_1))
\end{eqnarray*}
for some constant $C_3 > 0$, by Hoeffding's inequality.  Therefore,
\begin{eqnarray*}
 \mathbb{P}_f( T_n < 0 ) &\le& \mathbb{P}({\cal E}_1) + \mathbb{P}_f(T_n < 0|{\cal E}_1^c)  \\
    & \le & 2 \exp\bigg( - C_2 \frac{n \hat{d}^2(f_0, f_1) }{ \eta + \hat{d}^2(f_0, f_1) } \bigg) + \exp(-C_3 n \hat{d}^2(f_0, f_1)).
\end{eqnarray*}
Thus, we have shown that
\begin{align*}
    \mathbb{E}_0 \phi_n &\le \exp(-C n \hat{d}^2(f_0, f_1))\\
    \sup_{\{ f \in {\cal F}_n: \hat{d}(f,f_1) \le \xi \hat{d}(f_0,f_1) \} } \mathbb{E}_f(1 - \phi_n) &\le 2 \exp\bigg( - C \frac{n \hat{d}^2(f_0, f_1) }{ \eta + \hat{d}^2(f_0, f_1) } \bigg) + \exp(-C n \hat{d}^2(f_0,f_1))  .
\end{align*}
for some constant $C > 0$ and $\xi \in (0,1)$.
\end{proof}

\subsection{Proof of Lemma \ref{lemma_global_test}}

\begin{proof}
For each $n$, we partition the alternative into disjoint sets:
\begin{eqnarray*}
 \{ f \in {\cal F}_n: \hat{d}(f, f_0) > M \epsilon_n \} &\subset& \bigcup_{j=M}^{\infty} \{f \in {\cal F}_n: j \epsilon_n < \hat{d}(f,f_0) \le (j+1) \epsilon_n \} \\
 &=:& \bigcup_{j=M}^{\infty} {\cal S}_{nj}(\epsilon_n) .   
\end{eqnarray*}
For each ${\cal S}_{nj}(\epsilon_n)$, we can find $N_{nj} = {\cal N}(\xi j \epsilon_n, {\cal S}_{nj}(\epsilon_n), \hat{d})$ many functions $f_{njl} \in {\cal S}_{nj}(\epsilon_n)$ such that
$$ {\cal S}_{nj}(\epsilon_n) \subset \bigcup_{l=1}^{N_{nj}} \{ f \in {\cal F}_n: \hat{d}(f, f_{njl}) \le \xi j \epsilon_n \} .$$
By Lemma \ref{lemma_exist_test}, we can construct individual test function $\phi_{njl}$ for each function $f_{njl}$ with type I and type II errors bounded by
\begin{align*}
    \mathbb{E}_0 \phi_{njl} &\le \exp(-C n j^2 \epsilon_n^2) \\
    \sup_{ \{f \in {\cal F}_n: \hat{d}(f,f_{njl}) \le \xi \hat{d}(f_0, f_{njl}) \} } \mathbb{E}_f(1 - \phi_{njl}) &\le 2 \exp\bigg( - C \frac{n j^2 \epsilon_n^2 }{ \eta + j^2 \epsilon_n^2 } \bigg) + \exp(-Cnj^2\epsilon_n^2) .
\end{align*}
We then define the global test as 
$$\phi_n = \sup_{j \ge M} \max_{1 \le l \le N_{nj}} \phi_{njl}.$$
We can upper bound the type I error of $\phi_n$ as
\begin{eqnarray*}
    \mathbb{E}_0 \phi_n \le \sum_{j=M}^{\infty} \sum_{l=1}^{N_{nj}} \mathbb{E}_0 \phi_{njl} \le \sum_{j=M}^{\infty} N_{nj} \exp(-C n j^2 \epsilon_n^2 ) .
\end{eqnarray*}
Type II error can also be upper bounded:
\begin{eqnarray*}
    && \sup_{ \{f \in {\cal F}_n: \hat{d}(f, f_0) > M \epsilon_n \} } \mathbb{E}_f(1 - \phi_n) \\
    & \le& \sup_{j \ge M} \sup_{l=1,\ldots,N_{nj}} \sup_{ \{f \in {\cal F}_n: \hat{d}(f, f_{njl}) \le \xi \hat{d}(f_0, f_{njl}) \} } \mathbb{E}_f(1-\phi_n) \\
    & \le & \exp(-C M^2 n \epsilon_n^2) + 2 \exp\bigg( - C \frac{n M^2 \epsilon_n^2}{M^2 \epsilon_n^2 + \eta}   \bigg) .
\end{eqnarray*}
\end{proof}

\subsection{Proof of Lemma \ref{lemma_exp_lower_bnd}}

\begin{proof}
Let $\Pi(\intd f|B_n)$ be the restriction of the prior to the set $B_n = B(f_0; \epsilon_n)$. For any $C>0$, there exists $C'>0$ such that the following equality holds
\begin{eqnarray*}
    && \mathbb{P}_0 \bigg( \int_{{\cal F}} \prod_{i=1}^{n} \frac{p_f(x_i, y_i)}{p_0(x_i,y_i)} \Pi(\intd f) \le \Pi(B_n) \exp(-C n \epsilon_n^2) \bigg) \\
    & = &\mathbb{P}_0 \bigg( \int_{{\cal F}} \sum_{i=1}^{n} (f_0(X_i) - f(X_i))^T Y_i \Pi(\intd f|B_n) \ge C' n \epsilon_n^2  \bigg) .
\end{eqnarray*}

Applying Fubini's theorem, we obtain
\begin{eqnarray*}
    \mathbb{E}_0 \bigg[\int_{{\cal F}} (f_0(X) - f(X))^T Y \Pi(\intd f|B_n)  \bigg] &=& \int_{{\cal F}} \mathbb{E}_0 \big( (f_0(X) - f(X))^T Y \big) \Pi(\intd f|B_n) \\
    &=& \int_{{\cal F}} \mathbb{E} \big( \mathbb{E}_0 \big( (f_0(X) - f(X))^T Y \big| X = x \big) \big) \Pi(\intd f|B_n) \\
    &=& \int_{{\cal F}} \int_{\mathbb{S}^p} A_p(\kappa) (1 - f_0(x)^T f(x)) d\mu(x) \Pi(\intd f|B_n) \\
    &=& \int_{{\cal F}} A_p(\kappa) \hat{d}^2(f, f_0) \Pi(\intd f|B_n) \\
    & \le & A_p(\kappa) \epsilon_n^2,
\end{eqnarray*}
where the inequality follows since $\hat{d}^2(f,f_0) \le \epsilon_n^2$ for $f \in B_n$. 
\\\\
We also have that by Cauchy-Schwarz inequality and Fubini's theorem
\begin{eqnarray*}
    \mathbb{E}_0 \bigg[ 
\bigg( \int_{{\cal F}} (f_0(X) - f(X))^T Y \Pi(\intd f|B_n) \bigg)^2  \bigg] &\le& \mathbb{E}_0 \bigg[ 
 \int_{{\cal F}} \big( (f_0(X) - f(X))^T Y \big)^2 \Pi(\intd f|B_n)  \bigg] \\
 & = & \int_{{\cal F}} \mathbb{E}_0 \bigg( \big( (f_0(X) - f(X))^T Y \big)^2 \bigg) \Pi(\intd f|B_n) \\
  & \le & \int_{{\cal F}} A_p(\kappa) \hat{d}^2(f, f_0) \Pi(\intd f|B_n) \\
    & \le & A_p(\kappa) \epsilon_n^2 .
\end{eqnarray*}
We apply Bernstein's inequality and the two bounds above:
\begin{eqnarray*}
    && \mathbb{P}_0 \bigg(    \sum_{i=1}^{n} \int_{{\cal F}}  (f_0(X_i) - f(X_i))^T Y_i \Pi(\intd f|B_n)  - n \int_{{\cal F}} A_p(\kappa) \hat{d}^2(f_0,f) \Pi(\intd f|B_n) \\
    && \quad \ge C n \epsilon_n^2 - n \int_{{\cal F}} A_p(\kappa) \hat{d}^2(f_0,f) \Pi(\intd f|B_n)  \bigg) \\
   & \le &  \exp \bigg( \frac{-\frac{1}{2} (C n \epsilon_n^2 - n A_p(\kappa) \epsilon_n^2)^2}{n A_p(\kappa) \epsilon_n^2 + \frac{1}{3} C n \epsilon_n^2}  \bigg) \\
   & \le & \exp(-C_1 n \epsilon_n^2) \rightarrow 0
\end{eqnarray*}
for all $C > A_p(\kappa)$ and $C_1 > 0$ is some constant. 
\end{proof}

\subsection{Proof of Lemma \ref{lemma_stability}}
\begin{proof}
    
By Lemma 3.2 of \cite{Merigot2020}, it is sufficient to consider two discrete measures $\nu^0, \nu^1$ supported on the same fixed set ${\cal Z} := \{ z_1, \ldots, z_N \} \subset \mathbb{S}^{p}$. Let $c(x,z) = d(x,z)^2/2$ be the cost function where $d$ is the geodesic distance on $\mathbb{S}^{p}$. Recall that $\mu$ is the uniform measure on $\mathbb{S}^{p}$, and $S_{\nu^k}$ be the optimal transport map from $\mu$ to $\nu^k$, $k=0,1$. 

\vskip5pt
For each $k=0,1$, let $\phi^k$ be the Brenier potential such that
\begin{equation}\label{OTmap}
    S_{\nu^k}(x)=\exp_x[-\nabla\phi^k(x)].
\end{equation}
Let $\psi^k$ be the dual potential function on ${\cal Z}$:
$$\psi^k(z)=\inf_{x\in\mathbb{S}^{p}}c(x,z)-\phi^k(x), \quad z \in {\cal Z}.$$
Note that $\psi^k$ can be identified with the vector $\boldsymbol{\psi}^k$ defined by $\boldsymbol{\psi}^k_{i} = \psi^k(z_i)$, $i=1,\cdots,N$. 
Meanwhile,
\begin{equation}\label{potenPhi}
    \phi^k(x)=\inf_{z\in\cal Z}c(x,z)-\psi^k(z)\quad\text{ for }x\in\mathbb{S}^{p}. 
\end{equation}

\vskip5pt
Recall that the vector $\boldsymbol{\psi}^k$ induces a partition of $\mathbb{S}^{p}$ into Laguerre cells:
\begin{equation}\label{Lagu}
    \textrm{Lag}_{z_i}(\boldsymbol{\psi}^k) = \{ x \in \mathbb{S}^{p}: \boldsymbol{\psi}_j^k \le \boldsymbol{\psi}_i^k + c(z_j, x) - c(z_i, x) \quad \forall\,j=1, \ldots, N \}.
\end{equation} 

\vskip5pt
\noindent
For a given $\boldsymbol{\psi}^k$, let $$ G_i(\boldsymbol{\psi^k}) = \mu(\textrm{Lag}_{z_i}(\boldsymbol{\psi}^k))  =: \nu^k_i $$ 
be the measure of the set $\textrm{Lag}_{z_i}(\boldsymbol{\psi}^k)$, and let 
$$ \nu^k = (\nu^k_i)_{1 \le i \le N} \in \mathbb{R}^N .$$ 
We let ${\cal S}_{+} \subset \mathbb{R}^N$ be the set of potentials such that all Laguerre cells $\textrm{Lag}_{z_i}(\boldsymbol{\psi}^k)$ have positive measures
$$ {\cal S}_{+} = \{ \boldsymbol{\psi} \in \mathbb{R}^N: G_i(\boldsymbol{\psi}) > 0 \quad \forall\, i=1,\cdots,N \}. $$
Similarly, we define $\mathcal{P}^+(\cal Z)$ to be the set of positive measures on $\cal Z$
$$\mathcal{P}^+({\cal Z}) = \{\nu=({\nu}_1,\cdots,{\nu}_N)\in\mathbb{R}^N : \nu_i=\nu(\{z_i\})>0 \quad\forall\,i=1,\cdots,N\}.$$

\vskip5pt
By the uniqueness of optimal transportation, for each $\nu\in\mathcal{P}^+({\cal Z})$, there exists a unique optimal mapping $S_\nu : \mathbb{S}^{p}\to\mathcal{Z}$. Our aim is to show the \textbf{stability of the optimal mappings}, namely for any $\nu^0, \nu^1\in\mathcal{P}^+({\cal Y})$,
    \begin{equation}\label{stabOT}
        \tilde{d}( S_{\nu^0}, S_{\nu^1}) \leq CW_1(\nu^0,\nu^1)^{\frac{1}{9}}.
    \end{equation}
where we recall that 
$$  \tilde{d}^2( S_{\nu^0}, S_{\nu^1}) = \int_{\mathbb{S}^p} d^2(S_{\nu^0}(x), S_{\nu^1}(x)) \intd\mu(x) . $$

\vskip5pt
From \eqref{OTmap} and \eqref{potenPhi}, the key part of proving \eqref{stabOT} is to obtain the \textbf{stability of dual potentials}, namely
    \begin{equation}\label{stabPo}
        \|\psi^0-\psi^1\|_{L^2(\nu^0)}^2 \lesssim W_1(\nu^0,\nu^1)^{\frac{2}{3}},
    \end{equation}
provided their difference is normalised (by adding a proper constant)
    \begin{equation}\label{normPo}
        \sum_{i=1}^N(\psi^0_i-\psi^1_i)\nu^0_i=0.
    \end{equation}

\vskip5pt
Assuming \eqref{stabPo} for the moment, we first prove \eqref{stabOT} as follows.
\begin{proof}[Proof of \eqref{stabOT}]
Heuristically, \eqref{stabPo} says $\psi^0$ and $\psi^1$ are close in $L^2$. By Chebyshev's inequality, they are close in $L^\infty$ in a large portion. Hence, their Legendre transforms $\phi^0$ and $\phi^1$ are close in $L^\infty$ in a large portion. By a Gagliardo-Nirenberg type inequality, this implies that $\nabla\phi^0$ and $\nabla\phi^1$, and thus $T_{\nu^0}$ and $T_{\nu^1}$ are close in $L^2$, namely \eqref{stabOT}.

\vskip5pt
Below, we divide the proof into several steps:

\vskip5pt
\noindent\textbf{($i$): Although \eqref{stabPo} assumed $\phi^0$ and $\phi^1$ are close in $L^2$ measured by $\mu^0$, we show that this also holds when measuring by $\mu^1$.}

Note that the difference $(\psi^0-\psi^1)$ is normalised w.r.t. $\nu^0$, i.e. \eqref{normPo}, but not w.r.t. $\nu^1$. Let's define
    \begin{equation*}
        \tilde\psi^0_i = \psi^0_i - \sum_{i=1}^N(\psi^0_i-\psi^1_i)\nu^1_i, \qquad
        \tilde\psi^1_i = \psi^1_i.
    \end{equation*}
It's easy to see the difference $(\tilde\psi^0-\tilde\psi^1)$ is now normalised w.r.t. $\nu^1$ in the sense that $\sum_{i=1}^N(\tilde\psi^0_i-\tilde\psi^1_i)\nu^1_i=0.$ Hence, by \eqref{stabPo} we obtain
    $$\|\tilde\psi^0-\tilde\psi^1\|_{L^2(\nu^1)}\lesssim W_1(\nu^0,\nu^1)^{\frac{1}{3}}.$$

Therefore, 
    \begin{align}\label{stabL2u1}
        \|\psi^0-\psi^1\|_{L^2(\nu^1)} &= \|(\psi^0-\tilde\psi^0)-(\psi^1-\tilde\psi^1)+(\tilde\psi^0-\tilde\psi^1)\|_{L^2(\nu^1)} \nonumber\\
            &\le \left|\sum_{i=1}^N(\psi^0_i-\psi^1_i)\nu^1_i\right| + \|\tilde\psi^0-\tilde\psi^1\|_{L^2(\nu^1)} \nonumber\\
            &= \left|\int_{\mathcal{Z}}(\psi^0-\psi^1)\intd(\nu^1-\nu^0)\right| + \|\tilde\psi^0-\tilde\psi^1\|_{L^2(\nu^1)} \\
            &\lesssim W_1(\nu^0,\nu^1) + W_1(\nu^0,\nu^1)^{\frac{1}{3}} \lesssim W_1(\nu^0,\nu^1)^{\frac{1}{3}}, \nonumber
    \end{align}
where we used Kantorovich-Rubinstein's theorem and assumed $W_1(\nu^0,\nu^1)$ is sufficiently small.

\vskip5pt
\noindent\textbf{($ii$): From $L^2$-close to $L^\infty$-close in a large portion.}

From \eqref{stabPo} and \eqref{stabL2u1}, for $k\in\{0,1\}$
    \begin{equation}\label{stabPPP}
        \|\psi^0-\psi^1\|_{L^2(\nu^k)}^2 \lesssim W_1(\nu^0,\nu^1)^{\frac{2}{3}} =: \varepsilon.
    \end{equation}
For $\alpha\in(0,1)$, define
    \begin{equation}\label{setYa}
        \mathcal{Z}_\alpha:=\{z\in\mathcal{Z} : |\psi^0(z)-\psi^1(z)| \le \varepsilon^\alpha\}.
    \end{equation}
By Chebyshev's inequality and \eqref{stabPPP} we have that for $k\in\{0,1\}$,
    $$\varepsilon^{2\alpha}\nu^k(\mathcal{Z}\setminus\mathcal{Z}_\alpha) \le \|\psi^0-\psi^1\|_{L^2(\nu^k)}^2\lesssim\varepsilon.$$
This implies that
    \begin{equation}\label{largeY}
        1-\nu^k(\mathcal{Z}_\alpha) \lesssim \varepsilon^{1-2\alpha}.
    \end{equation}

Recall that $\mathcal{Z}=\{z_1,\cdots,z_N\}$ and from \eqref{potenPhi},
    $$\phi^k(x) = \min_i c(x,z_i) - \psi^k(z_i),\quad x\in\mathbb{S}^{p}.$$
Define
    \begin{equation}\label{resLegen}
        \phi^{k,\alpha}(x) = \min_{z\in\mathcal{Z}_\alpha} c(x,z) - \psi^k(z), \quad x\in\mathbb{S}^{p}.
    \end{equation}
One can see that $\phi^{k,\alpha} \geq \phi^k$, and if $z_i\in\mathcal{Z}_\alpha$,
    \begin{equation}\label{same1}
        \phi^{k,\alpha}(x)=\phi^k(x)\quad\text{ for } x\in \textrm{Lag}_{z_i}(\psi^k),
    \end{equation}
since the minimums are both attained at the point $z_i$. This implies that $\phi^{k,\alpha} \equiv \phi^k$ on the subset 
    \begin{equation}\label{same2}
        \mathcal{X}^k_\alpha = \bigcup_{z_i\in\mathcal{Z}_\alpha}\textrm{Lag}_{z_i}(\psi^k)\subset\mathbb{S}^{p}.
    \end{equation}
By measure-preserving one has $\mu(\textrm{Lag}_{z_i}(\psi^k))=\nu^k(z_i)$, and thus
$$\mu(\mathcal{X}^k_\alpha)=\sum_{z_i\in\mathcal{Z}_\alpha}\mu(\textrm{Lag}_{z_i}(\psi^k))=\sum_{z_i\in\mathcal{Z}_\alpha}\nu^k(z_i)=\nu^k(\mathcal{Z}_\alpha).$$
Then by \eqref{largeY} we obtain
    \begin{equation}\label{smallcomp}
        \mu(\mathbb{S}^{p}\setminus\mathcal{X}^k_\alpha) = 1-\mu(\mathcal{X}^k_\alpha) \lesssim \varepsilon^{1-2\alpha}.
    \end{equation}

For $x\in\mathbb{S}^{p}$, by definitions \eqref{resLegen} and \eqref{setYa} we have
    \begin{align*}
        \phi^{0,\alpha}(x) &=  \min_{z\in\mathcal{Z}_\alpha} c(x,z) - \psi^0(z)    \\
                    &\ge \min_{z\in\mathcal{Z}_\alpha} c(x,z) - \psi^1(z) - \varepsilon^\alpha \\
                    &= \phi^{1,\alpha}(x) -\varepsilon^\alpha,
    \end{align*}
thus by symmetry 
    \begin{equation}\label{Linfa}
        \|\phi^{0,\alpha}-\phi^{1,\alpha}\|_{L^\infty(\mathbb{S}^{p})} \le \varepsilon^\alpha. 
    \end{equation}

\vskip5pt
\noindent\textbf{($iii$): From $L^\infty$-close to $H^2$-close.}

We are now ready to prove \eqref{stabOT}. We first show that the exponential map $\exp_x$ is Lipschitz continuous with Lipschitz constant 1. 
In particular, $\exp_x$ does not increase distances. Therefore, 
    \begin{align}\label{mainest}
       \tilde{d}( S_{\nu^0}, S_{\nu^1}) &= \left(\int_{\mathbb{S}^{p}} d^2(\exp_x[-\nabla\phi^0(x)],\exp_x[-\nabla\phi^1(x)])\,\intd\mu(x)\right)^{\frac{1}{2}} \nonumber \\
            &\le \left(\int_{\mathbb{S}^{p}} |\nabla\phi^0(x)-\nabla\phi^1(x)|^2\,\intd\mu(x)\right)^{\frac{1}{2}} \\
            &\le \|\nabla\phi^0-\nabla\phi^{0,\alpha}\|_{L^2(\mathbb{S}^{p})} + \|\nabla\phi^{0,\alpha}-\nabla\phi^{1,\alpha}\|_{L^2(\mathbb{S}^{p})} + \|\nabla\phi^{1,\alpha}-\nabla\phi^1\|_{L^2(\mathbb{S}^{p})} \nonumber\\
            &=: I + I\!I + I\!I\!I. \nonumber
    \end{align}

Estimate on $I,I\!I\!I$: By \eqref{same1} and \eqref{same2}, $\nabla\phi^{k,\alpha}=\nabla\phi^k$ on $\mathcal{X}^k_\alpha$. Since both the gradients $\nabla\phi^{k,\alpha}$ and $\nabla\phi^k$ are bounded by $\pi$, from \eqref{smallcomp} we have
    \begin{align}\label{est1-3}
        I,I\!I\!I &= \|\nabla\phi^{k,\alpha}-\nabla\phi^k\|_{L^2(\mathbb{S}^{p})} \nonumber \\
            &= \|\nabla\phi^{k,\alpha}-\nabla\phi^k\|_{L^2(\mathbb{S}^{p}\setminus\mathcal{X}^k_\alpha)} \nonumber \\
            &\le 2\pi\mu^{1/2}(\mathbb{S}^{p}\setminus\mathcal{X}^k_\alpha) \\
            &\lesssim \varepsilon^{\frac{1}{2}-\alpha}. \nonumber
    \end{align}

Estimate on $I\!I$: We need to utilise \eqref{Linfa}, and use the \emph{claim} that 
\begin{equation}\label{claim}
    \|\nabla\phi^{0,\alpha}-\nabla\phi^{1,\alpha}\|_{L^2(\mathbb{S}^{p})} \leq C \|\phi^{0,\alpha}-\phi^{1,\alpha}\|_{L^\infty(\mathbb{S}^{p})}^{\frac{1}{2}},
\end{equation}
and thus obtain that 
\begin{equation}\label{est2}
    \|\nabla\phi^{0,\alpha}-\nabla\phi^{1,\alpha}\|_{L^2(\mathbb{S}^{p})} \lesssim \varepsilon^\alpha. 
\end{equation}

Combining \eqref{est1-3} and \eqref{est2} into \eqref{mainest}, we have
\begin{equation*}
\tilde{d}( S_{\nu^0}, S_{\nu^1})
 \lesssim \varepsilon^{\frac{1}{2}-\alpha} + \varepsilon^{\frac{\alpha}{2}}.
\end{equation*}
Setting $\alpha=\frac{1}{3}$ so that $\frac{1}{2}-\alpha=\frac{\alpha}{2}$, and recalling $\varepsilon=W_1(\nu^0,\nu^1)^{\frac{2}{3}}$, we then obtain the desired estimate
\begin{equation*}
   \tilde{d}( S_{\nu^0}, S_{\nu^1})  \lesssim W_1(\nu^0,\nu^1)^{\frac{1}{9}}.
\end{equation*}

Therefore, it suffices to prove the claim \eqref{claim}. 
By calculation,
    \begin{align*}
        \int_{\mathbb{S}^{p}} |\nabla\phi^{0,\alpha}(x) - \nabla\phi^{1,\alpha}(x)|^2\,\intd\mu(x) &= - \int_{\mathbb{S}^{p}} (\phi^{0,\alpha}-\phi^{1,\alpha}) \Delta (\phi^{0,\alpha}-\phi^{1,\alpha} ) \,\intd\mu \\
            &\le \|\phi^{0,\alpha}-\phi^{1,\alpha}\|_{L^\infty(\mathbb{S}^{p})} \left( \int_{\mathbb{S}^{p}}|\Delta\phi^{0,\alpha}| +| \Delta\phi^{1,\alpha}|\,\intd\mu\right).
    \end{align*}
Note that for each $k=\{0,1\}$, $\phi^{k,\alpha}$ defined in \eqref{resLegen} is a Brenier potential of the optimal mapping $S^{k,\alpha} : \mathbb{S}^{p} \to \mathcal{Z}_\alpha$. Correspondingly, the Laguerre cells will be
$$ \textrm{Lag}_{z_i}^{k,\alpha} = \{x\in\mathbb{S}^{p} : \text{the minimum in \eqref{resLegen} is attained at $z_i\in\mathcal{Z}_\alpha$}\} $$
and the target measure will be
$$\nu^{k,\alpha}(\{z_i\})=\mu(\textrm{Lag}_{z_i}^{k,\alpha}),\quad\text{ if } z_i\in\mathcal{Z}_\alpha;\qquad \nu^{k,\alpha}(\{z_i\})=0,\quad\text{ if } z_i\notin\mathcal{Z}_\alpha.$$
One can see that $\nu^{k,\alpha}$ is supported on $\mathcal{Z}_\alpha$ and $\nu^{k,\alpha} \ge \nu^k>0$ on $\mathcal{Z}_\alpha$.
Therefore, Lemma \ref{stayaway} applies, and thus for each $z_i\in\mathcal{Z}_\alpha$, 
\begin{equation}\label{3ks}
    d(-z_i, \textrm{Lag}_{z_i}^{k,\alpha}) \ge \delta
\end{equation}
for some constant $\delta>0$.
\\\\
Again, from the definition \eqref{resLegen},
$$\phi^{k,\alpha}(x) = \frac{1}{2}d^2(x,z_i)-\psi^k(z_i)\quad\text{ for } x\in \textrm{Lag}_{z_i}^{k,\alpha},\ \ z_i\in\mathcal{Z}_\alpha.$$
Let $(r, \theta)$ be the polar coordinates around $z_i$. It can be shown that if $d\left(x, z_i\right)=r$ and $\phi^{k, \alpha}(x)=\frac{1}{2} d^2\left(x, z_i\right)-\psi^k\left(z_i\right)$, then
$\Delta \phi^{k, \alpha}(x)=1+(p-1) r \cot r$. By \eqref{3ks}, $\textrm{Lag}_{z_i}^{k, \alpha} \subset\{0 \leq r \leq \pi-\delta\}$. In particular, $\left|\Delta \phi^{k, \alpha}(x)\right| \leq 1+(p-1) r|\cot r|$ is uniformly bounded by $C=C(p, \delta)=\max \{p,1+(p-1)(\pi-\delta)|\cot (\pi-\delta)|\}$ on $\textrm{Lag}_{z_i}^{k, \alpha}$. Therefore
\begin{align*}
\int_{\mathbb{S}^{p}}\left|\Delta \phi^{k, \alpha}\right| \intd\mu & =\sum_{z_i \in \mathcal{Z}_\alpha} \int_{\textrm{Lag}_{z_i}^{k, \alpha}}\left|\Delta \phi^{k, \alpha}\right| \intd\mu \\
& \leq \sum_{z_i \in \mathcal{Z}_\alpha} C \mu\left(\textrm{Lag}_{z_i}^{k, \alpha}\right) \\
& =C \mu\left(\mathbb{S}^{p}\right)=C.
\end{align*}
Hence, we obtain $\|\nabla\phi^{0,\alpha}-\nabla\phi^{1,\alpha}\|_{L^2(\mathbb{S}^{p})} \lesssim  \|\phi^{0,\alpha}-\phi^{1,\alpha}\|_{L^\infty(\mathbb{S}^{p})}^{\frac{1}{2}}$.
\end{proof}

\begin{lemma}
\label{lemma_exp_map_lipschitz}
 For any $x \in \mathbb{S}^p$, the exponential map  $\exp_x$ is Lipschitz continuous with Lipschitz constant 1.
\end{lemma}

\begin{proof}
    Let $M=\mathbb{S}^p$. We claim that the exponential map has Lipschitz constant 1, which is implied by showing $\left|\left(d \exp_x\right)\right|_v(w)|\le |w|$. As $\left(d \exp_x\right)_0=\mathrm{id}$, we only have to consider $\left(d \exp_x\right)_v$ for $v \ne 0 \in T_x M$. Let $0 \ne v \in T_x M$ and $w \in T_v\left(T_x M\right)$, then
$$
\left. \left(d \exp_x\right)\right|_v(w)=\left. \frac{d}{d s}\right|_{s=0} \exp_x(v+s w).
$$

Now, consider the variation of geodesics

$$
\Gamma(s, t):=\exp_x(t(v+s w)).
$$

The corresponding Jacobi field is

$$
J(t)=\left. \frac{\partial}{\partial s}\right|_{s=0} \Gamma(s, t).
$$

From this we see that
$$
\left. d \exp_x\right|_v(w)=J(1).
$$

To estimate $d \exp_x$, it is equivalent to estimate $J(1)$.

Suppose $w$ is parallel to $v$, then we have
$|J(1)|=|(d\exp_x)_v(w)|=|w|$, thanks to the radial isometry of the exponential map.

If $w$ is perpendicular to $v$, then by Lemma 10.8 in \cite{lee2006riemannian}, we have

$$
|J(1)| =\frac{\sin (|v|)}{|v|}\left|J^{\prime}(0)\right| \le \left|J^{\prime}(0)\right|.
$$

We claim that $J^{\prime}(0)=w$. This follows from the computation

\begin{align*}
J^{\prime}(0) & =\left. \nabla_{t=0} \frac{\partial}{\partial s}\right|_{s=0} \Gamma(s, t) \\
& =\left. \nabla_{s=0} \frac{\partial}{\partial t}\right|_{t=0} \Gamma(s, t) \\
& =\left. \nabla_{s=0}\left(d \exp_x\right)\right|_0(v+s w) \\
& =\nabla_{s=0}(v+s w) \\
& =w.
\end{align*}

From this, it also follows that if $w$ is perpendicular to $v$, then
$$
\left|\left(d \exp_x\right)\right|_v(w)|\le |w|.
$$

In general, by the Gauss lemma, if we have $w=w^T+w^{\perp}$, where $w^T$ is parallel to $v$ and $w^\perp$ is perpendicular to $v$, then the decomposition $(d\exp_x)_v(w)=(d\exp_x)_v(w^T)+(d\exp_x)_v(w^\perp)$ remains perpendicular by the Gauss lemma (Theorem 6.8 in \cite{lee2006riemannian}), and so the above shows that
$\left|\left(d \exp_x\right)\right|_v(w)|\le |w|$.
\end{proof}

\vskip5pt

\begin{lemma}\label{stayaway}
    Assume $N\geq2$, $\mathcal{Z}=\{z_1,\cdots,z_N\}$. Let $\nu$ be a probability measure supported on $\mathcal{Z}$ satisfying $\nu(\{z_i\})>0$ for all $1\leq i\leq N$. Let $S:(\mathbb{S}^{p},\mu) \to (\mathcal{Z},\nu)$ be the optimal mapping and $\psi$ be the dual potential function on $\mathcal{Z}$, accordingly the Laguerre cells $\textrm{Lag}_{z_i}(\psi)$ is defined in \eqref{Lagu}.
    Then, there exists a constant $\delta>0$ such that 
    $$d(\hat z_i, \textrm{Lag}_{z_i}(\psi))\geq \delta \quad\forall\,i=1,\cdots,N,$$
    where $\hat z_i =-z_i$ is the antipodal point of $z_i$. 
\end{lemma}
\begin{proof}
    This follows from an inequality
    \begin{equation}\label{CLY2.4}
        d(Sx, \hat x_0) \leq 2\pi \frac{d(Sx_0, \hat x_0)}{d(x,x_0)}
    \end{equation}
    where $x_0,x\in\mathbb{S}^{p}$ and $x\neq x_0$. The above inequality relies on the optimality of $S$ and was proved in \cite{CLY} on hemisphere. For completeness, we include it here.

    Define the function
    \begin{equation}\label{defF}
        F(p)=\frac{1}{2}d^2(p,x) - \frac{1}{2}d^2(p,x_0) \quad\text{ for } p\in\mathbb{S}^{p}.
    \end{equation}
    By calculation, $\nabla F(p)=v_x-v_{x_0}$, where $v_x, v_{x_0} \in T_p \mathbb{S}^{p}$ satisfy $\exp _p\left(-v_x\right)=x$ and $\exp _p\left(-v_{x_0}\right)=x_0$. So 
    $$|\nabla F(p)|=\left|v_x-v_{x_0}\right| \geq d\left(x_0, x\right)$$ by the distance-nondecreasing property of $\exp _p$. 
   \\\\
Let us consider on $\mathbb{S}^{p}\setminus\{\hat x_0\}$ the normalized steepest descent equation (with arc-length parameter $s$):
$$\dot p(s)=-\frac{\nabla_pF[p(s)]}{|\nabla_pF(p(s))|}.$$
Then a solution $p(s)$ satisfies
$$\frac{d}{ds}F[p(s)]=-|\nabla_pF(p(s))|\leq -d(x_0,x).$$
    Since the function $F(p)$ attains its minimum at $p=\hat x_0$, starting from $p(0)=p_0$, the minimum of $p\mapsto F(p)$ is reached by flowing along an integral curve of length $L\geq d(p_0,\hat x_0)$. 
    Writing 
    $$F(p_0)-F(\hat x_0)=-\int_0^L\frac{d}{ds}F[(p(s))]\,\intd s,$$
we then have 
$$F(p_0)-F(\hat x_0) \geq \int_0^Ld(x_0,x) \intd s \geq d(x_0,x)d(p_0,\hat x_0).$$
    It implies that for $x\neq x_0$, 
    \begin{equation}\label{CLY2.8}
        d(p_0,\hat x_0)\leq\frac{F(p_0)-F(\hat x_0)}{d(x_0,x)}\quad\forall\, p_0\in\mathbb{S}^{p}.
    \end{equation}

    Since $T$ is the optimal mapping, it must be $c$-monotone \cite{McCann2001}, namely
        $$d^2(x_0,S x_0)+d^2(x,S x) \leq d^2(x_0,S x)+d^2(x,S x_0).$$
        From the definition of function $F$, we get
        $$F(S x)\leq F(S x_0).$$
    Now, setting $p_0=S x$ in \eqref{CLY2.8}, we have
    $$d(S x,\hat x_0)\leq\frac{F(S x)-F(\hat x_0)}{d(x_0,x)}\leq\frac{F(S x_0)-F(\hat x_0)}{d(x_0,x)}.$$
    Hence, since $p\mapsto F(p)$ is $2\pi$-Lipschitz, we obtain \eqref{CLY2.4}, namely
    $$d(Sx,\hat x_0)\leq 2\pi\frac{d(Sx_0, \hat x_0)}{d(x_0,x)}.$$
\\\\
    We can now prove the lemma by contradiction. 
    Suppose there does not exist a constant $\delta>0$ such that
    $$d(\hat z_i, \textrm{Lag}_{z_i}(\psi))\geq \delta \quad\forall\,i=1,\cdots,N.$$
This implies that
$$ \min_{i=1,\ldots,N}  d(\hat{z}_i, \textrm{Lag}_{z_i}(\psi)) < \frac{1}{k}$$
for all $k \in \mathbb{N}$,
which implies that
$$ \min_{i=1,\ldots,N} d(\hat{z}_i, \textrm{Lag}_{z_i}(\psi)) = 0 .$$
Therefore, there is some $1\leq i\leq N$ such that
    $$d(\hat z_i, \textrm{Lag}_{z_i}(\psi))=0,$$
    namely there is $x_0\in \textrm{Lag}_{z_i}(\psi)$ such that $d(T x_0,\hat x_0)=0$. By \eqref{CLY2.4}, we have
        $$d(Sx,\hat x_0)=0\quad\forall\, x\neq x_0,\, x\in\mathbb{S}^{p},$$
    namely $S$ maps the whole sphere $\mathbb{S}^{p}$ into a single point $\hat x_0$, which contradicts the fact $\mu>0$ and $N\geq2$.

\end{proof}

\vskip5pt

In the following we shall prove the key estimate \eqref{stabPo}.

\begin{lemma}
\label{lemma_2}
Suppose $v := \boldsymbol{\psi}^1 - \boldsymbol{\psi}^0$ satisfies $\langle v | G(\psi^0)\rangle=0$, where $\langle \cdot \,|\, \cdot \rangle$ denotes the inner product on $\mathbb{R}^N$. This is equivalent to \eqref{normPo}. 
Then, we have
$$ \langle v^2 | G(\boldsymbol{\psi}^0) \rangle \lesssim W_1(\nu^0, \nu^1)^{\frac{2}{3}}. $$
\end{lemma}

\begin{proof}
By Kantorovich-Rubinstein's theorem, we have
    \begin{align}\label{nb1}
        \left|\langle v | G(\psi^1) - G(\psi^0)\rangle\right| &= \left|\int_\mathcal{Z} (\psi^1-\psi^0) \,\intd(\nu^1-\nu^0)\right| \nonumber\\
                &\le \max_{Lip(f)\le 2\pi} \int_\mathcal{Z} f\,\intd(\nu^1-\nu^0) \nonumber\\
                &= 2\pi \max_{Lip(f)\le 1} \int_\mathcal{Z} f\,\intd(\nu^1-\nu^0)  \nonumber\\
                &\lesssim W_1(\nu^0,\nu^1). 
    \end{align}

By Taylor's formula
    $$\langle v | G(\psi^1) - G(\psi^0)\rangle = \int_0^1 \langle DG(\psi^t)v|v\rangle\,\intd t, $$
where $\psi^t=(1-t)\psi^0+t\psi^1$.
From the Discrete Poincaré-Wirtinger inequality [Proposition 3.4 of Mérigot et al. (2020)], we have
    \begin{equation}\label{DisPW}
        \langle v^2 | G(\psi^t)\rangle - \langle v| G(\psi^t)\rangle^2 \lesssim -\langle DG(\psi^t)v|v\rangle
    \end{equation}
and $\langle DG(\psi^t)v|v\rangle \leq 0$ for all $t\in[0,1]$. 

\vskip5pt
To verify that inequality \eqref{DisPW} holds in our setting, note that by Lemma \ref{stayaway} we obtain \eqref{3ks} for each $i$. Although the background space in [Mérigot et al. (2020)], following [Eymard et al., 2000], is Euclidean, the above Discrete Poincaré-Wirtinger inequality remains valid in our case, up to a constant depending on the uniform constant $\delta$ from Lemma \ref{stayaway}. For the readers' convenience and for completeness, we provide additional details below. For clarity, we denote by LHS and RHS the left- and right-hand sides of inequality~\eqref{DisPW} respectively. 

\vskip5pt
Define a function $u$ on $\mathbb{S}^p$ such that $u|_{\textrm{Lag}_{z_i}(\psi^t)}=v_i$. Then by computation, the left hand side of \eqref{DisPW}
    \[ LHS = \frac12 \int_{\mathbb{S}^p\times\mathbb{S}^p} (u(x)-u(y))^2 \intd y \intd x. \]
Note that $u(x)-u(y)=v_i-v_j$ if and only if $x\in\textrm{Lag}_{z_i}(\psi^t)$ and $y\in\textrm{Lag}_{z_j}(\psi^t)$. By \eqref{Lagu} and $c=\frac12d^2$, we have
    \begin{equation}\label{2bla}
        0\leq d^2(z_i,y)-d^2(z_i,x)+d^2(z_j,x)-d^2(z_j,y). 
    \end{equation} 
Similarly as \eqref{defF}, let $F_{xy}(p):=\frac12d^2(p,y)-\frac12d^2(p,x)$. Then \eqref{2bla} implies
    \begin{equation}\label{Cmono}
        0\leq F_{xy}(z_i)-F_{xy}(z_j) \leq 2\pi^2. 
    \end{equation} 
Denote $\|x-y\|$ the chordal distance in the ambient $\mathbb{R}^{p+1}$, which is equivalent to the geodesic distance in the sense that
    \begin{equation}\label{eqdd}
        \frac{2}{\pi}d(x,y) \leq \|x-y\| \leq d(x,y) =: d.
    \end{equation}
In particular, in small-distance asymptotics (that is, in our applications when $z_i$ is sufficiently close to $z_j$), we have $\|x-y\|=d-\frac{1}{24}d^3+O(d^5)$.
Replacing $d$ by the chordal distance in \eqref{2bla}, we have
        \begin{equation}\label{eqbili}
        F_{xy}(z_i) - F_{xy}(z_j) \asymp
        C \left\langle\!\left\langle z_i-z_j, x-y \right\rangle\!\right\rangle,
    \end{equation}
where $\left\langle\!\left\langle \cdot, \cdot \right\rangle\!\right\rangle$ is the inner product in $\mathbb{R}^{p+1}$, and the constant $C$ depends on the uniform constant $\delta$ in Lemma \ref{stayaway}.

\vskip5pt
Given $(i,j)$ and $(x,y)\in\mathbb{S}^p\times\mathbb{S}^p$, denote $\chi_{ij}(x,y)=1$ if \eqref{Cmono} holds and
\[\textrm{Lag}_{z_i}(\psi^t)\cap \gamma_{xy} \neq\emptyset \text{ and }\textrm{Lag}_{z_j}(\psi^t) \cap \gamma_{xy}\neq \emptyset,\]
where $\gamma_{xy}:[0,1]\to\mathbb{S}^p$, with $\gamma(0)=x$, $\gamma(1)=y$, is a geodesic segment from $x$ to $y$; 
while $\chi_{ij}(x,y)=0$ otherwise. 
Let $Q\in SO(p+1)$ such that $Q(x)=\hat x$. 
Analogously as in [Mérigot et al. (2020)], denote
    \[ c_{ij,w} = \frac{1}{d(\hat x, w)d_{ij}}\left|F_{\hat xw}(Q(z_i)) - F_{\hat xw}(Q(z_j))\right|, \]
where $d_{ij}=d(z_i,z_j)$. From \eqref{eqdd} and \eqref{eqbili}, one can see that $c_{ij,w}\asymp 1$. 
Geometrically, it is comparable to the cosine of the angle between $\overrightarrow{\exp_{\hat x}^{-1}(w)}$ and $\overrightarrow{\exp_{\hat x}^{-1}(Q(z_i))\exp_{\hat x}^{-1}(Q(z_j))}$ in the tangent space at $\hat x$ . 

\vskip5pt
By Cauchy-Schwarz's inequality we get
    \begin{align*}
        (u(x)-u(y))^2 &= \left(\sum_{i\neq j}(v_i-v_j)\chi_{ij}(x,y)\right)^2 \\
            &\leq \sum_{i\neq j}\frac{(v_i-v_j)^2}{d_{ij}c_{ij,Q(y)}}\chi_{ij}(x,y) \sum_{i\neq j}d_{ij}c_{ij,Q(y)}\chi_{ij}(x,y).
    \end{align*}
Note that when $\chi_{ij}(x,y)=1$, by \eqref{Cmono} we have
    \begin{align*}
        \sum_{i\neq j}d_{ij}c_{ij,Q(y)}\chi_{ij}(x,y) &\leq C \sum_{i\neq j} |F_{\hat xQ(y)}(Q(z_i)) - F_{\hat xQ(y)}(Q(z_j))| \\
            &\leq C \sum_{i\neq j} |F_{xy}(z_i)-F_{xy}(z_j)| \leq 2C\pi^2. 
    \end{align*}
Therefore, since $Q\in SO(p+1)$, we obtain
    \begin{align*}
        LHS &\leq C\pi^2\int_{\mathbb{S}^p\times\mathbb{S}^p} \sum_{i\neq j}\frac{(v_i-v_j)^2}{d_{ij}c_{ij,Q(y)}}\chi_{ij}(x,y) \intd x \intd y \\
            &= C\pi^2 \int_{\mathbb{S}^p} \sum_{i\neq j}\frac{(v_i-v_j)^2}{d_{ij}c_{ij,w}}\left(\int_{\mathbb{S}^p}\chi_{ij}(x,Q^{-1}(w))\intd x\right)dw.
    \end{align*}
Denoting $m_{ij}=\mathcal{A}^{p-1}\left(\textrm{Lag}_{z_i}(\psi^t)\cap \textrm{Lag}_{z_j}(\psi^t)\right)$, where $\mathcal{A}^{p-1}$ denotes the $(p-1)$-dimensional measure induced by the round metric on $\mathbb{S}^p$, we get
    \[ \int_{\mathbb{S}^p}\chi_{ij}(x,Q^{-1}(w))dx \leq Cd(\hat x, w)m_{ij}c_{ij,w} \leq C\pi m_{ij}c_{ij,w}, \]
and thus
    \begin{equation}\label{LHS}
        LHS \leq C(p)\pi^{p+3} \sum_{i\neq j}\frac{m_{ij}}{d_{ij}}(v_i-v_j)^2.
    \end{equation}

\vskip5pt
On the right hand side of \eqref{DisPW}, by the $C^1$ regularity of optimal maps \cite{loeper2011regularity}, from \cite{Kitagawa2016} one has the matrix
    \[ DG(\psi^t)=\{\partial_jG_i\}=\left\{\frac{\partial G_i}{\partial\psi^t_j}\right\} = \left\{\frac{m_{ij}}{d_{ij}}\right\}. \]
By measure-preserving, one can see that $\partial_iG_i=-\sum_{j\neq i}\partial_jG_i$. Hence, by the same computation as in [Mérigot et al. (2020)] we have
    \begin{align}\label{RHS}
        RHS &= -\sum_{i,j}\partial_jG_iv_iv_j \nonumber\\
            &= \sum_{j\neq i}\partial_jG_i v_i(v_j-v_i) \\
            &= \sum_{i\neq j}\frac{m_{ij}}{d_{ij}}v_i(v_j-v_i).  \nonumber
    \end{align}
Note that 
    \[\sum_{i\neq j}\frac{m_{ij}}{d_{ij}}(v_i-v_j)^2 = \sum_{i\neq j}\frac{m_{ij}}{d_{ij}}v_i(v_i-v_j) - \sum_{i\neq j}\frac{m_{ij}}{d_{ij}}v_j(v_i-v_j) = -2RHS.\]    
Combining with \eqref{LHS} we can obtain the desired inequality \eqref{DisPW}.

\vspace{5pt}

Hence, for any $\theta\in[0,1]$
    \begin{align}\label{nb2}
        \left|\langle v | G(\psi^1) - G(\psi^0)\rangle\right| &= \int_0^1 -\langle DG(\psi^t)v|v\rangle\,\intd t \nonumber\\
            &\ge \int_0^\theta -\langle DG(\psi^t)v|v\rangle\,\intd t \nonumber\\
            &\gtrsim \int_0^\theta \langle v^2 | G(\psi^t)\rangle - \langle v| G(\psi^t)\rangle^2\,\intd t.
    \end{align}

At the moment, we \emph{claim} that for $t\in[0,\frac{1}{4}]$
    \begin{equation}\label{nbclaim1}
        G_i(\psi^t) \gtrsim G_i(\psi^0)\quad\forall\,i=1,\cdots,N.
    \end{equation}
Then the first term of \eqref{nb2} is bounded by
    \begin{equation}\label{nb4}
        \int_0^\theta \langle v^2 | G(\psi^t)\rangle \,\intd t \gtrsim \int_0^\theta \langle v^2 | G(\psi^0)\rangle \,\intd t = \theta \langle v^2 | G(\psi^0)\rangle.
    \end{equation}
For the second term of \eqref{nb2}, using the assumption $\langle v|G(\psi^0)\rangle=0$ we have
    \begin{align}\label{nb5}
        |\langle v|G(\psi^t)\rangle| &= |\langle v|G(\psi^t)-G(\psi^0)\rangle| \nonumber\\
            &\le \|v\|_\infty \|G(\psi^t)-G(\psi^0)\|_1 \nonumber\\
            &\lesssim \|G(\psi^t)-G(\psi^0)\|_1.
    \end{align}
By Lemma 5.5 of \cite{Bansil2022} which can be extended to Riemannian manifold, we have
\begin{eqnarray*}
\|G(\boldsymbol{\psi}^t) - G(\boldsymbol{\psi}^0)\|_1 &\le& C_1 N^2 \|\boldsymbol{\psi}^t - \boldsymbol{\psi}^0 \|_{\infty} \\
    & =& C_1 N^2 t \|\boldsymbol{\psi}^1 - \boldsymbol{\psi}^0 \|_{\infty} \\
    & \le & C_2 N^2 t
\end{eqnarray*}
for some constants $C_1, C_2 > 0$. 
Together with \eqref{nb5} we then obtain
    \begin{equation*}
        |\langle v|G(\psi^t)\rangle| \lesssim N^2t,
    \end{equation*}
and thus
    \begin{equation}\label{nb6}
        \int_0^\theta \langle v|G(\psi^t)\rangle^2\,\intd t \lesssim \int_0^\theta N^4t^2\,\intd t \lesssim N^4\theta^3. 
    \end{equation}

Now, combining \eqref{nb1}, \eqref{nb2}, \eqref{nb4} and \eqref{nb6}, we obtain
$$\theta\langle v^2 | G(\psi^0)\rangle - N^4\theta^3 \lesssim W_1(\nu^0,\nu^1),$$
namely
$$\langle v^2 | G(\psi^0)\rangle \lesssim \frac{1}{\theta}W_1(\nu^0,\nu^1) + N^4\theta^2. $$

If $W_1(\nu^0, \nu^1)$ is small enough, by choosing $\theta^3=\frac{W_1(\nu^0, \nu^1)}{N^4}$ we obtain the desired estimate.
If  $W_1(\nu^0, \nu^1)$, the estimate trivially holds by choosing, for example $\theta=1/4$.  
\end{proof}

\vskip5pt

Last, we prove the claim \eqref{nbclaim1} that was used in the above proof.
\begin{lemma}
\label{lemma_1}
Given $\boldsymbol{\psi}^0, \boldsymbol{\psi}^1 \in S_{+}$, let $\boldsymbol{\psi}^t = (1-t) \boldsymbol{\psi}^0 + t \boldsymbol{\psi}^1$ for $t \in [0,1]$. Then,
$$G_i\left(\boldsymbol{\psi}^t\right)^{\frac{1}{p}} \ge c\left((1-t) G_i\left(\boldsymbol{\psi}^0\right)^{\frac{1}{p}}+t G_i\left(\boldsymbol{\psi}^1\right)^{\frac{1}{p}}\right)  \quad \forall\, i=1,\cdots,N,$$
where $0<c<1$ is some constant.
In particular, we have
$$ G_i(\boldsymbol{\psi}^t) \gtrsim G_i(\boldsymbol{\psi}^0)\quad\forall\,t\in\left[0,\frac14\right].$$
\end{lemma}

\begin{proof} 
We write with slight abuse of notation $\psi^0 = (\psi^0)^{cc}$ and $\psi^1 = (\psi^1)^{cc}$, so that $\psi^0$ and $\psi^1$ are now functions on $\mathbb{S}^{p}$.
From definition \eqref{Lagu}, for a fixed $1\le i\le N$, if $x_0 \in \textrm{Lag}_{z_i}(\psi^0)$ one has 
$$\psi^0(z_j) \le \psi^0(z_i) + c(z_j, x_0) - c(z_i, x_0) \quad\forall\, j=1,\cdots,N, $$
and if $x_1 \in \textrm{Lag}_{z_i}(\psi^1)$ one has
$$ \psi^1(z_j) \le \psi^1(z_i) + c(z_j, x_1) - c(z_i, x_1) \quad\forall\, j=1,\cdots,N. $$
These imply that
$$ D_z c(z_i, x_0) \in \partial \psi^0(z_i) $$
and
$$ D_z c(z_i, x_1) \in \partial \psi^1(z_i) ,$$
where $\partial\psi^k(z_i)$ is the subgradient of $\psi^k$ at $z_i$. Since
$$ \partial \psi^t(z_i) = (1-t) \partial \psi^0(z_i) + t \partial \psi^1(z_i) .$$
we have that
$$ (1-t) D_z c(z_i, x_0) + t D_z c(z_i, x_1) \in \partial \psi^t(z_i) .$$

Recall that the cost function $c(x,z) = d^2(x,z)/2$, we have
$$ D_z c(z, \cdot) = \exp_z^{-1}(\cdot) .$$
It follows that
$$ (1-t) \exp_{z_i}^{-1}(\textrm{Lag}_{z_i}(\psi^0)) + t \exp_{z_i}^{-1}(\textrm{Lag}_{z_i}(\psi^1)) \subset \exp_{z_i}^{-1}(\textrm{Lag}_{z_i}(\psi^t)) .$$
The Laguerre cell $\textrm{Lag}_{z_i}(\psi^k)$ is the contact set of the $c$-convex function $\phi^k$ with its $c$-support function.
From \cite{loeper2011regularity}, we know that $\textrm{Lag}_{z_i}(\psi^k)$ is $c$-convex with respect to $z_i$, namely $\exp_{z_i}^{-1}(\textrm{Lag}_{z_i}(\psi^k))$ is a convex set in $\mathbb{R}^{p}$. 
Let $|\cdot|$ denote the Lebesgue measure on $\mathbb{R}^{p}$. It follows from Brunn-Minkowski inequality that
\begin{equation}\label{BM}
|\exp_{z_i}^{-1}(\textrm{Lag}_{z_i}(\psi^t))|^{1/p} \ge (1-t) | \exp_{z_i}^{-1}(\textrm{Lag}_{z_i}(\psi^0))|^{1/p} + t | \exp_{z_i}^{-1}(\textrm{Lag}_{z_i}(\psi^1))|^{1/p} . 
\end{equation}

Under the polar coordinates $(r, \theta)$ around $y_i$, the spherical measure is given by $\intd\mu=\sin ^{p-1} r \intd r \intd \theta$ and the Euclidean measure is given by $\intd x=r^{p-1} \intd r \intd \theta$. Clearly $\intd\mu \le \intd x$. By Lemma \ref{stayaway}, $\exp _{z_i}^{-1}\left(\textrm{Lag}_{z_i}\left(\psi^k\right)\right)$ lies in $\{0 \leq r \leq \pi-\delta\}$, and so $\intd \mu \geq C \intd x$ when restricted to $\exp _{z_i}^{-1}\left(\textrm{Lag}_{z_i}\left(\psi^k\right)\right)$, where $0<C<1$ depends only on $p$ and $\delta$. From this and \eqref{BM}, we obtain
$$ G_i(\boldsymbol{\psi}^t) ^{\frac{1}{p}}\ge c \bigg( (1-t) G_i(\boldsymbol{\psi}^0)^{\frac{1}{p}} + t G_i( \boldsymbol{\psi}^1) ^{\frac{1}{p}}\bigg) ,$$
 for some positive constant $c<1$.

\end{proof}
\end{proof}

\bibliographystyle{abbrvnat}
\bibliography{reference}

\end{document}